\documentclass[12pt]{amsart} 
\usepackage{amsmath,amssymb}
\usepackage{amsthm,bbm,upgreek}
\usepackage{physics}
\usepackage[toc,page]{appendix}
\usepackage{fullpage}
\usepackage{hyperref}
\usepackage{appendix}
\usepackage{xcolor}
\usepackage{t1enc}
\usepackage[shortlabels]{enumitem}
\usepackage{eucal}
\usepackage{mathrsfs}
  
\numberwithin{equation}{section}

\theoremstyle{plain}
\newtheorem{theorem}{Theorem}[section]
\newtheorem{lemma}[theorem]{Lemma}

\newtheorem{proposition}[theorem]{Proposition}
\newtheorem{corollary}[theorem]{Corollary}

\theoremstyle{definition}
\newtheorem{definition}[theorem]{Definition}
\newtheorem{remark}[theorem]{Remark}

\def\mathbbZ{\mathbb{Z}}

\def\sfN{\mathsf{N}}

\def\Re{\mathrm{Re}}

\DeclareMathOperator{\sym}{sym}

\newcommand{\bb}[1]{\mathbb{#1}}
\newcommand{\fk}[1]{\mathfrak{#1}}
\newcommand{\cl}[1]{\mathcal{#1}}
\renewcommand{\sf}[1]{\mathsf{#1}}
\renewcommand{\rm}[1]{\mathrm{#1}}

\title{Probabilistic representation and limit theorems for particle numbers of quasi-free states}

\author{Fanch Coudreuse}
\address{Bocconi University, Department of Decision Sciences and BIDSA, Milan, (Italy)}
\email{coudreuse@math.univ-lyon1.fr}

\author{Simone Rademacher} 
\address{Department of Mathematics, University of Mannheim, B6, 68159 Mannheim, Germany}
\email{simone.rademacher@uni-mannheim.de}

\author{Oliver Tse}
\address{Department of Mathematics and Computer Science, Eindhoven University of Technology, The Netherlands} 
\email{o.t.c.tse@tue.nl}


\begin{document}

\begin{abstract}
  We study the particle number distribution of locally interacting bosonic quasi-free states arising in various areas of mathematical physics. We show that the particle number decomposes into an infinite sum of independent geometrically distributed random variables, confirming predictions from the physics literature. This representation yields exponential tail bounds for the particle number, a law of small numbers together with a large deviation principle at logarithmic speed when the lattice spacing diverges, and a central limit theorem when it vanishes. As an application, we obtain a detailed description of the statistics of the quantum depletion in Bose--Einstein condensates, with the constants explicit in terms of the scattering length of the interaction potential.
\end{abstract}

\maketitle


\section{Introduction}
\label{sec:intro}

Quasi-free states play a central role in many areas of physics \cite{Araki,BratelliRobinson,Haag}, including quantum field theory \cite{Wald}, many-body quantum mechanics \cite{LiebSeiringer}, and quantum information theory \cite{Weedbrook}. They generalize the notion of free (non-interacting) states to certain interacting systems whose correlation functions are still completely determined by their two-point functions. Recent advances in both theoretical and experimental methods have increased interest in the full particle-number statistics of quasi-free states. Particle number statistics play a central role in the description of quantum fluctuations in Bose--Einstein condensates in many-body quantum mechanics \cite{Bogoliubov,Idziaszek,Rademacher} and in bosonic sampling problems based on Gaussian states in quantum information theory \cite{Bravyi,Quesada,Weedbrook}.

A central challenge in analyzing quantum many-body systems is the quantitative description of fluctuations. While expectation values and low-order moments are often accessible, the full probability distribution of physically relevant observables remains largely out of reach. This is particularly true for particle-number statistics in bosonic systems, where fluctuations encode key physical phenomena, such as \emph{quantum depletion} in Bose--Einstein condensates \cite{Kristensen,Vibel} and output statistics in bosonic sampling experiments \cite{KocharovskyBosonSampling,KocharovskyCNBIII,Weedbrook}.

In this work, we focus on the distribution of the particle-number observable with respect to a class of locally interacting bosonic quasi-free states (cf.\ Definition~\ref{def:bogostate} below). These states play a distinguished role across quantum field theory, many-body physics, and quantum information, as they provide effective descriptions of weakly interacting systems while retaining non-trivial correlations \cite{Bogoliubov,LiebSeiringer,BBCS2018}. Despite their structural simplicity---entirely characterized by two-point correlation functions---their particle-number distribution is non-trivial. A key difficulty is that the particle-number observable involves contributions from infinitely many modes, leading to expressions that are infinite sums of dependent quantum observables. As a consequence, even basic quantities such as moment generating functions and tail probabilities are difficult to analyze directly.

In the physics literature, it has been suggested that the particle number in quasi-free states may be described in terms of independent geometric random variables, reflecting an underlying mode-wise structure. Although this viewpoint is consistent with heuristic arguments and known structural properties of quasi-free states, a comprehensive, fully rigorous probabilistic formulation remains unavailable, to our knowledge. Even more fundamentally, while some aspects of rare-event behavior are understood, a comprehensive theory in this setting is still lacking. In particular: 
\begin{enumerate}
	\item What is the probability that the particle number deviates significantly from its mean?
	\item Do these deviations obey a large deviation principle?
	\item How do these properties depend on the scaling regime of the underlying system?
\end{enumerate}
These questions are not only of theoretical interest but also increasingly relevant in light of experimental advances that make high-precision measurements of fluctuations and rare events accessible \cite{Kristensen,Vibel}. 

\subsection{Main contributions}

The present work addresses the above problem by developing a probabilistic description of particle-number statistics in quasi-free states and exploiting this structure to derive quantitative results on fluctuations.

\subsubsection*{Probabilistic representation}

Our first main result (cf.\ Theorem~\ref{thm:N+_bogo}) rigorously confirms predictions from the physics literature by showing that the particle number of a locally interacting quasi-free state can be represented as an infinite sum of independent geometrically distributed random variables (see e.g. \cite{Caves,Gasperini,Tanas}). This result reveals a remarkable structural simplification, namely, although the underlying quantum system exhibits non-trivial correlations, the particle-number observable behaves as a classical sum of independent random variables after an appropriate transformation.

This representation not only identifies the exact distributional structure of the particle number but also provides explicit formulas for expectations and moment-generating functions, and will serve as the starting point for all subsequent analysis.

\subsubsection*{Asymptotic limits and tail behavior}

Building on the probabilistic representation of Theorem~\ref{thm:N+_bogo}, we develop a systematic theory of rare-event fluctuations. We prove that the particle number satisfies deviation results in multiple regimes of the underlying lattice structure $\kappa\bb{Z}^d_0$, where $\bb{Z}^d_0 := \bb{Z}^d \setminus \{ 0 \}$ is the punctured cubic lattice, with lattice spacing parameter $\kappa>0$:
\begin{enumerate}
	\item In the \emph{intermediate regime} (fixed lattice spacing, $\kappa>0$), we derive tail bounds for the particle number distribution (Theorem~\ref{thm:ldp-intermediate}).
	\item In the \emph{sparse regime} (diverging lattice spacing, $\kappa\to+\infty$), we prove a law of small numbers, with an explicit Poissonian first-order asymptotic for the particle number and its exponential moments (Theorem~\ref{thm:poisson-sparse}), together with a large deviation principle at logarithmic speed (Theorem~\ref{thm:log-ldp}).
	\item In the \emph{dense regime} (vanishing lattice spacing, $\kappa\to 0$), we show convergence toward Gaussian fluctuations via a central limit theorem (Theorem~\ref{thm:central-limit}).
\end{enumerate}

We expect that such results extend to more general families of discrete sets whose rescaled counting measures converge, in an appropriate sense, to measures supported on more general limiting domains. We restrict our attention here to the punctured cubic lattice $\kappa \bb{Z}^3_0$, which provides the simplest non-trivial example that already displays the two main difficulties: possible singular behavior near a lower-dimensional excluded set, here reduced to the point $0$, and non-trivial behavior at infinity.

\subsubsection*{Application to Bose--Einstein condensates}

As a key application, we analyze the quantum depletion observable in Bose--Einstein condensates in Section~\ref{sec:BEC}, i.e., the number of particles outside the condensate. We show that the quantum depletion admits a representation as an infinite sum of independent geometrically distributed random variables and provide a quantitative analysis of its tail behavior in the large-particle limit. To our knowledge, this provides the first rigorous probabilistic description of rare fluctuations of quantum depletion. 

Our key observation is that the limiting Bogoliubov state can be interpreted as a locally interacting quasi-free state on a quotient of the punctured lattice $\kappa \bb{Z}^3_0$, with lattice spacing determined by the scattering length $\fk{a}$. This allows us to translate the sparse- and dense-lattice asymptotics of the punctured lattice into asymptotic results for the quantum depletion as $\mathfrak{a}$ converges to $0$ or $+\infty$.

%

\section{Main Results}

\subsection{Preliminaries}

Throughout this paper, we consider an at most countable discrete set $\Lambda$
and the Fock space
\[
	\fk{h}_{\Lambda} := \cl{F}_{\sym}(\ell^2_\bb{C}(\Lambda{\times}\{+,-\})), 
\]
where $\cl{F}_{\sym}(\cl{H})$ denotes the symmetric (or bosonic) Fock space over a Hilbert space $\cl{H}$ and $\{+,-\}$ denote spin states. We further define the spin space $\bb{S} = \text{span}_\bb{C}\{\ket{+},\ket{-}\}\cong\bb{C}^2$, where $\ket{+}$ and $\ket{-}$ denote the positive and negative spin vectors respectively. It is not difficult to see that
\[
	\fk{h}_{\Lambda} \cong \bigotimes_{p \in \Lambda} \fk{h}_p,\qquad \fk{h}_p = \cl{F}_{\sym}(\bb{S}).
\]
In this sense, the Fock space $\fk{h}_{\Lambda}$ may be interpreted as a standard quantum lattice model with local spin space $\bb{S}$. 

On each lattice site $p\in\Lambda$, the creation operator $\sf{a}_{(p,\pm)}^*$ creates a boson with positive, respectively, negative spin, while the annihilation operator $\sf{a}_{(p,\pm)}$ annihilates a boson with positive, respectively, negative spin. These operators satisfy the usual canonical commutation relations
\begin{align}\label{eq:CCR}\tag{CCR}
    [\sf{a}_\sigma^*,\sf{a}_\eta^*] = [\sf{a}_\sigma,\sf{a}_\eta] = 0, \;\; \text{and} \;\; [\sf{a}_\sigma, \sf{a}_\eta^*] = \updelta_{\sigma,\eta} \;\; \text{for} \;\; \sigma,\eta \in \Lambda{\times}\{+,-\} \, . 
\end{align}

%

\medskip
We are interested in a particular class of bosonic vectors in the bosonic lattice Fock space $\fk{h}_{\Lambda}$, namely the \textit{locally interacting quasi-free vector}, which we define below. 

\begin{definition}[Locally interacting quasi-free vectors/states]
\label{def:bogostate}
	For any $\nu \in \ell^2_\bb{R}(\Lambda)$, the \emph{locally interacting quasi-free vector} $\Uppsi_\nu$ on $\fk{h}_\Lambda$ is defined as 
	\begin{equation} \label{eq:bogostate}
		\Uppsi_\nu  := \exp(\sf{A}_\nu)\Upomega,\qquad \sf{A}_\nu := \sum_{p \in \Lambda} \nu_p \big[ \sf{a}_{(p,+)}^* \sf{a}_{(p,-)}^* - \sf{a}_{(p,+)} \sf{a}_{(p,-)}\big],
	\end{equation}
	where $\Upomega$ denotes the vacuum vector on $\fk{h}_\Lambda$. In the physics literature, vectors of this form are also known as \emph{squeezed vacuum states}.
	
	The associated \emph{locally interacting quasi-free state} is defined by $\varrho_\nu :=\langle \Uppsi_\nu,\cdot\;\Uppsi_\nu\rangle.$
\end{definition}

\begin{remark}
\begin{enumerate}[(i)]
\item By construction, $\sf{A}_\nu^* =-\sf{A}_\nu$ for any $\nu \in \ell^2_\bb{R}(\Lambda)$. In particular, $\exp(\sf{A}_\nu)$ is a unitary map, and thus a Bogoliubov transformation\footnote{We note that general Bogoliubov transformations take the form $A_{\mu} = \sum_{\sigma,\eta \in \Lambda{\times}\{+,-\}} \mu_{\sigma,\eta} ( \sf{a}_\sigma^*\sf{a}_\eta^* - \sf{a}_\sigma\sf{a}_\eta )$ for any intensity function $\mu \in \ell^2((\Lambda{\times}\{+,-\})^2)$.}. Bogoliubov transformations were introduced in mathematical physics \cite{Bogoliubov,LiebSeiringer} to diagonalize Hamiltonians quadratic in creation and annihilation operators. Thus, in this picture, locally interacting quasi-free states $\varrho_\nu$ can be interpreted as ground states of locally interacting quadratic Hamiltonians and are, in this context, simply called \textit{Bogoliubov states} or \textit{pure Gaussian states}.

\item 
In the literature, however, states of the form \eqref{eq:bogostate} have various names. In fact, the terminology of \textit{locally interacting quasi-free states} that we use in this paper is inspired by the following picture: A state $\varrho$ is called \textit{quasi-free} if it is completely characterized by its two-point correlations
\[
 \varrho( \sf{a}_\sigma^*\sf{a}_\eta ), \;\; \varrho (\sf{a}_\sigma\sf{a}_\eta), \;\;  \text{and} \;\;  \varrho(\sf{a}_\sigma^*\sf{a}_\eta^*), \quad \sigma,\eta \in \bb{Z}_\kappa^d{\times}\{+,-\} \, .  
\]
As a consequence of the canonical commutation relations \eqref{eq:CCR}, the only relevant (non-zero) two-point correlations for a state $\varrho$ in our context are given by $ \varrho(\sf{a}_{(p,+)}^*\sf{a}_{(p,+)})$, $\varrho(\sf{a}_{(p,+)} \sf{a}_{(p,-)})$, and $\varrho(\sf{a}_{(p,+)}^*\sf{a}_{(p,-)}^*)$ for $p\in\Lambda$, i.e., only the modes $(p,+)$ and $(p,-)$ couple. 
 Thus, the state $\varrho$ is a \emph{locally interacting quasi-free state}. In Section~\ref{sec:BEC}, and in the case where $\Lambda = \bb{Z}^d_0 / \sim$, where $p \sim q$ if $p = \pm q$, we explain that $\varrho$ may alternatively be interpreted as a translation-invariant quasi-free state, represented in momentum space.
%

\item The proper construction of the infinite tensor product (over a countable set) and of the infinite sums of operators involved follow from standard argument in Hilbert space theory and many-body particle systems. We refer to standard textbooks on the subject for the technical details. 

\end{enumerate}
\end{remark}

In this work, we investigate the distribution of the particle-number operator with respect to locally interacting quasi-free states $\varrho_\nu$. Here, the sum of bosons with positive and negative spin on each site $p\in\Lambda$ is counted by the number operator
\[
	\sf{N}_p = \sf{a}_{(p,+)}^*\sf{a}_{(p,+)} + \sf{a}_{(p,-)}^* \sf{a}_{(p,-)}\,. 
\]
The total particle number observable is then given by
\[
	\sf{N}_\Lambda := \sum_{p\in \Lambda} \sf{N}_p.
\]

\begin{remark}
	When $\Lambda = \kappa \bb{Z}^d_0$, with $\kappa>0$ interpreted as a lattice spacing, the space $\fk{h}_\Lambda$ embeds naturally into $\cl{F}_{\sym}(\ell^2(\kappa\bb{Z}^d))$. In this case $\sf{N}_\Lambda$, denoted by $\sf{N}_+$ in Section~\ref{sec:BEC}, is called the \emph{quantum depletion} operator in the physics literature, as it counts all bosons outside the lattice site $p=0$. Section~\ref{sec:BEC} is devoted to the analysis of $\sf{N}_+$, which depends implicitly on $\kappa$ through the sum over all sites of $\kappa\bb{Z}^d_0$.
\end{remark}

Throughout the manuscript, we will use the notation
\[
    \bb{P}_{\Uppsi}(\sf{A} \in E) := \langle \Uppsi,\mathbbm{1}_E(\sf{A})\Uppsi\rangle
\]
to denote the probability of measuring the observable $\sf{A}$ in the event $E\in\cl{B}(\bb{R})$ in the vector $\Uppsi$. Similarly, we will define the mean and variance of $\sf{A}$ in the vector state $\Psi$ 
\[
    \bb{E}_{\Psi}[\sf{A}] = \expval{\sf{A}}{\Psi}, \quad \text{and} \quad \bb{V}_\Psi[\sf{A}] := \expval{(\sf{A} - \bb{E}_{\Psi}[\sf{A}])^2}{\Psi},
\]
provided that those quantities are well-defined (e.g., when $A$ is self-adjoint, $\Psi \in \text{dom}(|A|)$). 

\subsection{Probabilistic representation}
Our first result gives a representation of the distribution of the particle-number operator $\sf{N}_\Lambda$ in the state $\varrho_\nu$ as summarized in the following statement.

\begin{theorem}[Number operator distribution]\label{thm:N+_bogo}
        Let $\nu\in\ell^2(\Lambda)$ and $\Uppsi_\nu$ be the corresponding locally interacting quasi-free vector given by \eqref{eq:bogostate}. Then,
    \[
        \bb{P}_{\Uppsi_\nu}(\sf{N}_\Lambda \in E) = \bb{P}( 2 G_\Lambda \in E )\qquad \text{for any Borel set $E\in\cl{B}(\bb{R})$},
    \]
    where $G_\Lambda = \sum_{p \in \Lambda} G_p$ is a sum of independent geometric random variables, with $G_p\sim \mathrm{Geo}(\sech^2(\nu_p))$, $p\in\Lambda$, each with parameter $\sech^2(\nu_p)$, i.e.,
    \[
		\bb{P}(G_p = k) = \sech^2(\nu_p) \tanh^{2k}(\nu_p)\qquad \text{for every $k \in \bb{N}_0$}.
    \]
    
    In other words, the distribution of the number operator $\sf{N}_\Lambda$ in the vector $\Uppsi_\nu$ coincides with the distribution of a sum of geometric random variables $2G_\Lambda$.
    \end{theorem}

As a consequence of Theorem~\ref{thm:N+_bogo}, we obtain explicit expressions not only for the mean and variance of $\sf{N}_\Lambda$ in the vector $\Uppsi_\nu$, but also for its characteristic and moment generating function in terms of the random variable $G_\Lambda$, which serve as indispensable tools in studying the large deviation and tail behavior of the distribution. Indeed, since
\begin{align}\label{eq:mean-variance}
	\mu_\Lambda:=\bb{E}[G_\Lambda] = \sum_{p\in\Lambda} \sinh^2(\nu_p),\qquad \sigma_\Lambda^2 := \bb{V}[G_\Lambda] = \sum_{p\in\Lambda} \cosh^2(\nu_p)\sinh^2(\nu_p),
\end{align}
are the mean and variance for the random variable $G_\Lambda$, we find the expressions
\[
	\bb{E}_{\Uppsi_\nu}[\sf{N}_\Lambda] = 2\mu_\Lambda,\qquad \bb{V}_{\Uppsi_\nu}[\sf{N}_\Lambda] = 4 \sigma_\Lambda^2.
\]

\subsection{Asymptotic limits and tail behavior}\label{subsec:tail}
Next, we prove deviation estimates for the particle-number distribution in three regimes. The first holds for an arbitrary index set $\Lambda$. The remaining two concern the lattice $\kappa \bb{Z}^d_0$ and its spacing $\kappa\in(0,+\infty)$. While it might be possible to study other lattices, we chose to focus on this specific case for two reasons: first, it is a sufficiently simple model that still allows for singular behavior (due to the removal of the origin); and second, it is the relevant example for Bose--Einstein condensation (see Section~\ref{sec:BEC}). 

\subsubsection*{\bf\em Intermediate regime}
\label{sec:intermediate-results}
Since the first regime does not specify a lattice spacing, we define it for an arbitrary at most countable set $\Lambda$ and weights $\nu\in\ell^2(\Lambda)$. Besides the mean $\mu_\Lambda$ of \eqref{eq:mean-variance}, the result involves the largest of the individual means,
\[
	\overline{\mu}_\Lambda := \sup_{p\in\Lambda} \sinh^2(\nu_p),
\]
which is attained, since $\nu\in\ell^2(\Lambda)$ implies that $\nu_p\to 0$.

\begin{theorem} \label{thm:ldp-intermediate}
Let $\nu \in \ell^2(\Lambda)$ with $\nu\not\equiv 0$. Then the distribution of $\sf{N}_\Lambda$ in the quasi-free vector $\Uppsi_\nu$ satisfies the exponential tail asymptotics
\[
	\lim_{\lambda \to \infty} \frac{1}{\lambda} \ln \mathbb{P}_{\Uppsi_\nu} \big( \sf{N}_\Lambda >  2\lambda \mu_\Lambda \big) = -\cl{I}_\Lambda,\qquad \cl{I}_\Lambda := \mu_\Lambda \ln \Bigl( 1 + \frac{1}{\overline{\mu}_\Lambda} \Bigr).
\]
\end{theorem}

Since $\bb{E}_{\Uppsi_\nu}[\sf{N}_\Lambda] = 2\mu_\Lambda$, the event above is that $\sf{N}_\Lambda$ exceeds $\lambda$ times its mean. Theorem~\ref{thm:ldp-intermediate} thus states that the upper tail decays exponentially to leading order,
\[
	\mathbb{P}_{\Uppsi_\nu} \big( \sf{N}_\Lambda >  2\lambda \mu_\Lambda \big) = \exp\bigl(-\lambda\,(\cl{I}_\Lambda + o(1))\bigr)\qquad \text{as $\lambda\to +\infty$},
\]
with a decay rate governed by the geometric random variable of largest mean. 

The proof of Theorem \ref{thm:ldp-intermediate} in Section~\ref{sec:proof-intermediate} in fact yields two-sided bounds at every fixed $\lambda \geq 1$ (with the lower bound being also valid for $\lambda \in (0,1)$), given in \eqref{eq:tail-lower-bound} and \eqref{eq:tail-upper-bound}, namely
\[
	\bigl(\lfloor \lambda\mu_\Lambda\rfloor + 1\bigr)\ln\biggl(\frac{\overline{\mu}_\Lambda}{1+\overline{\mu}_\Lambda}\biggr)
	\le \ln \bb{P}_{\Uppsi_\nu}\big( \sf{N}_\Lambda > 2\lambda\mu_\Lambda \big)
	\le -\lambda\mu_\Lambda \ln\biggl(\frac{\lambda(1+\overline{\mu}_\Lambda)}{1+\lambda\overline{\mu}_\Lambda}\biggr) - \frac{\mu_\Lambda}{\overline{\mu}_\Lambda}\ln\biggl(\frac{1+\overline{\mu}_\Lambda}{1+\lambda\overline{\mu}_\Lambda}\biggr).
\]
These hold for any infinite sum of independent geometric random variables and, as such, extend the tail bounds of \cite{Janson} for finite sums.
%

\subsubsection*{\bf\em Sparse regime, $\kappa\to +\infty$}
In the sparse regime, we study rare events and we establish a large-deviation principle for the family of distributions of the number operator $\sf{N}^\kappa := \sf{N}_{\kappa \bb{Z}^d_0}$ in the quasi-free vector $\Uppsi_\nu^\kappa$ as the lattice spacing diverges, i.e., $\kappa \rightarrow +\infty$. Setting $G^\kappa := G_{\kappa \bb{Z}^d_0}$ (associated with $\nu^\kappa \in \ell^2(\kappa \bb{Z}^d_0)$), this is equivalent to studying a scaled limit of the log-Laplace transform
\begin{align*}
	\mathscr{L}_\kappa(t) 
	&:=\ln \bb{E}_{\Uppsi_\nu^\kappa}[e^{t\sf{N}_+^\kappa}] \\
	&\phantom{:}= \ln \bb{E}[e^{2tG_\Lambda^\kappa}] = \begin{cases}
		\displaystyle -\sum_{p\in \Lambda^\kappa} \ln\left(1 - \mu_p^\kappa(e^{2t}-1)\right) &\text{for $t< \tfrac{1}{2}\ln\bigl(1+1/\overline{\mu}^\kappa\bigr)$},\\
		+\infty & \text{otherwise},
	\end{cases}
\end{align*}
with the local means $\mu_p^\kappa = \sinh^2(\nu_p^\kappa)$ for the geometric random variables $G_p^\kappa$ for $p \in \kappa \bb{Z}^d_0$.

For this result, we consider weights $\nu_p^\kappa = \nu(p)$, $p\in \kappa \bb{Z}^d_0$, where $\nu:\bb{R}^d_0 \to \bb{R}$ is continuous and $\bb{R}^d_0 := \bb{R}^d \setminus \{0 \}$. We also assume that $\nu$ satisfies the tail condition
\begin{align}\label{eq:sparse-condition}\tag{$\sf{A}\nu$}
	\fk{c}_\nu:=\lim_{|x|\to \infty} |x|^\gamma \sinh^2(\nu(x)) \in (0,+\infty),\qquad \gamma>d.
\end{align}
Under this condition, the scaled mean satisfies
\[
	\kappa^\gamma\, \bb{E}_{\Uppsi_\nu^\kappa}[\sf{N}^\kappa] \longrightarrow 2\fk{c}_\nu \zeta_d(\gamma)\qquad\text{for $\kappa\to +\infty$},
\]
where
\[
	\zeta_d(\gamma) := \sum_{p\in\bb{Z}^d_0} \frac{1}{|p|^{\gamma}}
\]
is an Epstein Zeta function in dimension $d$, well-defined for all $\gamma > d$.

\medskip
The scaled log-Laplace transform in fact converges, which identifies the limiting law of $\sf{N}^\kappa$ in the sparse regime as a \emph{law of small numbers}.

\begin{theorem}[Law of small numbers in the sparse regime]\label{thm:poisson-sparse}
	Let $\nu:\bb{R}^d_0 \to \bb{R}$ be a continuous function satisfying condition \eqref{eq:sparse-condition} for some $\fk{c}_\nu>0$ and exponent $\gamma>d$, and set $\fk{m}_\nu := \fk{c}_\nu\,\zeta_d(\gamma)$. Then, as $\kappa\to+\infty$:
	\begin{enumerate}[(i)]
		\item \emph{(Scaled log-Laplace transform)} For every $t\in\bb{R}$,
		\begin{equation}\label{eq:sparse-laplace-limit}
			\lim_{\kappa\to+\infty}\kappa^\gamma \mathscr{L}_\kappa(t) = \fk{m}_\nu\bigl(e^{2t}-1\bigr).
		\end{equation}
		\item \emph{(Two-point law)} For every $k\in\bb{N}_0$,
		\begin{equation}\label{eq:sparse-two-point}
			\bb{P}_{\Uppsi_\nu^\kappa}\bigl[\sf{N}^\kappa = 2k\bigr] = o\bigl(\kappa^{-\gamma}\bigr) + \begin{cases}
				1- \fk{m}_\nu\,\kappa^{-\gamma} & \text{for $k=0$}, \\
				\fk{m}_\nu\,\kappa^{-\gamma} & \text{for $k=1$}, \\
				0 & \text{otherwise}.
			\end{cases}
		\end{equation}
		\item \emph{(Exponential moments)} For every $\theta>0$ and $\lambda\in\bb{C}$,
		\begin{equation}\label{eq:sparse-moments}
			\lim_{\kappa \to +\infty} \kappa^\gamma\, \bb{E}_{\Uppsi_\nu^\kappa}\bigl[\bigl(\sf{N}^\kappa \bigr)^{\theta} e^{\lambda \sf{N}^\kappa}\bigr] = \fk{m}_\nu\, 2^{\theta} e^{2\lambda}.
		\end{equation}
	\end{enumerate}
\end{theorem}

Since $t\mapsto \fk{m}_\nu(e^{2t}-1)$ is precisely the log-Laplace transform of $2\,\mathrm{Poi}(\fk{m}_\nu)$, part~(i) states that, to leading order in $\kappa^{-\gamma}$,
\[
	N^\kappa \;\approx\; 2\,\mathrm{Poi}\bigl(\fk{m}_\nu\,\kappa^{-\gamma}\bigr) \qquad\text{as $\kappa\to+\infty$,}
\]
that is, the particle number obeys a law of small numbers with vanishing intensity, and parts~(ii) and~(iii) make this leading order explicit. We emphasize that the normalization $\kappa^{-\gamma}$ in \eqref{eq:sparse-laplace-limit} \emph{vanishes} as $\kappa\to+\infty$, so that \eqref{eq:sparse-laplace-limit} does not by itself express a large deviation principle: deviations in the sparse regime are polynomially, not exponentially, rare.

\medskip
Estimate \eqref{eq:sparse-two-point} identifies the leading correction to $\sf{N}^\kappa = 0$, but bounds the probability of observing $k\ge 2$ pairs only by $o(\kappa^{-\gamma})$. The correct order for every $k$ is captured by a large deviation principle at logarithmic speed.

\begin{theorem}[LDP at logarithmic speed]\label{thm:log-ldp}
Let $\nu:\bb{R}^d\backslash\{0\}\to\bb{R}$ be continuous and satisfy \eqref{eq:sparse-condition} with $\fk{c}_\nu\in(0,+\infty)$ for some $\gamma>d$. Then:
\begin{enumerate}[(i)]
	\item \emph{(Local asymptotics)} For every $k\in\bb{N}_0$,
	\[
		\lim_{\kappa\to+\infty} \frac{1}{\ln\kappa}\ln \bb{P}\big(G^\kappa = k\big) = -\gamma k .
	\]
    More precisely, we have, uniformly in $k \in \bb{N}_0$,
    \[
        \ln \bb{P}(G^\kappa = k) = -\gamma k \ln(\kappa) + (k+1) O(1).
    \]
    
	\item \emph{(Logarithmic LDP)} The family $\{G^\kappa\}_{\kappa}$ satisfies, as $\kappa\to+\infty$, a large deviation principle on $\bb{R}$ with speed $\ln\kappa$ and good rate function
	\[
		\mathscr{J}_\gamma \!(x) := \begin{cases} \gamma x & \text{for $ x\in\bb{N}_0$},\\ +\infty & \text{otherwise},\end{cases}
	\]
	i.e., for every Borel set $E\in\cl{B}(\bb{R})$,
	\[
		-\inf_{x\in \overset{\circ}{E}} \mathscr{J} \!(x)
		\le \liminf_{\kappa\to+\infty} \frac{\ln \bb{P}(G^\kappa\in E)}{\ln\kappa}
		\le \limsup_{\kappa\to+\infty} \frac{\ln \bb{P}(G^\kappa\in E)}{\ln\kappa}
		\le -\inf_{x\in \overline{E}} \mathscr{J}_\gamma \!(x).
	\]
	\item Consequently, the family of number operators $\sf{N}^\kappa$ in the quasi-free vectors $\Uppsi_\nu^\kappa$ satisfies a large deviation principle with speed $\ln\kappa$ and good rate function
	\[
		\mathscr{J}_{\sf{N}} : x \to \mathscr{J}_{\gamma}(x/2).
	\]
	In particular, the local asymptotics satisfy 
	\[
		\ln\bb{P}_{\Uppsi_\nu^\kappa}(\sf{N}^\kappa = 2k) = -\gamma k \ln \kappa + k O(1),
	\]
    uniformly in $k \in \bb{N}_0$. 
\end{enumerate}
\end{theorem}

\subsubsection*{\bf\em Dense regime, $\kappa\to 0$} In this regime, we assume that the weights take the form $\nu_p^\kappa=\nu(p)$, where $\nu:\bb{R}^d_0 \to \bb{R}$ is a continuous function satisfying
\begin{align}\label{ass:dense-conditions}\tag{$\sf{B}\nu$}
	\left\{\quad\begin{aligned}
			\limsup_{|x|\to +\infty} |x|^\gamma\sinh^2(\nu(x))<+\infty &\qquad \text{for some $\gamma>d$}, \\
		\limsup_{|x|\to 0} |x|^\beta\sinh^2(\nu(x))<+\infty &\qquad \text{for some $\beta\in(0,d/2)$}.
	\end{aligned}\right.
\end{align}

The first condition provides tail estimates, while the second controls integrability around zero upon rescaling. Under these two conditions, we obtain the following central limit theorem for the sequence of rescaled observables as $\kappa\to 0$, where we set $\mu^\kappa := \bb{E}[G^\kappa]$ and $(\sigma^\kappa)^2 := \bb{V}[G^\kappa]$ (in particular $\bb{E}_{\Psi_\nu^\kappa}[N^\kappa] = 2 \mu^\kappa$ and $\bb{V}_{\Psi_\nu^\kappa}[N^\kappa] = 4 (\sigma^\kappa)^2$). 

\begin{theorem}\label{thm:central-limit}
    Let $\nu:\bb{R}^d\backslash\{0\}\to \bb{R}$ be a continuous function satisfying condition \eqref{ass:dense-conditions} and consider the sequence of scaled centered random variables
    \[
    	X^\kappa := (G^\kappa - \mu^\kappa)/\sigma^\kappa.
    \]
    Then the sequence of distributions $\rm{Law}(X^\kappa)$ converges weakly to the standard normal distribution $\cl{N}(0,1)$ along with all its $\theta$-moments for $\theta\in[1,d/\beta)$.
        
    As a result, the family of scaled centered observables $(\sf{N}^\kappa-2\mu^\kappa)/(2\sigma^\kappa)$ in the quasi-free vector $\Uppsi_\nu^\kappa$ converges in distribution to a standard normal random variable.
\end{theorem}

A formal argument for the result is easily obtained by examining the asymptotic behavior of the characteristic function. Indeed, let $\tau_p^2:= \tanh^2(\nu(p))\in(0,1)$. Then, the characteristic function of the random variable $G^\kappa$ takes the form
\begin{align*}
    \mathbb{E}\big[ e^{it G^\kappa} \big] &= \prod_{p \in \Lambda^\kappa} \frac{1- \tau_p^2}{1-e^{it}\tau_p^2} 
    = \exp\bigl(i \vartheta(t) - \Upsilon(t)\bigr),
\end{align*}
with the functions
\begin{align*}
    \vartheta(t) &= \sum_{p \in \Lambda^\kappa} \arctan\left(\frac{\sin(t)\tau_p^2}{1-\cos(t)\tau_p^2}\right) = \sum_{p \in \Lambda^\kappa} \sum_{n\ge 1} \tau_p^{2n}\frac{\sin(nt)}{n},\\
    \Upsilon(t) &= \frac{1}{2}\sum_{p \in \Lambda^\kappa} \ln\left(1 + \frac{4\sin^2(t/2)\tau_p^2}{(1- \tau_p^2)^2}\right) = 2\sum_{p \in \Lambda^\kappa} \sum_{n\ge 1} \tau_p^{2n}\frac{\sin^2(nt/2)}{n},
\end{align*}
where $\Upsilon$ is related to the logarithm of the Poisson kernel and $\vartheta$ is related to its phase.

In the asymptotic regime $t\to 0$, we obtain
\begin{align*}
    \vartheta(t) &= t \sum_{p \in \Lambda^\kappa} \frac{\tau_p^2}{1-\tau_p^2} + O(t^3) = \mu_\Lambda^\kappa t + O(t^3),\\
    \Upsilon(t) &= \frac{1}{2}t^2 \sum_{p \in \Lambda^\kappa} \frac{\tau_p^2}{(1-\tau_p^2)^2} + O(t^4) = \frac{1}{2}(\sigma_\Lambda^\kappa)^2 t^2 + O(t^4).
\end{align*}
Therefore, when considering random variables $X^\kappa = (G^\kappa - \mu^\kappa)/\sigma^\kappa$, we obtain
\[
	\mathbb{E}\big[ e^{it X^\kappa} \big] \longrightarrow e^{-\frac{1}{2}t^2},
\]
where the right-hand side is the characteristic function of the standard normal distribution. 

While the characteristic function approach is intuitive, making the computations above rigorous requires additional effort. In Section~\ref{sec:dense-regime}, we take a different and more elementary approach to proving Theorem~\ref{thm:central-limit} that is based on a generalization of Lindeberg's Central Limit Theorem.

\medskip
As a by-product of Theorem~\ref{thm:central-limit}, we obtain the law of large numbers.

\begin{corollary}[Law of large numbers]
Let the family $\{\kappa^d G^\kappa\}_{\kappa>0}$ of random variables be defined on a common probability space $(\Omega,\cl{F},\bb{P})$. Then,
\[
	\kappa^d G^\kappa \longrightarrow \mu^*:=\int_{\bb{R}^d_0} \sinh^2(\nu(x)) \dd{x} \qquad\text{in $L^\theta(\Omega)$\;\; for any $\theta\in[1,d/\beta)$.}
\]
As a consequence, the sequence of rescaled number operators $\kappa^d N^\kappa$ in $\Psi^\kappa_\nu$ satisfies 
\[
    \kappa^d N^\kappa \to 2 \mu^* = 2 \int_{\bb{R}^d_0} \sinh^2(\nu(x)) \dd{x} \qquad\text{in $L^\theta(\Omega)$\;\; for any $\theta\in[1,d/\beta)$.}
\]
Here we abuse notation by writing $\kappa^d N^\kappa$ for a realization as random variable of the law of $\kappa^d N^\kappa$ under $\Psi^\kappa_\nu$ on a common probability space $(\Omega,\cl{F},\bb{P})$. 
\end{corollary}

\section{Application to Bose--Einstein condensates}\label{sec:BEC}

As mentioned in the introduction, we apply the theory developed above to a concrete model in many-body quantum physics, namely, the mathematical description of Bose--Einstein condensates. In fact, locally interacting quasi-free states in the sense of Definition~\ref{def:bogostate} appear in the physics and mathematics literature as an effective description of quantum fluctuations in Bose--Einstein condensates.

In Theorem \ref{thm: Bogo_Proba_Representation} we show that, in the limit $N \rightarrow \infty$, the statistics of the quantum depletion are described by an infinite sum of independent geometrically distributed random variables. Moreover, we prove that all moments, as well as the moment generating function of the quantum depletion, admit, in the large-particle limit, a representation in terms of the same infinite sum of geometric random variables. These results rely on approximating fluctuations around the condensate by a Bogoliubov state. This Bogoliubov approximation of quantum fluctuations around Bose--Einstein condensates is well established in both the physical and mathematical literature and has been extensively studied in recent years \cite{BBCS2018,BBCS2019,BSS2, HST2022,NT}. 

To the best of our knowledge, however, a mathematically rigorous derivation of this probabilistic description is new. Explicit asymptotic formulas for the first and second moments of the quantum depletion were previously derived in \cite{BBCS2019}. More recently, bounds on the moment-generating function \cite{NR2024, BBR} were refined in \cite{Rademacher}, where an implicit asymptotic characterization was obtained. In Theorem \ref{thm: Bogo_Proba_Representation}, we complement this result by deriving an explicit asymptotic formula expressed solely in terms of the scattering length of the interaction potential. 

Our results illustrate the role of interactions in the probabilistic representation of the quantum depletion, i.e., of the number of particles outside the condensate, for an \textit{interacting} Bose--Einstein condensate at zero temperature on the unit torus. Recently, however, \cite{Deuchert} studied condensate number fluctuations in \textit{non-interacting} Bose--Einstein condensates at positive temperature in the canonical ensemble. There, deviations of the condensate occupation number from its mean are characterized by left- and right-tail bounds. Earlier, in \cite{Deuchert2}, the distribution of the condensate occupation number was analyzed for the grand-canonical Gibbs state of the mean-field Bose gas. 

\subsection{Bose--Einstein condensation}

The description of Bose--Einstein condensates is based on a gas of $N\in\bb{N}$ bosons on the three-dimensional unit torus $\bb{T}^3$ described on the symmetric Hilbert space $L_{\sym}^2(\bb{T}^{3N})$ by the Hamiltonian
\begin{align*}
H_N = \sum_{i=1}^N (-\Delta_i) + \sum_{i<j} v_N(x_i-x_j),\qquad v_N(x) := N^2 v(Nx).
\end{align*}
Here, the scaling of the two-body interaction potential $N^2 v( N \cdot )$ in $N$ is chosen so that the scattering length is of order $N^{-1}$ in the limit $N \to \infty$. More precisely, let $f$ be the solution of the zero-energy scattering equation
\[
	\bigg( -\Delta + \frac{1}{2} v \bigg) f = 0,\qquad \text{with\;\;$f(x) \to 1$\; as\; $|x| \to \infty$}.
\]
Then the associated scattering length is defined by
\[
\fk{a} := \frac{1}{8\pi}\int v(x)f(x)\,dx,
\]
equivalently through the asymptotics $f(x) = 1 - \fk{a}/|x| + o(|x|^{-1})$ as $|x|\to\infty$. For the scaled interaction potential $v_N$, the corresponding scattering length is precisely $\fk{a}/N$. This scaling is known as the Gross--Pitaevskii regime and describes a system with strong, short-range particle interactions.

It is mathematically convenient to describe the $N$-particle system on the bosonic Fock space with bosonic creation and annihilation operators. The Fock space formalism is particularly useful for describing fluctuations around Bose--Einstein condensates, since it naturally allows for fluctuations in particle number. Moreover, it is well adapted to Fourier analysis and therefore provides a natural framework for passing from position space to momentum space. Indeed, on the torus $\mathbb{T}^3$, the system is translation invariant, and the Hamiltonian commutes with the translation operator. It is therefore natural to work in the eigenbasis of translations, namely the plane waves $\varphi_k(x)=e^{ik\cdot x}$ with $k\in 2\pi \mathbb Z^3,$ which corresponds precisely to expanding the system into Fourier modes. Expressed in terms of the momentum-space creation and annihilation operators, the Hamiltonian takes the form 
\[
    H_N = \sum_{p \in 2 \pi \mathbb{Z}^3} p^2 \sf{a}_p^*\sf{a}_p + \frac{1}{N} \sum_{p, q, \ell \in 2 \pi \mathbb{Z}^3} \widehat{v} (\ell/N) \sf{a}_{p + \ell}^* \sf{a}_{q-\ell}^* \sf{a}_p \sf{a}_q .
\]

It is known that the ground state vector $\uppsi_N\in L_{\sym}^2(\bb{T}^{3N})$ of the Hamiltonian $H_N$ exhibits Bose--Einstein condensation with respect to the condensate wave function $\varphi \equiv 1$ on $\bb{T}^3$, corresponding in momentum space to the zero mode $p=0$. Morally, this means that the expected number of particles occupying the zero mode is asymptotically equal to $N$, while only a negligible fraction occupies orthogonal modes. Mathematically, Bose--Einstein condensation is formulated in terms of the quantum depletion operator, defined by
\[ 
 \sf{N}_+ := \sum_{p \in 2\pi \mathbbZ^3_0} \sf{a}_p^*\sf{a}_p,
\]
which counts the number of particles outside the condensate.

As first proved in \cite{LiebSeiringer2002}, the ground state $\uppsi_N$ satisfies Bose--Einstein condensation in the sense that 
 \begin{align}
 \label{eq:BEC}
  \frac{1}{N}\,\mathbb{E}_{\uppsi_N} [ \mathsf{N}_+ ]  \to 0 \qquad \text{as $N\rightarrow \infty$.}
 \end{align}
Since then, the convergence rate in \eqref{eq:BEC} has been widely studied and optimal error estimates have been established \cite{BBCS2018,BBCS2019,H2021,HST2022}. Later, these results were also extended from Bose gases on the torus $\mathbb{T}^3$ to Bose gases in $\mathbb{R}^3$ with trapping potentials \cite{BSS1,BSS2,LS2006,NNRT,NRS,NT}.

\medskip 
These results rely on an effective description of the quantum fluctuations around the condensate in terms of quasi-free states in position space that, in the setting of the torus, preserve translation invariance. Passing to momentum space, these translation-invariant quasi-free states naturally give rise to locally interacting quasi-free states in the sense of Definition \ref{def:bogostate}. Remarkably, these locally interacting quasi-free states admit an effective description on a lattice whose spacing is determined by the scattering length of the interaction potential.

\subsection{Asymptotic probabilistic representation of the depletion operator}

The connection of Bose--Einstein condensation to the framework of locally interacting quasi-free states developed above rests on a precise asymptotic result for the depletion operator in the large-particle limit. To state this result, we define the operator on $\cl{F}_{\sym}(\ell^2(2 \pi \bb{Z}^3))$
\[
    \sf{K}_\fk{a} := \frac{1}{2} \sum_{p \in 2 \pi \bb{Z}^3 \setminus \{0 \}} \nu_p^\fk{a} (\sf{a}_{p}^* \sf{a}_{-p}^* - \sf{a}_p \sf{a}_{-p}), \quad \text{where} \quad \nu_p^\fk{a} = \frac{1}{4} \ln \left ( \frac{|p|^2}{|p|^2 + 16 \pi \fk{a}} \right ). 
\]
We also define the corresponding quasi-free vector $\Uppsi_\fk{a} = e^{\sf{K}_\fk{a}} \Upomega$. It turns out that this vector precisely captures the fluctuations of $\uppsi_N$ away from the condensate. While this can be stated in various ways, we shall use the following statement on the moment generating function proved by the second named author in \cite{Rademacher}\footnote{We remark that the published version \cite{Rademacher} contains a minor typographical error in the definition of the quasi-free state. More precisely, in \cite[Eqs~(3.54),(4.2)]{Rademacher}, the factor $1/2$ in the operator $K_\nu$ is missing.}. 

\begin{proposition}[Lemma 4.1 of \cite{Rademacher}]  \label{prop:SimoneAsymptotic}
    There exists $\lambda_0 > 0$ such that for all $\lambda \in \bb{C}$ with $|\lambda| \leq \lambda_0$, and $k \in \bb{N}$, we have
    \[
        \bb{E}_{\uppsi_N}[(\sf{N}_+)^k e^{\lambda \sf{N}_+}] = \bb{E}_{\Uppsi_{\fk{a}}}[(\sf{N}_+)^k e^{\lambda \sf{N}_+}] + \cl{O}(N^{-1/4}).
    \]
\end{proposition}

In other words, the moment generating function of the quantum depletion operator in the state $\uppsi_N$ converges to the one in the state $\Uppsi_{\fk{a}}$. In \cite{Rademacher}, this convergence result was supplemented with an integral representation of the limiting characteristic function. Using our results on locally interacting quasi-free states developed in the previous Section, we will improve on the integral formula by deriving a probabilistic representation of the law of $\sf{N}_+$ in the quasi-free state $\Uppsi_{\fk{a}}$. 

\medskip
The key observation is that only pairs $\{p,-p\}$ couple in $\sf{K}_\fk{a}$, which suggests indexing the modes by such pairs. Accordingly, let $\sim$ be the equivalence relation on $\bb{R}^3_0$: for all $x,y \in \bb{R}^3_0$
\[
    x \sim y \Longleftrightarrow x = \pm y,
\]
and observe that
\[
    \sf{K}_\fk{a} = \sum_{p \in 2 \pi \bb{Z}^3_0 / \sim} \nu_p^{\fk{a}} (\sf{a}_p^* \sf{a}_{-p}^* - \sf{a}_p \sf{a}_{-p}),
\]
the factor $\tfrac12$ having been absorbed by summing over pairs rather than over all modes.

The weights $\nu^\fk{a}$ then prescribe a natural rescaling of the lattice. Indeed, writing
\[
    \nu(x) := \frac{1}{4} \ln \left ( \frac{|x|^2}{|x|^2 + 1} \right ),
\]
we have $\nu_p^\fk{a} = \nu\bigl(p / 4 \sqrt{\pi \fk{a}}\bigr)$ for every $p\in 2\pi\bb{Z}^3_0$, so that the substitution $x = p/4\sqrt{\pi\fk{a}}$ carries $2\pi\bb{Z}^3_0$ onto $\kappa\bb{Z}^3_0$ with
\begin{equation}\label{eq:kappa-a}
    \kappa := \frac{2\pi}{4\sqrt{\pi\fk{a}}} = \sqrt{\frac{\pi}{4\fk{a}}}\,.
\end{equation}
Let us introduce the quotient of the punctured lattice with spacing $\kappa$, $\Lambda^\kappa := \kappa \bb{Z}^3_0/{\sim}$ and, for each $p\in\Lambda^\kappa$, fix once and for all a representative of the class $p$. In the following, we shall write $A_\Lambda^\kappa$ instead of $A_{\Lambda^\kappa}$ for any object whose definition depends on $\Lambda^\kappa$. 

\begin{proposition}[Depletion as a locally interacting quasi-free observable]\label{prop:depletion-representation}
    Let $\fk{a}>0$, let $\kappa$ be given by \eqref{eq:kappa-a}, and let $\Uppsi^\kappa_\nu$ be the locally interacting quasi-free vector on $\fk{h}_{\Lambda}^\kappa$ associated with $\nu_p = \nu(p)$. Then the law of $\sf{N}_+$ in the state $\Uppsi_\fk{a}$ coincides with the law of $\sf{N}_\Lambda^\kappa$ in the state $\Uppsi_\nu^\kappa$. In particular,
    \[
        \bb{E}_{\Uppsi_\fk{a}}\bigl[\sf{N}_+^k e^{\lambda \sf{N}_+}\bigr] = \bb{E}_{\Uppsi_\nu^\kappa}\bigl[(\sf{N}_\Lambda^\kappa)^k e^{\lambda \sf{N}_\Lambda^\kappa}\bigr]
    \]
    for every $k\in \bb{N}_0$ and every $\lambda\in\bb{C}$ for which either side is finite.
\end{proposition}
\begin{proof}
    Splitting the one-particle space according to the zero mode and the pairs $\{p,-p\}$, and identifying each such pair with the spin space $\bb{S}$ by $\ket{\pm}\leftrightarrow \pm\tfrac{2\pi}{\kappa}p$, gives the orthogonal decomposition
    \[
        \ell^2(2\pi\bb{Z}^3) \;\cong\; \bb{C}\oplus\bigoplus_{p\in\Lambda^\kappa} \bb{S}.
    \]
    Second quantization turns this direct sum into a tensor product, so that $\cl{F}_{\sym}(\ell^2(2 \pi \bb{Z}^3)) \cong \cl{F}_{\sym}(\bb{C}) \otimes \fk{h}_{\Lambda}^\kappa$, under an isometry determined by
    \[
        \sf{a}_{\pm \frac{2 \pi}{\kappa} p} \longleftrightarrow \sf{a}_{(p,\pm)}, \quad \text{and} \quad \sf{a}_0 \longleftrightarrow \sf{a}_{(0)}.
    \]
    Under this isometry, $\sf{K}_{\fk{a}}$ is mapped to the operator $\mathrm{I}_{\cl{F}_{\sym}(\bb{C})} {\otimes} \sf{K}_\nu^\kappa$ with
    \[
        \sf{K}_\nu^\kappa = \sum_{p \in \Lambda^\kappa} \nu(p)(\sf{a}^*_{(p,+)} \sf{a}^*_{(p,-)} - \sf{a}_{(p,+)} \sf{a}_{(p,-)}),
    \]
    since $\nu^\fk{a}_{(2\pi/\kappa)p} = \nu(p)$ by construction, and consequently $\Uppsi_\fk{a} \cong \Upomega_0 {\otimes} \Uppsi_\nu^\kappa$. The number operator decomposes in the same way, $\sf{N}_+$ being mapped to $\mathrm{I}_{\cl{F}_{\sym}(\bb{C})} {\otimes} \sf{N}_\Lambda^\kappa$. Since the vacuum factor $\Upomega_0$ carries no particles, both assertions follow.
\end{proof}

Combining this with the probabilistic representation provided by Theorem \ref{thm:N+_bogo}, and the convergence result of \cite{Rademacher} recalled in Proposition \ref{prop:SimoneAsymptotic}, we derive an asymptotic probabilistic representation of the law of the quantum depletion operator in the large-particle limit. 

\begin{theorem} \label{thm: Bogo_Proba_Representation}
    Let $\uppsi_N$ be the ground state of the Hamiltonian $H_N$ with scattering length $\fk{a}$, and let $\kappa$ be given by \eqref{eq:kappa-a}. Then, there exists $\lambda_0 > 0$ such that for all $\lambda \in \bb{C}$ with $|\lambda| \leq \lambda_0$ and $k \in \bb{N}$,
    \[
        \bb{E}_{\uppsi_N}[\sf{N}_+^k e^{\lambda \sf{N}_+}] = \bb{E}_{\Uppsi_\nu^\kappa}[(\sf{N}_\Lambda^\kappa)^k e^{\lambda \sf{N}_\Lambda^\kappa}] + \cl{O}_{\kappa,k}(N^{-1/4}).
    \]
    Consequently, the statistics of the quantum depletion operator $\sf{N}_+$ in the large-particle limit are approximated by the statistics of the random variable 
    \[
        2 G^\kappa_\Lambda = 2 \sum_{p \in \Lambda^\kappa} G_p,
    \]
    where $(G_p)_{p \in \Lambda^\kappa}$ are independent with $G_p \sim \mathrm{Geo}(\sech^2(\nu(p)))$.
\end{theorem}

\subsection{Asymptotics of the quantum depletion}

In the following, we study two asymptotic regimes of the scattering length: the regime $\fk{a} \to 0$ corresponds to a weak-coupling limit in which interaction effects become asymptotically negligible, and the Bose gas approaches an ideal Bose gas. 
The regime of $\fk a \to \infty$ corresponds to a Bose gas dominated by correlation effects. We also provide tail estimates in the intermediate regime. 

Notice that the scattering length is inversely proportional to the square of the lattice spacing, and as a consequence:
\begin{itemize}
    \item The small scattering regime $\fk{a} \to 0$ corresponds to the sparse limit $\kappa \to +\infty$.
    \item The large scattering regime $\fk{a} \to +\infty$ corresponds to the dense limit $\kappa \to 0$. 
\end{itemize}

Moreover, the weight function $\nu$ satisfies property \eqref{ass:dense-conditions}. Indeed, setting
\begin{equation*}
    \mu^\mathsf{QD}(x) := \sinh^2(\nu(x)) =  \frac{\left ( \sqrt{\frac{|x|^2}{1+|x|^2}} - 1 \right )^2}{4 \sqrt{\frac{|x|^2}{1+|x|^2}}},\qquad x\in \bb{R}^3_0,
\end{equation*}
then elementary computations yield
\begin{equation*}
    \limsup_{|x|\to +\infty} |x|^\gamma\mu^\mathsf{QD}(x) = \begin{cases} 0 & \text{for $\gamma \in[0,4)$}, \\
    \tfrac{1}{16} & \text{for $\gamma=4$}, \\
    +\infty & \text{for $\gamma>4$}.
    \end{cases},\quad
    \limsup_{|x|\to 0} |x|^\beta\mu^\mathsf{QD}(x) = \begin{cases} +\infty & \text{for $\beta \in (0,1)$}, \\
    \tfrac{1}{4} & \text{for $\beta=1$}, \\
    0 & \text{for $\beta>1$}.
    \end{cases}
\end{equation*}
In particular, property \eqref{ass:dense-conditions} is satisfied with the non-trivial parameters $\gamma=4>3$ and $\beta=1<\frac{3}{2}$. Moreover, the limit superior at infinity is in fact attained as a genuine limit, so that $\nu$ also satisfies condition \eqref{eq:sparse-condition} with the same exponent $\gamma=4>3$ and constant $\fk{c}_\nu = \tfrac{1}{16}$. Both regimes of Section~\ref{subsec:tail} are therefore available.

\medskip
To apply the results of Section~\ref{subsec:tail}, we need to replace the sum over the lattice $\Lambda^\kappa$ by a sum over $\kappa \bb{Z}^3_0$. To do so, we consider for each $p \in \Lambda^\kappa$ two independent copies $G_p^+,G_p^-$ of $G_p$, and we let $G^\kappa_+,G^\kappa_-$ be the corresponding sums. We also consider for $p \in \kappa \bb{Z}^3_0$ independent random variables $G^\rm{sum}_p$ with $G^\rm{sum}_p \sim \mathrm{Geo}(\sech^2(\nu(p)))$. It is not hard to see that $G^\kappa_+ + G^\kappa_-$ has the same law as $G^\kappa_\rm{sum} = \sum_{p \in \kappa \bb{Z}^3_0} G^\rm{sum}_p$. Owing to the results of Section~\ref{subsec:tail}, $G^\kappa_\rm{sum}$ obeys a law of small numbers and a large deviation principle in the regime $\kappa \to +\infty$, and satisfies a central limit theorem in the regime $\kappa \to 0$. 

\subsubsection*{\bf\em Small scattering limit, $\fk{a} \to 0$}

In the limit $\fk a \to 0 $, corresponding to asymptotically weak interactions, the quantum depletion becomes asymptotically negligible, and the system approaches the behavior of an almost ideal gas. Theorem \ref{thm: BEC-small a} captures this regime by establishing a law of small numbers for the quantum depletion, thereby characterizing the statistics of rare events away from complete condensation.

\begin{theorem}[Law of small numbers for quantum depletion]
\label{thm: BEC-small a}
    Set $\fk{c}_\mathsf{QD} := \zeta_3(4)/(2 \pi^2)$. Then, as $\fk{a} \to 0$, the quantum depletion $\sf{N}_+$ in the quasi-free state $\Uppsi_\fk{a}$ satisfies:
    \begin{enumerate}[(i)]
        \item \emph{(Scaled log-Laplace transform)} For every $t \in \bb{R}$,
        \begin{equation}\label{eq:bec-laplace-limit}
            \lim_{\fk{a} \to 0} \frac{1}{\fk{a}^2} \ln \bb{E}_{\Uppsi_\fk{a}}\bigl[e^{t \sf{N}_+}\bigr] = \fk{c}_\mathsf{QD}\bigl(e^{2t} - 1\bigr).
        \end{equation}
        \item \emph{(Two-point law)} For every $k \in \bb{N}_0$,
        \begin{equation}\label{eq:bec-two-point}
            \bb{P}_{\Uppsi_\fk{a}}\bigl(\sf{N}_+ = 2 k\bigr) = o(\fk{a}^2) +
            \begin{cases}
                1 - \fk{c}_\mathsf{QD}\, \fk{a}^2 & \text{for $k = 0$}, \\
                \fk{c}_\mathsf{QD}\, \fk{a}^2 & \text{for $k = 1$}, \\
                0 & \text{otherwise}.
            \end{cases}
        \end{equation}
        \item \emph{(Exponential moments)} For every $\theta > 0$ and every $\lambda \in \bb{C}$,
        \begin{equation}\label{eq:bec-moments}
            \bb{E}_{\Uppsi_\fk{a}}\bigl[\sf{N}_+^{\theta} e^{\lambda \sf{N}_+}\bigr] = \fk{c}_\mathsf{QD}\, 2^{\theta} e^{2 \lambda} \fk{a}^2 + o(\fk{a}^2).
        \end{equation}
    \end{enumerate}
\end{theorem}

In other words, the quantum depletion obeys a law of small numbers in the weak coupling limit: to leading order in $\fk{a}^2$,
\[
    \sf{N}_+ \approx 2\,\mathrm{Poi}\bigl(\fk{c}_\mathsf{QD}\,\fk{a}^2\bigr)\qquad\text{as $\fk{a}\to 0$}.
\]
Condensation is thus complete with probability $1 - \fk{c}_\mathsf{QD}\fk{a}^2 + o(\fk{a}^2)$, and the leading correction consists of the excitation of a single pair of particles with opposite momenta.

\begin{remark}
    The constant $\zeta_3(4)$ has, for instance, been computed in \cite{beane_n-boson_2007}, giving the value $\zeta_3(4) \approx 16.532315959$. This then gives $\fk{c}_\mathsf{QD} \approx 0.83753691065$.
\end{remark}

\begin{proof}[Proof of Theorem~\ref{thm: BEC-small a}]
    Since $\nu$ satisfies condition \eqref{eq:sparse-condition} on $\kappa\bb{Z}^3_0$ with exponent $\gamma = 4 > 3$ and constant $\fk{c}_\nu = \tfrac{1}{16}$, Theorem~\ref{thm:poisson-sparse}\emph{(i)}, applied to $G^\kappa_\rm{sum}$, gives
    \begin{equation*}
        \kappa^4 \ln \bb{E}\bigl[e^{2 t G^\kappa_\rm{sum}}\bigr] \longrightarrow \tfrac{1}{16} \zeta_3(4)\,(e^{2 t} - 1)\qquad\text{as $\kappa \to +\infty$}.
    \end{equation*}
    On the other hand, $G^\kappa_\rm{sum}$ has the same law as $G^\kappa_+ + G^\kappa_-$, where $G^\kappa_\pm$ are independent copies of $G^\kappa$, so that $\bb{E}[e^{2 t G^\kappa_\rm{sum}}] = \bigl(\bb{E}[e^{2 t G^\kappa}]\bigr)^2$ and therefore
    \begin{equation*}
        2 \kappa^4 \ln \bb{E}\bigl[e^{2 t G^\kappa}\bigr] \longrightarrow \frac{1}{16} \zeta_3(4)\,(e^{2 t} - 1).
    \end{equation*}
    By Proposition~\ref{prop:depletion-representation} together with Theorem~\ref{thm:N+_bogo}, the observable $\sf{N}_+$ in the state $\Uppsi_\fk{a}$ has the law of $2G^\kappa$, hence, $\ln \bb{E}[e^{2 t G^\kappa}] = \ln \bb{E}_{\Uppsi_\fk{a}}[e^{t \sf{N}_+}]$. Substituting $\kappa = \sqrt{\pi / 4 \fk{a}}$, that is $\kappa^{-4} = 16\fk{a}^2/\pi^2$, we conclude that
    \begin{equation*}
        \frac{1}{\fk{a}^2} \ln \bb{E}_{\Uppsi_\fk{a}}\bigl[e^{t \sf{N}_+}\bigr] \longrightarrow \frac{\zeta_3(4)}{2 \pi^2}\,(e^{2t} - 1) = \fk{c}_\mathsf{QD}\,(e^{2t}-1),
    \end{equation*}
    which is \eqref{eq:bec-laplace-limit}.

    Parts~\emph{(ii)} and~\emph{(iii)} then follow from Proposition~\ref{prop:small-numbers} in Appendix~\ref{app:small-numbers}, applied along an arbitrary sequence $\fk{a}_N \to 0$ with $X_N$ having the law of $\sf{N}_+$ in $\Uppsi_{\fk{a}_N}$, and with $r_N = \fk{a}_N^2$, $\fk{c} = \fk{c}_\mathsf{QD}$ and $K = 2$.
\end{proof}

As in the $\kappa \bb{Z}^d_0$ lattice case, the previous Theorem identifies the leading correction but does not capture the correct order. We thereby complement it with a large-deviation principle.

\begin{theorem}[LDP at logarithmic speed for quantum depletion]
    The following statement holds.
    \begin{enumerate}[(i)]
        \item \emph{(Local asymptotics)} As $\fk{a} \to 0$, uniformly in $k \in \bb{N}_0$,
        \[
            \ln \bb{P}_{\Psi_\fk{a}}(\sf{N}_+ = 2 k) = -2 k \ln \fk{a}^{-1} + (k+1) O(1).
        \]
        \item \emph{(Logarithmic LDP)} The family of laws of the depletion operator $\sf{N}_+$ in the condensate $\Psi_\fk{a}$ satisfies, as $\fk{a} \to 0$, a large deviation principle on $\bb{R}$ with speed $\ln \fk{a}^{-1}$ and good rate function 
        \[
            \mathscr{J}_{\mathrm{BE}}(x) = \begin{cases}
                x & \text{if $x \in 2 \bb{N}_0$}, \\ +\infty & \text{else}.
            \end{cases}
        \]
        i.e. for every Borel $E \in \cl{B}(\bb{R})$ one has
        \[
            -\inf_{\overset{\circ}{E}} \mathscr{J}_{\mathrm{BE}} \leq \liminf_{\fk{a} \to 0} \frac{1}{\ln \fk{a}^{-1}} \ln \bb{P}_{\Psi_\fk{a}}(\sf{N}_+ \in E) \leq \limsup_{\fk{a} \to 0} \frac{1}{\ln \fk{a}^{-1}} \ln \bb{P}_{\Psi_\fk{a}}(\sf{N}_+ \in E) \leq -\inf_{\overline{E}} \mathscr{J}_{\mathrm{BE}}
        \]
    \end{enumerate}
\end{theorem}

\begin{proof}
    \emph{Local lower bound.} A careful inspection of the proof of Lemma~\ref{lem:lower-local} shows that one can use $\Lambda^\kappa$ instead of $\kappa \bb{Z}^d_0$. Indeed, the \emph{site bound} of Lemma~\ref{lem:site-bounds} extends to the quotient lattice, and hence the proof of the local bound follows. Therefore, for $\gamma = 4$, we get
    \[
        \ln \bb{P}(G^\kappa = k) \geq -4 k \ln \kappa + (k+1) O(1),
    \]
    uniformly in $k$ when $\kappa \to +\infty$. Substituting Substituting the expression for $\kappa$ we obtain
    \[
        \ln \bb{P}(G^\kappa = k) \geq -2 k \ln \fk{a}^{-1} + (k+1) O(1)
    \]
    uniformly in $k$ as $\fk{a} \to 0$. 

    \medskip \noindent \emph{Local upper bound.} We have
    \[
        \ln \bb{P}(G^\kappa = k) = \frac{1}{2} \ln \bb{P}(G^\kappa_+ = k, G^\kappa_- = k) \leq \frac{1}{2} \ln \bb{P}(G^\kappa_{\rm{sum}} = 2k).
    \]
    By Theorem~\ref{thm:log-ldp}(i) the latter is $-4 (2k) \ln \kappa + k O(1)$. Therefore, substituting again the expression for $\kappa$, we get
    \[
        \ln \bb{P}(G^\kappa = k) \leq -2 k \ln \fk{a}^{-1} + (k+1) O(1).
    \]

    \medskip \noindent \emph{Log-LDP.} Combining the local lower and upper bounds we get as $\fk{a} \to 0$, uniformly in $k$,
    \[
        \ln \bb{P}(G^\kappa = k) = -2 k \ln \fk{a}^{-1} + (k+1) O(1).
    \]
    By Proposition~\ref{prop: two_sided_to_LDP} with $\gamma = 2$, after replacing the parameter $\kappa$ in the statement by $\fk{a}^{-1}$, $(G^\kappa)_\kappa$ satisfies a LDP as $\fk{a} \to 0$ with speed $\ln \fk{a}^{-1}$ and good rate function $\mathscr{J}_2$. By transfer, the law of $\sf{N}_+$ in $\Psi_\fk{a}$ satisfies a LDP with speed $\ln \fk{a}^{-1}$ and good rate function $x \to \mathscr{J}_2(x/2)$. 
\end{proof}

\subsubsection*{\bf\em Large scattering limit, $\fk{a}\to +\infty$}

The limit $\fk{a} \to +\infty$, also referred to as the Thomas-Fermi regime (see, for example, \cite{Coreggi1, Coreggi2, Dimonte}), corresponds physically to a regime of strongly interacting Bose gases. In this regime, correlation effects become dominant, and the fluctuation structure changes qualitatively. 
This behavior is reflected in Theorem \ref{thm: BEC-large a}, where the quantum depletion is shown to satisfy a central limit theorem. We set
\[
    \mu_\fk{a}^\mathsf{QD} = \frac{1}{2} \bb{E}_{\Uppsi_\fk{a}}[\sf{N}_+] = \sum_{p \in \Lambda^\kappa} \mu_p^\mathsf{QD}, \quad \text{and} \quad (\sigma_\fk{a}^\mathsf{QD})^2 = \frac{1}{4} \bb{V}_{\Uppsi_\fk{a}}[\sf{N}_+] = \sum_{p \in \Lambda^\kappa} \mu_p^\mathsf{QD}(1+\mu_p^\mathsf{QD}). 
\]

\begin{theorem}[Large scattering limit, $\fk{a}\to +\infty$]
\label{thm: BEC-large a}
    As $\fk{a} \to +\infty$ we have:
    \begin{enumerate}[(i)]
        \item The distributions of the scaled centered observables
        \[
            \frac{\sf{N}_+ - 2\mu_{\fk{a}}^\mathsf{QD}}{2 \sigma_\fk{a}^\mathsf{QD}}\quad\text{converge weakly to $\cl{N}(0,1)$,}
        \]
        together with all its moments of order $\theta \in [1,3)$.
        \item The mean and the variance of the quantum depletion satisfy
        \begin{equation}\label{eq:bec-mean-variance}
            \bb{E}_{\Uppsi_\fk{a}}[\sf{N}_+] = 2 \mu_\fk{a}^\mathsf{QD} = \frac{8}{3 \sqrt{\pi}} \fk{a}^{3/2}  + O(\sqrt{\fk{a}}), \quad \text{and} \quad \bb{V}_{\Uppsi_\fk{a}}[\sf{N}_+] \sim 2\sqrt{\pi}\, \fk{a}^{3/2}.
        \end{equation}
        \item Consequently, the law in $\Uppsi_\fk{a}$ of the operators
        \[
            \fk{a}^{3/4} \left(\frac{1}{\fk{a}^{3/2}} \sf{N}_+ - \frac{8}{3 \sqrt{\pi}} \right)\quad \text{converges weakly to $\cl{N}\bigl(0,2\sqrt{\pi}\bigr)$,}
        \]
        together with its moments of order $\theta \in [1,3)$.
    \end{enumerate}
\end{theorem}

\begin{proof}
    We recall that $\sf{N}_+$ has the law of $2 G_\kappa$ in the condensate $\Uppsi_\fk{a}$, where $\kappa = \sqrt{\pi / 4 \fk{a}}$. We also recall the notation $G^\kappa_\mathrm{sum} = G^\kappa_+ + G^\kappa_- = \sum_{p \in \kappa \bb{Z}^3_0} G_p^{sum}$, where $G_\pm^\kappa$ are independent copies of $G^\kappa$. 

	\emph{(i)} As $\mu^\mathsf{QD}$ satisfies assumption \eqref{ass:dense-conditions-proof} with $\gamma = 4$ and $\beta = 1$, the variable $G_\mathrm{sum}^\kappa = G^\kappa_+ + G^\kappa_-$ satisfies a central limit theorem. More precisely, the law of the scaled random variable 
	\[
		X^\kappa_\mathrm{sum} := \frac{G_\mathrm{sum}^\kappa - \bb{E}[G_\mathrm{sum}^\kappa]}{\sqrt{\bb{V}[G_\mathrm{sum}^\kappa]}}\qquad\text{converges weakly to $\cl{N}(0,1)$,}
	\]
	together with all its moments of order $\theta \in [1,3)$. But we can write $X^\kappa_\mathrm{sum} = X^\kappa_+ + X^\kappa_-$, where $X_\pm^\kappa$ are independent copies having the law of $\frac{G^\kappa - \bb{E}[G^\kappa]}{\sqrt{2\bb{V}[G^\kappa]}}$. By Lemma \ref{lemma: CLT_sum_CLT_non_sum}, we deduce that $\frac{G^\kappa - \bb{E}[G^\kappa]}{\sqrt{\bb{V}[G^\kappa]}}$ also converges weakly to $\cl{N}(0,1)$, together with all its moments of order $\theta \in [1,3)$, which is precisely the assertion since $\bb{E}[G^\kappa] = \mu_\fk{a}^\mathsf{QD}$ and $\bb{V}[G^\kappa] = (\sigma_\fk{a}^\mathsf{QD})^2$.
     
     \medskip   
     \emph{(ii)} We have
        \[
            \bb{E}_{\Uppsi_\fk{a}}[\sf{N}_+] = 2 \mu_\fk{a}^\mathsf{QD} = 2 \sum_{p \in \Lambda^\kappa} \mu_p^\mathsf{QD} = \sum_{p \in \kappa \bb{Z}^3_0} \mu_p^\mathsf{QD}.
        \]
        Similarly
        \[
            \bb{V}_{\Uppsi_\fk{a}}[\sf{N}_+] = 4 (\sigma_{\fk{a}}^\mathsf{QD})^2 = 2 \sum_{p \in \kappa \bb{Z}^3_0} \mu_p^\mathsf{QD}(\mu_p^\mathsf{QD} + 1).
        \]
        As $\mu_x^\mathsf{QD} \asymp |x|^{-1}$ as $|x| \to 0$, and $\mu_x^\mathsf{QD} \asymp |x|^{-4}$ when $|x| \to +\infty$, we have $\mu_x^\mathsf{QD}(\mu_x^\mathsf{QD} + 1) \asymp |x|^{-2}$ as $|x| \to 0$, and $\asymp |x|^{-4}$ as $|x| \to +\infty$. Therefore we can apply Lemma \ref{lemma: Riemann_approx} to deduce that
        \[
            \kappa^3 \sum_{p \in \kappa \bb{Z}^3_0} \mu_p^\mathsf{QD}(\mu_p^\mathsf{QD} + 1) \to \int_{\bb{R}^3_0} \mu_x^\mathsf{QD}(\mu_x^\mathsf{QD} + 1) \dd{x}. 
        \]
        By Lemma~\ref{lem:QD-integrals} below, we have $\int_{\bb{R}^3_0} \mu_x^\mathsf{QD}(1+\mu_x^\mathsf{QD})\dd{x} = \pi^2/8$. Replacing $\kappa = \sqrt{\pi/4 \fk{a}}$, we deduce that
        \[
            \bb{V}_{\Uppsi_\fk{a}}[\sf{N}_+] \sim \frac{\pi^2}{4} \frac{1}{\kappa^3} \sim 2\sqrt{\pi}\, \fk{a}^{3/2}.
        \]
        For the mean, we prove that $\mu_x^\mathsf{QD}$ satisfies the assumptions of Lemma \ref{lem: improved_Riemann_approx}. We recall that if $f(x) = h(|x|^2)$ for some function $h$ of class $C^2(\bb{R}_+^*)$, then
        \[
            D^2 f(x) = 2 h'(|x|^2)\, \rm{I}_3 + 4 h''(|x|^2) \, x \otimes x.
        \]
        In particular $|D^2 f(x)| \leq C (|h'(|x|^2)| + |x|^2 |h''(|x|^2)|)$. In our case, we have $\mu_x^\mathsf{QD} = m_\mathsf{QD}(|x|^2)$ with
        \begin{equation*}
            m_\mathsf{QD}(r) = \frac{1}{4} \left ( \sqrt{\frac{r}{1+r}} + \sqrt{\frac{r+1}{r}} - 2 \right ).
        \end{equation*}
        This gives 
        \begin{equation*}
            m_\mathsf{QD}'(r) = -\frac{1}{8} \frac{1}{r^{3/2} (1+r)^{3/2}}, \quad \text{and} \quad m''_\mathsf{QD}(r) = \frac{3}{16} \frac{2r+1}{r^{5/2}(r+1)^{5/2}}.
        \end{equation*} 
        For $r \gg 1$ we get $m'_\mathsf{QD}(r) \leq C/r^3$ and $m''_\mathsf{QD}(r) \leq C/r^4$. This gives $|D^2 \mu_x^\mathsf{QD}| \leq C/|x|^6$. Therefore, the asymptotic behavior for large $|x|$ of the Hessian holds with $\beta = 6 > 3$. It remains to check the behavior as $|x| \to 0$. To do so, we let $g_x^\mathsf{QD} = \mu_x^\mathsf{QD} - \frac{1}{4 |x|}$.
        \begin{equation*}
            D^2 g_x^\mathsf{QD} = 2 \left ( m'_\mathsf{QD}(|x|^2) + \frac{1}{8 |x|^3} \right ) \rm{I}_3 + 4 \left ( m''_\mathsf{QD}(|x|^2) - \frac{3}{16 |x|^5} \right ) x \otimes x.
        \end{equation*}
        But for small $|x|$, we have
        \[
            m'_\mathsf{QD}(r) = -\frac{1}{8 r^{3/2}}\bigl(1 + O(r)\bigr) = -\frac{1}{8r^{3/2}} + O \biggl ( \frac{1}{\sqrt{r}} \biggr ),
        \]
        and
        \[
            m''_\mathsf{QD}(r) = \frac{3}{16 r^{5/2}} \bigl ( 1+ O(r) \bigr ) = \frac{3}{16 r^{5/2}} + O \biggl ( \frac{1}{r^{3/2}} \biggr ).
        \]
        Thus,
        \[
            |D^2 g_x^\mathsf{QD}| = O \biggl ( \frac{1}{|x|} \biggr ).
        \]
        Since we also have $\mu_x^\mathsf{QD} = \frac{1}{4 |x|} + O(1)$, we deduce that the asymptotic assumption at $|x| \to 0$ holds with $\alpha = 1 < 3$. This shows that we can apply Lemma \ref{lem: improved_Riemann_approx}. Since $\int_{\bb{R}^3_0} \mu_x^\mathsf{QD} \dd{x} = \pi/3$ by Lemma~\ref{lem:QD-integrals}, we deduce that
        \[
            \bb{E}_{\Uppsi_\fk{a}}[\sf{N}_+] = \sum_{p \in \kappa \bb{Z}^3_0} \mu_p^\mathsf{QD} = \frac{1}{\kappa^3} \int_{\bb{R}^3_0} \mu_x^\mathsf{QD} \dd{x} + O(\kappa^{-1}) = \frac{8}{3 \sqrt{\pi}} \fk{a}^{3/2}  + O(\sqrt{\fk{a}})
        \]
        
        \medskip
        \emph{(iii)} It remains to prove the last statement. By part~(ii), we have
        \[
            \sqrt{\bb{V}_{\Uppsi_\fk{a}}[\sf{N}_+]} = \bigl(2\sqrt{\pi}\bigr)^{1/2} \fk{a}^{3/4}\,(1+o(1)),\qquad \bb{E}_{\Uppsi_\fk{a}}[\sf{N}_+] = \frac{8}{3\sqrt{\pi}}\fk{a}^{3/2} + O(\sqrt{\fk{a}}),
        \]
        and therefore
        \[
            \fk{a}^{3/4}\left(\frac{1}{\fk{a}^{3/2}}\sf{N}_+ - \frac{8}{3\sqrt{\pi}}\right)
            = \frac{\sf{N}_+ - \bb{E}_{\Uppsi_\fk{a}}[\sf{N}_+]}{\fk{a}^{3/4}} + O\bigl(\fk{a}^{-1/4}\bigr).
        \]
        By part~(i), the first term on the right-hand side equals
        \[
            \frac{\sqrt{\bb{V}_{\Uppsi_\fk{a}}[\sf{N}_+]}}{\fk{a}^{3/4}}\cdot \frac{\sf{N}_+ - \bb{E}_{\Uppsi_\fk{a}}[\sf{N}_+]}{\sqrt{\bb{V}_{\Uppsi_\fk{a}}[\sf{N}_+]}},
        \]
        where the first factor converges to $(2\sqrt{\pi})^{1/2}$ and the second converges weakly, together with its moments of order $\theta\in[1,3)$, to $\cl{N}(0,1)$. Since convergence in law and of moments is preserved under the addition of a sequence converging to zero, we conclude that
        \[
            \fk{a}^{3/4} \left ( \frac{1}{\fk{a}^{3/2}} \sf{N}_+ - \frac{8}{3 \sqrt{\pi}} \right ) \longrightarrow \cl{N}\bigl(0,2\sqrt{\pi}\bigr),
        \]
        weakly in the condensate $\Uppsi_\fk{a}$, together with its moments of order $\theta \in [1,3)$.
\end{proof}

The proof above relies on the following two explicit integral formulas, whose computation we defer to Appendix~\ref{app:refined-QD}.

\begin{lemma}[Explicit integrals]\label{lem:QD-integrals}
    Let $\mu^\mathsf{QD}_x = \sinh^2(\nu(x))$ with $\nu(x) = \tfrac14\ln\bigl(|x|^2/(|x|^2+1)\bigr)$. Then
    \[
        \int_{\bb{R}^3_0} \mu_x^\mathsf{QD} \dd{x} = \frac{\pi}{3},\qquad\text{and}\qquad \int_{\bb{R}^3_0} \mu_x^\mathsf{QD}\bigl(1+\mu_x^\mathsf{QD}\bigr) \dd{x} = \frac{\pi^2}{8}.
    \]
\end{lemma}

%


\subsubsection*{\bf\em Intermediate scattering length: tail estimates}

This follows by a direct application of Theorem~\ref{thm:ldp-intermediate} in the case of the quotient lattice $\Lambda^\kappa$. $\overline{\mu}_{\Lambda^\kappa}$ is then attained for $|p| = \kappa$. We thus introduce
\[
    \overline{\mu}_{\fk{a}} = \frac{1}{4} \left ( \sqrt{\frac{\pi/4}{\fk{a} + \pi/4}} + \sqrt{\frac{\fk{a} + \pi/4}{\pi/4}} - 2 \right ) = \overline{\mu}_{\Lambda^\kappa}.
\]
Set  $2 \mu_\fk{a} = \bb{E}_{\Psi_\fk{a}}[\sf{N}_+]$. Then Theorem~\ref{thm:ldp-intermediate} yields the following.

\begin{theorem}
    For all $\lambda \geq 1$, the following two-sided bound holds.
    \[
        (\lfloor \lambda \mu_\fk{a} \rfloor +1) \ln \left ( \frac{\overline{\mu}_\fk{a}}{1 + \overline{\mu}_\fk{a}} \right ) \leq \ln \bb{P}_{\Psi_\fk{a}}(\sf{N}_+ > 2 \lambda \mu_\fk{a}) \leq -\lambda \mu_\fk{a} \ln \left ( \frac{\lambda(1+\overline{\mu}_\fk{a}}{1+ \lambda \overline{\mu}_\fk{a}} \right ) - \frac{\mu_\fk{a}}{\overline{\mu}_\fk{a}} \ln \left ( \frac{1 + \overline{\mu}_\fk{a}}{1+ \lambda \overline{\mu}_\fk{a}} \right ).
    \]
    The lower bound is also valid for $\lambda \in (0,1)$. Consequently,
    \[
        \lim_{\lambda \to +\infty} \frac{1}{\lambda} \ln \bb{P}_{\Psi_\fk{a}}(\sf{N}_+ > 2 \lambda \mu_\fk{a}) = -\cl{J}_\fk{a}, \quad \cl{J}_\fk{a} := \mu_\fk{a} \ln \left ( 1 + \frac{1}{\overline{\mu}_\fk{a}} \right ).
    \]
\end{theorem}

Therefore the probability that $\sf{N}_+$ exceeds $\lambda$ times its mean decays exponentially, with the leading-order decay governed solely by the occupation law of the site closest to the origin.

\section{Probabilistic representation}
\label{sec:intro-proof}

As described in the introduction, we consider bosonic spin particles on $\Lambda\subset\bb{Z}_\kappa^d$ represented by the lattice Fock space 
\[
	\fk{h}_\Lambda = \cl{F}_{\sym}(\ell^2(\Lambda{\times}\{+,-\}))\cong \bigotimes_{p \in \Lambda} \fk{h}_p,\qquad \fk{h}_p = \cl{F}_{\sym}(\bb{S}).
\]
Here $\bb{S}=\text{span}_\bb{C}\{\ket{+},\ket{-}\}\cong\bb{C}^2$ is the local spin space. In this section, we prove Theorem \ref{thm:N+_bogo} on the distribution of the particle number operator
\[
	\sf{N}_\Lambda = \sum_{p\in \Lambda} \sf{N}_p,\qquad \sf{N}_p = \sf{a}_{(p,+)}^*\sf{a}_{(p,+)} + \sf{a}_{(p,-)}^* \sf{a}_{(p,-)}\,,
\]
with respect to locally interacting quasi-free states $\varrho_\nu = \langle\Uppsi_\nu,\cdot\,\Uppsi_\nu\rangle$ as defined in Definition~\ref{def:bogostate}, where, for any $\nu\in\ell_\bb{R}^2(\Lambda)$,
\[
	\Uppsi_\nu = \exp(\sf{A}_\nu)\Upomega,\qquad \sf{A}_\nu = \sum_{p\in\Lambda} \nu_p \big[ \sf{a}_{(p,+)}^* \sf{a}_{(p,-)}^* - \sf{a}_{(p,+)} \sf{a}_{(p,-)}\big].
\]

    The proof of Theorem~\ref{thm:N+_bogo} is based on the observation that since $\sf{A}_\nu$ is a sum of terms 
    \[
    	\sf{A}_{\nu_p}=\nu_p \big[ \sf{a}_{(p,+)}^* \sf{a}_{(p,-)}^* - \sf{a}_{(p,+)} \sf{a}_{(p,-)}\big],\qquad p\in\Lambda,
    \]
     involving only the creation and annihilation operators at the site $p\in\Lambda$, for which $\sf{A}_{\nu_p}$ and $\sf{A}_{\nu_q}$ commute when $p\ne q$, one can express the locally interacting quasi-free vector as
\[
	\Uppsi_\nu = \prod_{p \in \Lambda} \exp(\sf{A}_{\nu_p})\Upomega = \bigotimes_{p\in\Lambda} \exp(\nu_p \sf{A}_\mathrm{loc})\Upomega_p\,,\qquad \sf{A}_\mathrm{loc} := \sf{a}_{+}^*\sf{a}_{-}^* - \sf{a}_{+}\sf{a}_{-},
\]
where $\Upomega_p$ is the vacuum vector in $\fk{h}_p$ for each $p\in\Lambda$ and $\sf{a}_{\pm}$ ($\sf{a}_{\pm}^*$) are the annihilation (creation) operators on the single spin-particle Fock space $\fk{h}_p=\cl{F}_{\sym}(\bb{S})$.

Moreover, since the $\sf{N}_p$ and $\exp(\sf{A}_{\nu_q})$ commute for $p\ne q$, we then easily find that
\begin{align}\label{eq:tensorization}
	\begin{aligned}
	\bb{P}_{\Uppsi_\nu}(\sf{N}_p\in E) &= \langle \Uppsi_\nu,\mathbbm{1}_E(\sf{N}_p)\Uppsi_\nu\rangle = \langle \Upomega,\exp(-\sf{A}_{\nu_p})\mathbbm{1}_E(\sf{N}_p)\exp(\sf{A}_{\nu_p})\Upomega\rangle \\
	&= \langle \Uppsi_{\nu_p},\mathbbm{1}_E(\sf{N}_\mathrm{loc})\Uppsi_{\nu_p}\rangle_{\fk{h}_p},\qquad \Uppsi_{\nu_p} := \exp(\nu_p \sf{A}_\mathrm{loc})\Upomega_p\in\fk{h}_p,
	\end{aligned}
\end{align}
for any Borel set $E\in \cl{B}(\bb{R})$, 
where $\sf{N}_\mathrm{loc} = \sf{a}_{+}^*\sf{a}_{+} + \sf{a}_{-}^* \sf{a}_{-}$ is the number operator and $\Upomega_p$ is the vacuum vector on the single spin-particle Fock space $\fk{h}_p$.

Using a disentanglement formula for $\exp(\nu_p \sf{A}_\mathrm{loc})$ that we recall and prove in Section \ref{subsec:disentanglement}, we derive an explicit form of the single spin-particle Fock vector $\Uppsi_{\nu_p}$. This allows us to determine the distribution of $\sf{N}_\mathrm{loc}$ in the quasi-free vector $\Uppsi_{\nu_p}$. Together with \eqref{eq:tensorization}, we obtain the distribution of $\sf{N}_p$ in $\varrho_\nu$ and show the independence of $\sf{N}_p$ and $\sf{N_q}$ for $p\ne q$ in Section~\ref{sec:Np}. Finally, we use these results to prove Theorem \ref{thm:N+_bogo}.

\subsection{Disentanglement Formula}
\label{subsec:disentanglement}

    In this section, we prove a disentanglement formula as a preliminary result for the proof of Theorem \ref{thm:N+_bogo}. 
    Even though the formula is well known to experts in the field, where it appears as the \emph{normal form} of the displacement operator of $\fk{su}(1,1)$ (see \cite[\S5]{Perelomov})\footnote{We note that \cite[Eq.~(5.2.7)]{Perelomov} contains a minor typographical error: the coefficient of $\sf{K}_0$ there has the opposite sign from the exponential formula in Proposition~\ref{prop:disentanglement_formula}---the correct sign is the one that ensures the unitarity of the left-hand side.}, we shall recall it for the sake of completeness. We formulate and prove the disentanglement formula in a more general setting in Proposition~\ref{prop:disentanglement_formula} and apply it in Lemma~\ref{cor:bogo-dis} to the quasi-free states $\Uppsi_\nu$ of the form \eqref{eq:bogostate}. 
 
    \begin{proposition}
    \label{prop:disentanglement_formula}
        Let $\cl{H}$ be a separable Hilbert space, $\sf{K}_\pm,\sf{K}_0$ be three closed operators, admitting a common dense subspace of analytic vectors, such that $\sf{K}_-^* = \sf{K}_+$, and satisfying the $\fk{su}(1,1)$ commutation relations
        \[ 
            [\sf{K}_+,\sf{K}_-] = -2 \sf{K}_0 \qquad [\sf{K}_0,\sf{K}_\pm] = \pm \sf{K}_\pm \; 
        \]
        on this common dense subspace. Then for any $\nu \in \bb{R}$, we have
        \[
            \exp(\nu(\sf{K}_+ - \sf{K}_-)) = \exp(\tanh(\nu) \sf{K}_+) \exp(-2 \ln \cosh(\nu) \sf{K}_0) \exp(-\tanh(\nu) \sf{K}_-) \; . 
        \]
    \end{proposition}

    \begin{proof}
        Using the commutation relations, we obtain the following on a common stable core 
        \begin{align}\label{prop:ccr:disentanglement}
        \begin{aligned}
            e^{t \sf{K}_+} \sf{K}_0 e^{-t \sf{K}_+} &= \sf{K}_0 - t \sf{K}_+ \\ 
            e^{t \sf{K}_0} \sf{K}_- e^{-t \sf{K}_0} &= e^{-t} \sf{K}_- \\
            e^{t \sf{K}_+} \sf{K}_- e^{-t \sf{K}_+} &= \sf{K}_- - 2 t \sf{K}_0 + t^2 \sf{K}_+\,.
        \end{aligned} 
        \end{align}
        Let $f,h$ be two $C^1$-functions, satisfying $f(0) = h(0) = 0$. Consider the unitary operator
        \[
            U(\nu) := e^{f(\nu) \sf{K}_+} e^{h(\nu) \sf{K}_0} e^{-f(\nu) \sf{K}_-} \; .  
        \]
        Our goal is to choose appropriate functions $f$ and $h$ such that $\dot{U} = (\sf{K}_+ - \sf{K}_-) U$, which will imply that $U(\nu) = e^{\nu(\sf{K}_+-\sf{K}_-)}$, since $U(0) = I$. To do so, we compute
        \[
            \dot{U} = \biggl ( f' \sf{K}_+ + \underset{(1)}{\underbrace{h' e^{f \sf{K}_+} \sf{K}_0 e^{-f \sf{K}_+}}} - \underset{(2)}{\underbrace{f' e^{f \sf{K}_+} e^{h \sf{K}_0} \sf{K}_- e^{-h \sf{K}_0} e^{-f \sf{K}_+}}} \biggr) U\,  . 
        \]
        Using the relations \eqref{prop:ccr:disentanglement}, we have that
        \begin{align*}
            (1) &= h' (\sf{K}_0 - f \sf{K}_+)\,, \\
            (2) &= e^{-h} f' [\sf{K}_- - 2 f \sf{K}_0 + f^2 \sf{K}_+]\,,
        \end{align*}
        and therefore,
        \[
            \dot{U} = \Bigl([f' - h' f - e^{-h} f'f^2] \sf{K}_+ - e^{-h} f' \sf{K}_- + [h' + 2 f f' e^{-h}] \sf{K}_0\Bigr)U.
        \]
     	We obtain the system of ordinary differential equations 
        \begin{align*}
        	f' - h' f - e^{-h} f'f^2 = 1\,,\qquad e^{-h} f' = 1\,,\qquad h' + 2 f f' e^{-h} &= 0\,.
        \end{align*}
        Combining the second and third equations gives $h' = -2 f$, together with the first equation, this yields $f' + f^2 = 1$, hence $f = \tanh(\nu)$, and $h = -2 \ln \cosh(\nu)$. One can then check that this is a genuine solution.   
    \end{proof}

    We now apply Proposition~\ref{prop:disentanglement_formula} to a locally interacting quasi-free vector to obtain

    \begin{lemma}\label{cor:bogo-dis}
        Let $\nu\in\bb{R}$ and $\Uppsi_{\nu}=\exp(\nu\sf{A}_\mathrm{loc})\Upomega\in \cl{F}_{\sym}(\bb{S})$ be a locally interacting quasi-free vector on the single spin-particle Fock space. Then, $\Uppsi_{\nu}$ takes the explicit form
        \begin{align}\label{eq:disentanglement-vector}
            \Uppsi_\nu = \sech(\nu) \sum_{n\ge 0} \tanh^n(\nu) \ket{n,n},
        \end{align}
        where $(\ket{n,m})_{n,m \geq 0}$ is the eigenbasis of the number operator $\sf{N}_\mathrm{loc}$ on $\cl{F}_{\sym}(\bb{S})$.
    \end{lemma}
    \begin{proof}
        Consider the operators $\sf{K}_+ = \sf{a}_+^* \sf{a}_-^*$, $\sf{K}_- = \sf{a}_+ \sf{a}_-$ and $\sf{K}_0 = \frac{1}{2}(\sf{N}_\mathrm{loc} + 1)$. Since these operators satisfy the $\fk{su}(1,1)$ commutation relations, we can use Proposition~\ref{prop:disentanglement_formula}, which gives that
        \begin{align*}
            \Uppsi_\nu &= \exp(\nu(\sf{K}_+ - \sf{K}_-)) \Upomega \\
            &= \exp(\tanh(\nu) \sf{K}_+) \exp(-2 \ln(\cosh(\nu)) \sf{K}_0) \exp(-\tanh(\nu) \sf{K}_-) \Upomega\,.
        \end{align*}
        Since $\sf{K}_-$ only annihilates particles, $\sf{K}_- \Upomega = 0$, and we get that $\exp(-\tanh(\nu) \sf{K}_-) \Upomega = \Upomega$. Similarly, we have $\sf{K}_0 \Upomega = \frac{1}{2} \Upomega$ since $\sf{N} \Upomega = 0$, as there are no particles in the vacuum state. Hence, we have $\exp(-2 \ln(\cosh(\nu)) \sf{K}_0) \Upomega = e^{-\ln(\cosh(\nu))} \Upomega = \sech(\nu) \Upomega$. Consequently, we arrive at
        \[
            \Uppsi_\nu = \sech(\nu) \exp(\tanh(\nu) \sf{K}_+) \Upomega.
        \]
        We now expand the exponential. We have, for $\tau := \tanh(\nu)$,
        \begin{align*}
            \exp(\tau\sf{K}_+) \Upomega &= \left ( \sum_{n \ge 0}\frac{\tau^n}{n!} \sf{K}_+^n \right ) \Upomega
            = \left ( \sum_{n \ge 0} \frac{\tau^n}{n!} (\sf{a}_+^*)^n (\sf{a}_-^*)^n \right ) \ket{0,0} 
            = \sum_{n \ge 0} \tau^n \ket{n,n},
        \end{align*}
        where we used the fact that for all $n,m \in \bb{N}_0$ $\sf{a}_+^* \ket{n,m} = \sqrt{n+1} \ket{n+1,m}$ and $\sf{a}_-^* \ket{n,m} = \sqrt{m+1} \ket{n,m+1}$. Hence, we conclude that $\Uppsi_\nu$ takes the asserted form.
     \end{proof}

\subsection{Properties of $\sf{N}_p$ }
\label{sec:Np}

In this section, we prove that the local number operators $\mathsf{N}_p$ in a locally interacting quasi-free vector $\Uppsi_\nu$ from Definition \ref{def:bogostate} are mutually independent and geometrically distributed random variables for every $p \in \Lambda$.

\begin{lemma}
\label{lemma:Np}
	Let $\nu \in \ell^2(\Lambda)$ and $\Uppsi_\nu$ be the corresponding locally interacting quasi-free vector given by \eqref{eq:bogostate}. Then the local number operators $\sf{N}_p$ in the vector $\Uppsi_\nu$
	\begin{enumerate}[(i)]
		\item share the same distribution as the random variable $2G_p$, where $G_p\sim \mathrm{Geo}(\sech^2(\nu_p))$ is a geometric random variables for every $p\in\Lambda$, and are
		\item mutually independent.
	\end{enumerate}
\end{lemma}
\begin{proof}
We start by deriving the distribution of $\sf{N}_p$ in $\varrho_\nu$ for any fixed $p\in\Lambda$. Recall from \eqref{eq:tensorization} that
\[
	\bb{P}_{\Uppsi_\nu}(\sf{N}_p = k) = \langle \Uppsi_{\nu_p},\mathbbm{1}_k(\sf{N}_\mathrm{loc})\Uppsi_{\nu_p}\rangle,\qquad \Uppsi_{\nu_p} := \exp(\nu_p \sf{A}_\mathrm{loc})\Upomega_p\in\fk{h}_p.
\]
Using the disentanglement formula \eqref{eq:disentanglement-vector} and expressing $\sf{N}_\mathrm{loc}$ in its eigenbasis $(\ket{j,l})_{j,l\ge 0}$, we then obtain for any $k\in\bb{N}_0$,
\begin{align*}
	\langle \Uppsi_{\nu_p},\mathbbm{1}_k(\sf{N}_\mathrm{loc})\Uppsi_{\nu_p}\rangle &= \sech^2(\nu_p)\sum_{n,m\ge 0} \tanh^{m+n}(\nu_p)\bra{m,m} \mathbbm{1}_{k}(\sf{N}_\mathrm{loc})\ket{n,n} \\
	&= \sech^2(\nu_p)\sum_{n,m\ge 0} \sum_{j+l = k} \tanh^{m+n}(\nu_p)\bra{m,m}\ket{j,l}\bra{j,l}\ket{n,n} \\
	&=\sech^2(\nu_p)\sum_{2n = k} \tanh^{2n}(\nu_p) 
	= \begin{cases}\sech^2(\nu_p)\tanh^k(\nu_p) & \text{for $k$ even,} \\
 	0 &\text{for $k$ odd.}	
 \end{cases}
\end{align*}
Since $G_p\sim \text{Geo}(\sech^2(\nu_p))$ satisfies $\bb{P}(G_p=k) = \sech^2(\nu_p)\tanh^{2k}(\nu_p)$ for every $k\in\bb{N}_0$, we then conclude that $\mathbb{P}_{\Uppsi_\nu}( \mathsf{N}_p = k ) = \bb{P}(2G_p=k)$ as asserted. This single-site law is the one obtained, in the language of coherent states for the discrete series of $\fk{su}(1,1)$, in \cite[\S5]{Perelomov}.

Next, we prove that the random variables $(\mathsf{N}_p)_{p \in \Lambda}$ are mutually independent. To do so, it suffices to show that any finite family is independent. For the clarity, we only do it for two sites, but the argument extends immediately to a finite number of sites. Let us show that $\mathsf{N}_p$ and $\mathsf{N}_q$ are independent for $p \not= q$. For this, we write 
\[
	\mathbb{P}_{\Uppsi_\nu} \big( \mathsf{N}_p = k, \mathsf{N}_q = \ell \big) = \langle \Uppsi_\nu, \mathbbm{1}_{k}(\sf{N}_p) \mathbbm{1}_{\ell}(\sf{N}_q) \Uppsi_\nu \rangle  \; . 
\]
With similar arguments as before, we obtain
\begin{align*}
	\mathbb{P}_{\Uppsi_\nu} \big( \mathsf{N}_p = k, \mathsf{N}_q = \ell \big) &= \langle \Upomega, e^{-\sf{A}_{\nu_p}} \mathbbm{1}_{k}(\sf{N}_p)  e^{\sf{A}_{\nu_p}} e^{-\sf{A}_{\nu_q}} \mathbbm{1}_{\ell}(\sf{N}_q) e^{\sf{A}_{\nu_q}}\Upomega  \rangle \\
	&= \langle \Upomega_p, e^{- \nu_p \sf{A}_\mathrm{loc}} \mathbbm{1}_{k}(\sf{N}_\mathrm{loc}) e^{\nu_p \sf{A}_\mathrm{loc}} \Upomega_p\rangle \langle \Upomega_q,  e^{-\nu_q \sf{A}_\mathrm{loc}} \mathbbm{1}_{\ell}(\sf{N}_\mathrm{loc}) e^{\nu_q \sf{A}_\mathrm{loc}}\Upomega_q  \rangle.
\end{align*}
In particular, we find that
\[
	\bb{P}_{\Uppsi_\nu} ( \sf{N}_p = k, \sf{N}_q = l ) = \mathbb{P}_{\Uppsi_\nu} ( \sf{N}_p = k )\,\bb{P}_{\Uppsi_\nu} (\sf{N}_q = l ), 
\]
i.e., the random variables $\sf{N}_p$ and $\sf{N}_q$ are independent.
\end{proof}

We are now ready to prove Theorem~\ref{thm:N+_bogo}.

\begin{proof}[Proof of Theorem \ref{thm:N+_bogo}]  
Recall the definitions
\[
    \mathsf{N}_\Lambda = \sum_{p \in \Lambda} \sfN_p\,,\qquad G_\Lambda = \sum_{p\in\Lambda} G_p,\qquad G_p\sim\mathrm{Geo}(\sech^2(\nu_p)).
\]
Then, using the results obtained earlier, we obtain
\begin{align*}
	\bb{E}_{\Uppsi_\nu}[e^{it\sf{N}_\Lambda}] &= \bb{E}_{\Uppsi_\nu}\left[\prod_{p\in\Lambda}e^{it\sf{N}_p}\right] \qquad \text{since $\sf{N}_p$ and $\sf{N}_q$ commute for $p\ne q$,}\\
	&= \prod_{p\in\Lambda}\bb{E}_{\Uppsi_\nu}[e^{it\sf{N}_p}] \qquad \text{due to the independence of $\sf{N}_p$ and $\sf{N}_q$ for $p\ne q$,}\\
	&= \prod_{p\in\Lambda} \langle\Uppsi_{\nu_p},e^{it\sf{N}_p}\Uppsi_{\nu_p}\rangle_{\fk{h}_p} = \prod_{p\in\Lambda} \bb{E}_p[e^{it(2G_p)}]\qquad\text{due to Lemma~\ref{lemma:Np}},\\
	&= \bb{E}\left[\prod_{p\in\Lambda} e^{it(2G_p)} \right] = \bb{E}[e^{it(2G_\Lambda)}]\qquad\text{with $\bb{P}:= \bigotimes_{p\in\Lambda}\bb{P}_p$}.
\end{align*}
Since the characteristic functions coincide, we conclude that the distribution of $\sfN_\Lambda$ in the vector $\Uppsi_\nu$ coincides with the distribution of $2G_\Lambda$, as asserted.
\end{proof}

\section{Asymptotic limits and tail behavior}
\label{sec:ldp-proofs}
    Section~\ref{sec:proof-intermediate} concerns an arbitrary index set $\Lambda$. From Section~\ref{sec:proof-sparse} onwards, we specialize to the lattice $\Lambda^\kappa = \kappa\bb{Z}^d_0$ and drop the subscript $\Lambda$ for notational convenience. We recall the mean and variance of $G^\kappa$, given by
\[
	\mu^\kappa = \sum_{p\in\Lambda^\kappa} \sinh^2(\nu_p),\qquad (\sigma^\kappa)^2 = \sum_{p\in\Lambda^\kappa} \cosh^2(\nu_p)\sinh^2(\nu_p),
\]
such that $\bb{E}_{\Uppsi_\nu^\kappa}[\sf{N}^\kappa] = 2\mu^\kappa$ and $\bb{V}_{\Uppsi_\nu^\kappa}[\sf{N}^\kappa] = 4(\sigma^\kappa)^2$. We further recall the largest of the individual means,
\[
	\overline{\mu}^\kappa := \sup_{p\in\Lambda^\kappa} \sinh^2(\nu_p).
\]

\subsection{Intermediate regime}\label{sec:proof-intermediate}

The goal of this section is to prove Theorem~\ref{thm:ldp-intermediate}, that is, to derive tail estimates for an infinite sum of independent geometric random variables. Throughout, we abbreviate $\mu := \mu_\Lambda$ and $\overline{\mu} := \overline{\mu}_\Lambda$, and we write
\[
	\overline{\tau}^2 := \frac{\overline{\mu}}{1+\overline{\mu}} \in (0,1)
\]
for the ratio of the geometric law with the largest mean.

\begin{proof}[Proof of Theorem~\ref{thm:ldp-intermediate}]
	Since $\nu\in\ell^2(\Lambda)$ and $\nu\not\equiv 0$, the supremum defining $\overline{\mu}$ is attained at some $\overline{p}\in\Lambda$, and $\overline{\mu}>0$.

	\medskip\noindent\emph{Lower bound.} Keeping only the site $\overline{p}$ and using Theorem~\ref{thm:N+_bogo}, we deduce
	\begin{equation}\label{eq:tail-lower-bound}
		\bb{P}_{\Uppsi_\nu}(\sf{N}_\Lambda > 2\lambda \mu) = \bb{P}(G_\Lambda > \lambda \mu) \ge \bb{P}(G_{\overline{p}} > \lambda \mu) = \overline{\tau}^{\,2(\lfloor \lambda \mu\rfloor+1)},
	\end{equation}
	where $\lfloor r \rfloor$ denotes the largest integer smaller than $r\in\bb{R}_+$. Taking logarithms and dividing by $\lambda$ therefore gives
	\[
		\liminf_{\lambda\to \infty} \frac{1}{\lambda}\ln\bb{P}_{\Uppsi_\nu}(\sf{N}_\Lambda > 2\lambda \mu) \ge \mu \ln \overline{\tau}^{\,2} = -\mu\ln\Bigl(1 + \frac{1}{\overline{\mu}}\Bigr) = -\cl{I}_\Lambda.
	\]

	\medskip\noindent\emph{Upper bound.} Chernoff bound yields for any $\lambda > 0$
	\[
		\bb{P}_{\Uppsi_\nu}(\sf{N}_\Lambda > 2\lambda \mu) = \bb{P}(G_\Lambda > \lambda \mu) \le e^{-\lambda t\mu}\,\bb{E}[e^{tG_\Lambda}]
	\]
	for every $t \in [0,-\ln \overline{\tau}^2) = \cl{T}$, hence
	\[
		\ln \bb{P}_{\Uppsi_\nu}(\sf{N}_\Lambda > 2\lambda \mu) \le -\lambda t\mu + \ln \bb{E}[e^{tG_\Lambda}]
		= -\lambda t\mu - \sum_{p\in\Lambda} \ln \bigl(1 - \mu_p(e^t-1)\bigr),
	\]
	with $\mu_p=\sinh^2(\nu_p)$ the mean of $G_p$. Since $r\mapsto -\ln(1-r)$ is convex and vanishes at $r=0$, and since $\mu_p/\overline{\mu}\in[0,1]$, we have
	\[
		-\ln \Bigl(1 - \tfrac{\mu_p}{\overline{\mu}}(e^{t}-1) \overline{\mu} \Bigr) \le - \frac{\mu_p}{\overline{\mu}} \ln \bigl(1 - (e^{t}-1) \overline{\mu} \bigr),
	\]
	and therefore, summing over $p\in\Lambda$,
	\[
		\ln \bb{P}_{\Uppsi_\nu}(\sf{N}_\Lambda > 2\lambda \mu) \le -t\lambda\mu - \frac{\mu}{\overline{\mu}} \ln \bigl(1 - (e^{t}-1) \overline{\mu} \bigr).
	\]
	The right-hand side is minimized over $t\in\bb{R}$ at
	\[
		t_\lambda = \ln\frac{\lambda(1+\overline{\mu})}{1+ \lambda \overline{\mu}}, \quad \text{which belongs to $\cl{T}$ for all $\lambda \geq 1$.}
	\]
    For this value we have $1-(e^{t_\lambda}-1)\overline{\mu} = (1+\overline{\mu})/(1+\lambda\overline{\mu})$. Substituting gives the explicit bound
	\begin{equation}\label{eq:tail-upper-bound}
		\ln \bb{P}_{\Uppsi_\nu}(\sf{N}_\Lambda > 2\lambda \mu) \le -\lambda\mu \ln\frac{\lambda(1+\overline{\mu})}{1+\lambda\overline{\mu}} - \frac{\mu}{\overline{\mu}}\ln\frac{1+\overline{\mu}}{1+\lambda\overline{\mu}},
	\end{equation}
	valid for every $\lambda>0$. Dividing by $\lambda$ and letting $\lambda\to+\infty$, the first term converges to $-\mu\ln(1+1/\overline{\mu})$ and the second vanishes, so that
	\[
		\limsup_{\lambda\to \infty}\frac{1}{\lambda} \ln \bb{P}_{\Uppsi_\nu}(\sf{N}_\Lambda > 2\lambda \mu) \le -\mu\ln\Bigl(1+\frac{1}{\overline{\mu}}\Bigr) = -\cl{I}_\Lambda.
	\]
	Together, the two bounds yield the assertion.
\end{proof}

\subsection{Sparse regime, $\kappa\to+\infty$}\label{sec:proof-sparse}
Recall the log-Laplace transform for the number operator $\sf{N}_\Lambda^\kappa$ in the quasi-free vector $\Uppsi_\nu^\kappa$, given by
\[
	\mathscr{L}_\kappa(t) = \begin{cases}
		\displaystyle -\sum_{p\in \Lambda^\kappa} \ln\left(1 - \mu_p^\kappa(e^{2t}-1)\right) &\text{for $t< \tfrac{1}{2}\ln\bigl(1+1/\overline{\mu}^\kappa\bigr)$},\\
		+\infty & \text{otherwise}.
	\end{cases}
\]

The proof of Theorem~\ref{thm:poisson-sparse}\emph{(ii)} and \emph{(iii)} combines a direct computation of the limit in Theorem~\ref{thm:poisson-sparse}\emph{(i)} with the following general proposition, which converts a scaled log-Laplace limit of this form into the leading-order law and the exponential moments of an integer-valued random variable. Its proof, which is elementary but somewhat lengthy, is given in Appendix~\ref{app:small-numbers}.

\begin{proposition}[Law of small numbers from a scaled log-Laplace limit] \label{prop:small-numbers}
    Let $(X_N)_{N \geq 0}$ be a sequence of $\bb{N}_0$-valued random variables, and let $(r_N)_{N\ge 0}$ be a sequence of positive real numbers converging to $0$ as $N \to +\infty$. Suppose that, for some $\fk{c} > 0$ and some integer $K \geq 1$,
    \begin{equation}\label{eq:prop-laplace-hypothesis}
        \lim_{N \to +\infty} r_N^{-1} \ln \bb{E}[e^{t X_N}] = \fk{c}\,(e^{K t} - 1) \qquad\text{for every $t \in \bb{R}$}.
    \end{equation}
    Then $X_N$ is, to leading order in $r_N$, distributed as $K$ times a Bernoulli random variable of parameter $\fk{c}\,r_N$; that is,
    \begin{equation}\label{eq:prop-two-point}
        \bb{P}(X_N = k) = o(r_N) + \begin{cases} 1 - \fk{c}\, r_N & \text{for $k = 0$}, \\ \fk{c}\, r_N & \text{for $k = K$}, \\ 0 & \text{otherwise.} \end{cases}
    \end{equation}
    Moreover, for every $\theta > 0$ and every $\lambda \in \bb{C}$,
    \begin{equation}\label{eq:prop-moments}
        \lim_{N \to +\infty} \frac{1}{r_N}\, \bb{E}\bigl[X_N^{\theta}\, e^{\lambda X_N}\bigr] = \fk{c}\, K^{\theta} e^{K \lambda}.
    \end{equation}
\end{proposition}

We are now ready to prove Theorem~\ref{thm:poisson-sparse}.

\begin{proof}[Proof of Theorem~\ref{thm:poisson-sparse}]
	\emph{(i)} Let $t\in\bb{R}$ be arbitrary, but fixed. Further, let $\varepsilon>0$. By condition \eqref{eq:sparse-condition}, we find some large $\kappa_\varepsilon>0$ such that for every $\kappa>\kappa_\varepsilon$, the following estimate holds:
\begin{align}\label{eq:proof-sparse-estimate}
	\frac{\fk{c}_\nu - \varepsilon}{|p|^\gamma}\le \kappa^\gamma  \sinh^2(\nu(\kappa p)) \le \frac{\fk{c}_\nu + \varepsilon}{|p|^\gamma}\qquad\text{for every $p\in\bb{Z}^d_0$.}
\end{align}
In particular, 
\[
	m^\kappa := \sup_{p\in\bb{Z}^d\backslash\{0\}} \sinh^2(\nu(\kappa p)) \to 0\quad\text{as $\kappa\to +\infty$}.
\]
Now let $\kappa> \kappa_\varepsilon$ be sufficiently large such that $m^\kappa|e^{2t}-1|<1$. Using the elementary estimates for every $r < 1$, and the fact that $|\mu_p^\kappa(e^{2t}-1)| \le m^\kappa|e^{2t}-1|$ (which implies that $\mu_p^\kappa(e^{2t}-1) < 1$) uniformly in $p\in\mathbb{Z}^d_0$, we obtain, for $c_t = 1$ if $t \geq 0$, and $-1$ else:
\[
	\mu_p^\kappa(e^{2t}-1) \le -\ln\left(1 - \mu_p^\kappa(e^{2t}-1)\right) \le \frac{\mu_p^\kappa(e^{2t}-1)}{1 - c_t m^\kappa |e^{2t}-1|}\qquad\text{for every $p\in\bb{Z}^d_0$.}
\]
(the appearance of $c_t$ comes from treating the two cases $e^{2t}-1 \geq 0$ or $\leq 0$). Summing up over $p\in\bb{Z}^d\backslash\{0\}$, multiplying by $\kappa^\gamma$, and using \eqref{eq:proof-sparse-estimate} then gives
\[
	(e^{2t}-1)\sum_{p\in\bb{Z}^d\backslash\{0\}} \frac{\fk{c}_\nu - c_t \varepsilon}{|p|^\gamma} \le \kappa^\gamma\mathscr{L}_\kappa(t)\le \frac{(e^{2t}-1)}{1-c_t m^\kappa |e^{2t}-1|} \sum_{p\in\bb{Z}^d\backslash\{0\}} \frac{\fk{c}_\nu + c_t \varepsilon}{|p|^\gamma}.
\]
Passing first $\kappa\to+\infty$ and then $\varepsilon\to 0$, we then obtain the limit
\[
	\lim_{\kappa\to +\infty} \kappa^\gamma\mathscr{L}_\kappa(t) = \fk{c}_\nu(e^{2t}-1)\sum_{p\in\bb{Z}^d\backslash\{0\}} \frac{1}{|p|^\gamma} = \fk{c}_\nu\zeta_d(\gamma)(e^{2t}-1),
\]
which is $\fk{m}_\nu(e^{2t}-1)$, as asserted.

\medskip
\noindent\emph{(ii)} and \emph{(iii)} Consider a sequence $\kappa_N \to +\infty$, and let $X_N$ have the law of $\sf{N}^{\kappa_N}$ in the quasi-free vector $\Uppsi_{\nu}^{\kappa_N}$. Then $X_N$ is a non-negative integer-valued random variable, and part~(i) gives
    \[
        \kappa_N^\gamma \ln \bb{E}[e^{t X_N}] = \kappa_N^\gamma\, \mathscr{L}_{\kappa_N}(t) \longrightarrow \fk{m}_\nu (e^{2t} - 1)\qquad\text{for every $t\in\bb{R}$}.
    \]
    The assertions then follow immediately from Proposition \ref{prop:small-numbers} applied with $\fk{c} = \fk{m}_\nu$, $K = 2$ and $r_N = \kappa_N^{-\gamma}$.
\end{proof}

The proof of Theorem~\ref{thm:log-ldp} relies on three elementary lemmas: uniform two-sided bounds on every site, a Chernoff bound for the upper tail with the optimal polynomial order and a matching one-site lower bound. Those will be linked to the large-deviation principle through the following Proposition.

\begin{proposition}[From two-sided bounds to large-deviation principle] \label{prop: two_sided_to_LDP}
    Let $(X^\kappa)_{\kappa > 0}$ be a family of random variables valued in $\bb{N}_0$ such that, uniformly in $k \geq 0$, as $\kappa \to +\infty$ (with $\gamma > 0$)
    \[
        \ln \bb{P}(X^\kappa = k) = -\gamma k \ln \kappa + (k+1) O(1).
    \]
    Then $(X^\kappa)_{\kappa > 0}$ satisfies a LDP as $\kappa \to +\infty$ with speed $\ln \kappa$ and good-rate function
    \[
        \mathscr{J}_\gamma(x) = \begin{cases}
            \gamma x & \text{if $x \in \bb{N}_0$}, \\ +\infty & \text{else}. 
        \end{cases}
    \]
    That is, for all $E \in \cl{B}(\mathbb{R})$ 
    \[
        -\inf_{\overset{\circ}{E}} \mathscr{J}_\gamma \leq \liminf_{\kappa \to +\infty} \frac{1}{\ln \kappa} \ln \mathbb{P}(X^\kappa \in E) \leq \limsup_{\kappa \to +\infty} \frac{1}{\ln \kappa} \ln \mathbb{P}(X^\kappa \in E) \leq -\inf_{\overline{E}} \mathscr{J}_\gamma.
    \]
\end{proposition}

\begin{proof}
    Passing to the exponential and summing up the estimates over $k \geq n$, we easily see that, uniformly in $k \geq 0$, as $\kappa \to +\infty$
    \[
        \ln \bb{P}(X^\kappa \geq k) = -\gamma k \ln(\kappa) + (k+1) O(1).
    \]
    The function $\mathscr{J}_\gamma$ is a good rate function: it is lower semicontinuous, and its sublevel sets $\{ \mathscr{J}_\gamma \leq M \} = \bb{N}_0 \cap [0,M/\gamma]$ are finite, hence compact. We let $E \in \cl{B}(\mathbb{R})$.
    
    \medskip \noindent \emph{Upper bound.} If $\overline{E} \cap \bb{N}_0 = \emptyset$, then $\bb{P}(X^\kappa \in E) = 0$ for all $\kappa$ (as $X^\kappa \in \bb{N}_0$ almost surely) and the bound holds trivially. Otherwise, let $k_E := \min \{ \overline{E} \cap \bb{N}_0 \}$, so that $\inf_{\overline{E}} \mathscr{J}_\gamma = \gamma k_E$ and $\overline{E} \cap \bb{N}_0 \subset [k_E,+\infty)$. Then for $\kappa \gg 1$
    \[
        \frac{1}{\ln \kappa} \ln \mathbb{P}(X^\kappa \in E) \leq \frac{1}{\ln \kappa} \ln \mathbb{P}(X^\kappa \geq k_E) = -\gamma k_E + \frac{(k_E+1) O(1)}{\ln \kappa} = -\gamma k_E + o(1). 
    \]
    This proves the upper bound.
    
    \medskip \noindent \emph{Lower bound.} If $\overset{\circ}{E} \cap \mathbb{N}_0 = \emptyset$ then the bound is trivial. Otherwise, for all $k \in \overset{\circ}{E} \cap \mathbb{N}_0$ we have for $\kappa \gg 1$
    \[
        \frac{1}{\ln \kappa} \ln \mathbb{P}(X^\kappa \in E) \geq \frac{1}{\ln \kappa} \ln \mathbb{P}(X^\kappa = k) = -\gamma k + o(1),
    \]
    Hence
    \[
        \liminf_{\kappa \to +\infty} \frac{1}{\ln \kappa} \ln \mathbb{P}(X^\kappa \in E) \geq -\gamma k.
    \]
    Optimizing over $k$ yields the lower bound $-\inf_{\overset{\circ}{E}} \mathscr{J}_\gamma$. 
\end{proof}

Hence, to derive Theorem~\ref{thm:log-ldp}, it suffices to establish upper and lower bounds for the site probability. 

\begin{lemma}[Site bounds]\label{lem:site-bounds}
For every $\varepsilon\in(0,\frac{1}{2}]$ there exists $\kappa_\varepsilon\ge 1$ such that for all $\kappa\ge\kappa_\varepsilon$\,,
\begin{equation}\label{eq:site-bounds}
	(1-\varepsilon)\frac{\fk{c}_\nu}{|q|^{\gamma}}
	\le \kappa^\gamma  \sinh^2(\nu(\kappa q))
	\le (1+\varepsilon)\frac{\fk{c}_\nu}{|q|^{\gamma}}\qquad\text{for every $q\in \bb{Z}_0^d$.}
\end{equation}
Consequently, for all $\kappa\ge\kappa_\varepsilon$,
\[
	\overline{\mu}^\kappa \le (1+\varepsilon)\fk{c}_\nu\kappa^{-\gamma},
	\qquad
	m_\kappa \le (1+\varepsilon)\fk{c}_\nu\,\zeta_d(\gamma)\kappa^{-\gamma} < +\infty .
\]
Moreover, if $\kappa$ is additionally so large that $\overline{\mu}^\kappa\le \varepsilon$, then
\begin{equation}\label{eq:tau-mu}
	(1-\varepsilon)\,\mu_p^\kappa \le \tau_p^2 \le \mu_p^\kappa
	\qquad\text{for every $p\in \kappa \bb{Z}^d_0$}.
\end{equation}
\end{lemma}

\begin{proof}
By condition \eqref{eq:sparse-condition} there exists $R_\varepsilon\ge 1$ such that
\[
	(1-\varepsilon)\,\fk{c}_\nu \le |x|^\gamma \sinh^2(\nu(x)) \le (1+\varepsilon)\,\fk{c}_\nu
	\qquad\text{for all } |x|\ge R_\varepsilon .
\]
Set $R_\varepsilon:=\kappa_\varepsilon$. Since every $q\in\bb{Z}^d_0$ satisfies $|q|\ge 1$, every $p=\kappa q\in \kappa \bb{Z}^d_0$ with $\kappa\ge\kappa_\varepsilon$ satisfies $|p|=\kappa|q|\ge R_\varepsilon$, which then implies \eqref{eq:site-bounds}. Taking the supremum over $q\in\bb{Z}^d_0$ yields the bound on $\overline\mu_\kappa$. Similarly, summing over $q\in\bb{Z}^d_0$ and recalling that $\zeta_d(\gamma)<+\infty$ for $\gamma>d$ yields the estimate on $m_\kappa$. Finally, since $\tau_p^2 = \mu_p^\kappa/(1+\mu_p^\kappa)$, the upper bound $\tau_p^2\le\mu_p^\kappa$ holds trivially. As for the lower bound, we use the assumption $\mu_p^\kappa\le\overline{\mu}^\kappa\le\varepsilon$ to deduce 
\[
	\tau_p^2 \ge \frac{\mu_p^\kappa}{(1+\varepsilon)} \ge (1-\varepsilon)\mu_p^\kappa.\qedhere
\]
\end{proof}

\begin{lemma}[Uniform upper tail bound]\label{lem:upper-tail}
There exists $\kappa_0\ge 1$ such that for all $\kappa\ge\kappa_0$\,,
\[
	\bb{P}\big(G^\kappa \ge k\big) \le e^{\zeta_d(\gamma)}(3\fk{c}_\nu)^{k} \kappa^{-\gamma k}\qquad\text{for every $k\in\bb{N}_0$}.
\]
\end{lemma}

\begin{proof}
We begin by noting that for $\theta\ge 1$ with $\theta\sup_{p\in \kappa \bb{Z}^d_0}\tau_p^2\le \frac{1}{2}$, monotone convergence along finite sublattices and independence give
\[
	\bb{E}\big[\theta^{\,G^\kappa}\big]
	= \prod_{p\in \kappa \bb{Z}^d_0} \frac{1-\tau_p^2}{1-\theta\tau_p^2}
	\le \prod_{p\in\kappa \bb{Z}^d_0} \frac{1}{1-\theta\tau_p^2}
	\le \exp( 2\theta \sum_{p\in\kappa \bb{Z}^d_0}\tau_p^2)<+\infty,
\]
where we used the fact that $-\ln(1-x)\le x/(1-x)\le 2x$ for $x\in[0,\frac{1}{2}]$. 

We now apply Lemma~\ref{lem:site-bounds} with $\varepsilon = \frac{1}{2}$ and set $\kappa_0 := \max\{\kappa_\varepsilon,(3\fk{c}_\nu)^{1/\gamma}\}$, such that for $\kappa\ge\kappa_0$,
\[
	\overline{\mu}^\kappa \le \tfrac{3}{2}\fk{c}_\nu\kappa^{-\gamma},
	\qquad
	m_\kappa \le \tfrac{3}{2}\fk{c}_\nu\,\zeta_d(\gamma)\kappa^{-\gamma}.
\]
Setting $\theta_\kappa := \kappa^{\gamma}/(3\fk{c}_\nu)\ge 1$, we find that $\theta_\kappa\tau_p^2 \le \theta_\kappa\,\overline{\mu}^\kappa \le \frac{1}{2}$ for all $p\in\kappa \bb{Z}^d_0$. Hence, with $\theta=\theta_\kappa$ and $\sum_{p\in\kappa \bb{Z}^d_0} \tau_p^2 \le m_\kappa$, we obtain
\[
	2\,\theta_\kappa m_\kappa \le 2\cdot\frac{\kappa^\gamma}{3\fk{c}_\nu}\cdot \frac{3}{2}\,\fk{c}_\nu \zeta_d(\gamma)\,\kappa^{-\gamma} = \zeta_d(\gamma).
\]
Consequently, we obtain $\bb{E}[\theta_\kappa^{\,G^\kappa}]\le e^{\zeta_d(\gamma)}$. Finally, Markov's inequality yields
\[
	\bb{P}(G^\kappa \ge k)
	= \bb{P}(\theta_\kappa^{\,G^\kappa} \ge \theta_\kappa^{\,k})
	\le \theta_\kappa^{-k}\, \bb{E}[\theta_\kappa^{\,G^\kappa}]
	\le e^{\zeta_d(\gamma)}\, (3\fk{c}_\nu)^{k}\,\kappa^{-\gamma k},
\]
for every $k\in\bb{N}_0$ as asserted. 
\end{proof}

\begin{lemma}[Local lower bound]\label{lem:lower-local}
There exists $\kappa_1\ge 1$ such that for all $\kappa\ge\kappa_1$ and all $k\in\bb{N}_0$,
\[
	\bb{P}\big(G^\kappa = k\big) \ge e^{-1}\, \Big(\frac{\fk{c}_\nu}{4}\Big)^{k}\, \kappa^{-\gamma k}.
\]
\end{lemma}
\begin{proof}
As before, we apply Lemma~\ref{lem:site-bounds} with $\varepsilon=1/2$. We then choose $\kappa_1\ge\kappa_\varepsilon$ so large such that $\overline{\mu}^\kappa\le \frac{1}{2}$ and $m_\kappa\le \frac{1}{2}$ for all $\kappa\ge\kappa_1$. Now fix some $q_*\in\bb{Z}^d_0$ with $|q_*|=1$ and set $p_*:=\kappa q_*\in\kappa \bb{Z}^d_0$. By independence, we obtain
\[
	\bb{P}(G_\Lambda^\kappa = k)
	\ge \bb{P}(G_{p_*}=k) \prod_{p\,\ne\, p_*} \bb{P}(G_p = 0)
	= \tau_{p^*}^{2k} \prod_{p\in\kappa \bb{Z}^d_0} (1-\tau_p^2).
\]
For the first factor, \eqref{eq:site-bounds} and \eqref{eq:tau-mu} give
\[
	\tau_{p_*}^{2} \ge \tfrac{1}{2}\,\mu_{p_*}^\kappa \ge \tfrac{1}{2}\cdot\tfrac{1}{2}\,\fk{c}_\nu\kappa^{-\gamma} = \frac{\fk{c}_\nu}{4}\kappa^{-\gamma}.
\]
For the product, since $\tau_p^2\le\overline{\mu}^\kappa\le \frac{1}{2}$ we may use $1-x\ge e^{-2x}$ for $x\in[0,\frac{1}{2}]$ to estimate
\[
	\prod_{p\in\kappa \bb{Z}^d_0} (1-\tau_p^2) \ge \exp(-2\sum_{p\in\kappa \bb{Z}^d_0}\tau_p^2) \ge e^{-2m_\kappa} \ge e^{-1}. \qedhere
\]
\end{proof}

\begin{proof}[Proof of Theorem~\ref{thm:log-ldp}]
\emph{(i)} Since $\bb{P}(G^\kappa=k)\le\bb{P}(G^\kappa\ge k)$, Lemmas~\ref{lem:upper-tail} and \ref{lem:lower-local} sandwich the probability: for $\kappa\ge\max\{\kappa_0,\kappa_1\}$, we have that
\[
	-\gamma k + \frac{k\ln(\fk{c}_\nu/4) - 1}{\ln\kappa}
	\le \frac{\ln \bb{P}(G^\kappa=k)}{\ln\kappa}
	\le -\gamma k + \frac{k\ln(3\fk{c}_\nu) + \zeta_d(\gamma)}{\ln\kappa}\,,
\]
where both error terms vanish as $\kappa\to+\infty$. 

\medskip
\emph{(ii)} By point $(i)$, we can apply Proposition~\ref{prop: two_sided_to_LDP} to deduce that $(G^\kappa)_\kappa$ satisfies a LDP as $\kappa \to +\infty$ with speed $\ln \kappa$ and good rate function $\mathscr{J}_\gamma$. 

\medskip
\emph{(iii)} By Theorem~\ref{thm:N+_bogo}, $\bb{P}_{\Uppsi_\nu^\kappa}(\sf{N}^\kappa\in E)=\bb{P}(2G^\kappa\in E)$ for every $E\in\cl{B}(\bb{R})$. Since $x\mapsto 2x$ is a homeomorphism of $\bb{R}$, the contraction principle transfers the LDP of (ii) to $\sf{N}_\Lambda^\kappa$ with rate function $\mathscr{J}_{\sf{N}}(x)=\mathscr{J}_\gamma\!(x/2)$, which is again good. The local asymptotics follow from (i), and the concentration estimate from Lemma~\ref{lem:upper-tail} with $k=1$:
\[
	\bb{P}_{\Uppsi_\nu^\kappa}\big(\sf{N}^\kappa\neq 0\big) = \bb{P}\big(G^\kappa\ge 1\big) \le 3\,e^{\zeta_d(\gamma)}\fk{c}_\nu\,\kappa^{-\gamma}. \qedhere
\]
\end{proof}

\subsection{Dense regime, $\kappa\to 0$}\label{sec:dense-regime}
We consider an arbitrary sequence $\kappa_n\to 0$ for $n\to\infty$ and the family of geometric random variables
\begin{align}\label{eq:geometric-triangular-array}
	\bb{G} = \Bigl\{ G_p^{\kappa_n}\, :\, p\in \Lambda_n:=\kappa_n \bb{Z}^d_0,\, n\in\bb{N}_0\Bigr\},
\end{align}
with mean and variance
\[
	\left\{\quad\begin{aligned}
		\mu_p^n &= \sinh^2(\nu(p)) =: f_\mu(p),\\[0.2em]
		(\sigma_p^n)^2 &= \cosh^2(\nu(p))\sinh^2(\nu(p))=f_\mu(p)(1+f_\mu(p)),
	\end{aligned}\right.
\]
where $f_\mu:\bb{R}^d_0\to \bb{R}$ is a continuous function satisfying
\begin{align}\label{ass:dense-conditions-proof}\tag{$\sf{B}\nu$}
	\left\{\quad\begin{aligned}
			\limsup_{|x|\to +\infty} |x|^\gamma f_\mu(x)<+\infty &\qquad \text{for some $\gamma>d$}, \\
		\limsup_{|x|\to 0} |x|^\beta f_\mu(x)<+\infty &\qquad \text{for some $\beta\in(0,d/2)$}.
	\end{aligned}\right.
\end{align}

To prove Theorem~\ref{thm:central-limit}, we use two results. Morally, the first (cf.\ Lemma~\ref{lem:dense-mean-variance}) states that, under the assumption \eqref{ass:dense-conditions-proof}, both the ensemble mean and variance
\[
	\begin{aligned}
		\mu_n := \sum_{p\in\Lambda_n} \mu_p^n\quad\text{and}\quad \sigma_n^2 := \sum_{p\in\Lambda_n} (\sigma_p^n)^2\quad \text{scale like $\kappa^{-d}$ as $n\to\infty$,}
	\end{aligned}
\]
while the second (cf.\ Theorem~\ref{thm: CLT-geometric}) states that the sequence
\[
	S_{n} := \frac{1}{\sigma_n} \left(\,\sum_{p\in\Lambda_n}G_p^{\kappa_n} - \mu_n\right),
\]
of rescaled and recentered random variables converges in distribution to the standard normal random variable, i.e., $\rm{Law}(S_n)\rightharpoonup \cl{N}(0,1)$ as $n\to \infty$.

\medskip
More precisely, the first result states that, under assumption \eqref{ass:dense-conditions-proof} on $f_\mu$, the infinite sum of the rescaled mean and variance of $G_{\Lambda}^{\kappa_n}$ converges. Its proof is found in Appendix~\ref{app:riemann-sum} for completeness.

\begin{lemma}\label{lem:dense-mean-variance}
	Let $f_\mu:\bb{R}^d_0\to \bb{R}$ be a continuous function satisfying assumption \eqref{ass:dense-conditions-proof}. Then,
	\[
		\lim_{\kappa\to 0} \sum_{p \in \kappa \bb{Z}^d_0} \kappa^d f_\mu^\theta(p) = \int_{\bb{R}^d_0} f_\mu^\theta(x) \dd{x} =:\sf{m}_\theta <+\infty,\qquad \text{for every $\theta\in[1,d/\beta)$}.
	\]
\end{lemma}

\begin{remark}
	Notice that since $\beta\in(0,d/2)$, we have that $d/\beta > 2$. Hence, there exists some $\delta_0>0$ such that Lemma~\ref{lem:dense-mean-variance} holds for $\vartheta\in[1,2+\delta_0)$.
\end{remark}

As for the second statement, we prove in Appendix~\ref{app:clt} the following CLT result for the family of infinite sums of geometric random variables $\bb{G}$.

\begin{theorem}\label{thm: CLT-geometric}
	Let $\bb{G}$ be the family of geometric random variables defined in \eqref{eq:geometric-triangular-array} such that
	\begin{enumerate}
		\item $\{G_p^{\kappa_n}: p \in \Lambda_n\}$ is a family of independent random variables for every $n\ge 0$, 
		\item $G_p^{\kappa_n}$ has finite second moment for every $p\in\Lambda_n$ and $n\in\bb{N}_0$, and we have that
		\begin{equation*}
			0 < \sigma_n^2:=\sum_{p \in \Lambda_n} \rm{Var}(G_p^{\kappa_n}) < +\infty.
		\end{equation*}
	\end{enumerate}
	Suppose further that the Lyapunov condition holds for the sequence $S_{n,p} := \sigma_n^{-1} (G_p^{\kappa_n} - \bb{E}[G_p^{\kappa_n}])$ of rescaled and recentered random variable, i.e., there exists some $\delta>0$ such that
	\begin{equation} \label{eq: geometric-Lyapunov_condition}
		\lim_{n\to\infty}\sum_{p \in \Lambda_n} \bb{E}[|S_{n,p}|^{2+\delta}] = 0.
	\end{equation}
	Then, the sequence of distributions $\rm{Law}(S_n)$ with $S_n=\sum_{p\in\Lambda_n}S_{n,p}$ converges weakly to the standard normal distribution, along with all its $\theta$-moments for $\theta\in[1,2+\delta)$.
\end{theorem}

Using Lemma~\ref{lem:dense-mean-variance} and Theorem~\ref{thm: CLT-geometric}, we can now prove Theorem~\ref{thm:central-limit}.

\begin{proof}[Proof of Theorem~\ref{thm:central-limit}]
	From Theorem~\ref{thm:N+_bogo}, we know that $\bb{G}$ is a family of independent geometric random variables. Moreover, Lemma~\ref{lem:dense-mean-variance} implies that $0<\sigma_n^2<+\infty$ for every $n\in\bb{N}_0$. We are left to show the Lyapunov condition \eqref{eq: geometric-Lyapunov_condition} for the sequence $S_{n,p} := \sigma_n^{-1}(G_p^{\kappa_n}-\mu_p^n)$.
	
	For this, we use the power-deviation estimate derived in Proposition~\ref{prop: Lp_deviation_geom}, which gives
	\begin{align*}
		\sum_{p\in\Lambda_n}\bb{E}[|G_p^{\kappa_n}-\mu_p^n|^{2+\delta}] &\le \sf{c}_{1+\delta} \sum_{p\in\Lambda_n}\max\bigl(\mu_p^n,(\mu_p^n)^{2+\delta}\bigr) \\
		&\le \sf{c}_{1+\delta} \left(\,\sum_{p\in\Lambda_n} \mu_p^n + \sum_{p\in\Lambda_n} (\mu_p^n)^{2+\delta}\right).
	\end{align*}
	Choosing $0<\delta<\delta_0$ according to Lemma~\ref{lem:dense-mean-variance}, we then deduce
	\[
		\sum_{p\in\Lambda_n}\bb{E}[|G_p^{\kappa_n}-\mu_p^n|^{2+\delta}] = O(\kappa_n^{-d}).
	\]
	Since $\sigma_n^2 \asymp \kappa_n^{-d}$, we have that $\sigma_n^{2+\delta} \asymp \kappa_n^{-d(2+\delta)/2}$, and hence,
	\[
		\sum_{p\in\Lambda_n}\bb{E}[|S_{n,p}|^{2+\delta}] = \frac{1}{\sigma^{2 + \delta}} O(\kappa_n^{-d}) = O(\kappa_n^{d \delta/2}) \;\longrightarrow \;0\quad\text{as $n\to\infty$},
	\]
	which implies the Lyapunov condition for the sequence $(S_{n,p})$. We can then conclude the proof by applying Theorem~\ref{thm: CLT-geometric}.
\end{proof}

\section*{Acknowledgements}
Part of this research was performed while the authors were visiting the Institute for Pure and Applied Mathematics (IPAM), which is supported by the National Science Foundation (Grant Nos.\ DMS-1925919 and DMS-2422832). The authors are particularly grateful to Dima Shlyakhtenko for valuable discussions at IPAM. Moreover, we thank Andreas Deuchert for helpful discussions at IPAM and afterward. OT acknowledges support by the Netherlands Organisation for Scientific Research
(NWO) under Grant No.\ NGF.1582.22.009.
FC has been supported by funding from the European Research Council (ERC) under the European Union’s Horizon Europe research and innovation programme under ERC No. 101200514 OPTiMiSE and ERC AdG No. 101054420 EYAWKAJKOS. Views and opinions expressed are however those of the author only and do not necessarily reflect those of the European Union or the European Research Council Executive Agency. Neither the European Union nor the granting authority can be held responsible for them.
\appendix

\section{A law of small numbers for integer-valued random variables}\label{app:small-numbers}

In this section, we prove Proposition~\ref{prop:small-numbers}, which converts a scaled log-Laplace limit into the leading-order law and the exponential moments of an $\bb{N}_0$-valued random variable. It is used in the small scattering limit in Section~\ref{sec:BEC}, and again in the sparse regime in Section~\ref{sec:ldp-proofs}.

\begin{proof}[Proof of Proposition~\ref{prop:small-numbers}]
	Throughout the proof, we write
	\[
		\uppi_k^N := \frac{1}{r_N}\,\bb{P}(X_N = k),\qquad k\ge 1,
	\]
	for the normalized point masses of $X_N$ away from the origin. We begin by recording the form of the assumption \eqref{eq:prop-laplace-hypothesis} that is actually used below, namely that
	\begin{equation}\label{eq:prop-expanded}
		\lim_{N\to+\infty}\frac{1}{r_N}\,\bb{E}\bigl[e^{t X_N} - 1\bigr] = \fk{c}\bigl(e^{K t} - 1\bigr)\qquad\text{for every $t\in\bb{R}$}.
	\end{equation}
	Indeed, fix $t\in\bb{R}$. Since $r_N>0$ and the limit in \eqref{eq:prop-laplace-hypothesis} is finite, the assumption gives in particular that $\bb{E}[e^{t X_N}]<+\infty$ for all $N\gg 1$. For such $N$, the quantity $L_N(t) := \ln \bb{E}[e^{t X_N}]$ is well defined and finite, and \eqref{eq:prop-laplace-hypothesis} reads
	\[
		\frac{L_N(t)}{r_N} \longrightarrow \fk{c}\bigl(e^{Kt}-1\bigr)\qquad\text{as $N\to+\infty$}.
	\]
	Since $r_N \to 0$, this forces $L_N(t)\to 0$. Setting
	\[
		\upphi(u) := \frac{e^u-1}{u}\;\;\text{for $u\ne 0$},\qquad \upphi(0):=1,
	\]
	we obtain a continuous function on $\bb{R}$ satisfying $e^u-1 = \upphi(u)\,u$ for every $u\in\bb{R}$, the degenerate case $u=0$ included. Applying this identity with $u=L_N(t)$ gives
	\[
		\frac{1}{r_N}\,\bb{E}\bigl[e^{t X_N} - 1\bigr] = \frac{e^{L_N(t)}-1}{r_N} = \upphi\bigl(L_N(t)\bigr)\,\frac{L_N(t)}{r_N},
	\]
	where $\upphi(L_N(t)) \to \upphi(0) = 1$ by continuity. Hence, letting $N\to+\infty$ yields \eqref{eq:prop-expanded}.
	
	The proof proceeds in three steps: we first dominate $\uppi_k^N$ uniformly in $N$, then identify their limits, and finally transfer this convergence to the quantities of interest.

	\medskip\noindent\emph{Step 1: Uniform exponential domination.} We claim that for every $\lambda \ge \ln 2$, there exist a constant $C_\lambda>0$ and an index $N_\lambda\in\bb{N}$, \emph{both independent of $k$}, such that
	\begin{equation}\label{eq:prop-domination}
		\uppi_k^N \le C_\lambda e^{-\lambda k}\qquad\text{for every $k \ge 1$ and every $N \ge N_\lambda$}.
	\end{equation}
	\emph{Proof of claim.} Fix $\lambda\ge\ln 2$ and $k\ge 1$. Recall the elementary inequality $e^s \le 2(e^s-1)$, which holds when $e^s\ge 2$, i.e., when $s \ge \ln 2$. Since $\lambda k \ge \lambda \ge \ln 2$, it applies with $s = \lambda k$ and yields $e^{\lambda k}\le 2(e^{\lambda k}-1)$. Moreover, expanding the expectation over the values of $X_N$ gives
	\[
		\bb{E}\bigl[e^{\lambda X_N}-1\bigr] = \sum_{j\ge 0} \bigl(e^{\lambda j}-1\bigr)\,\bb{P}(X_N = j) \ge \bigl(e^{\lambda k}-1\bigr)\,\bb{P}(X_N = k),
	\]
	because every summand is non-negative, $\lambda$ being positive and $X_N$ being $\bb{N}_0$-valued. Combining the two inequalities and dividing by $r_N$, we obtain
	\[
		e^{\lambda k}\, \uppi_k^N \;\le\; \frac{2}{r_N}\bigl(e^{\lambda k}-1\bigr)\,\bb{P}(X_N = k) \le \frac{2}{r_N}\, \bb{E}\bigl[e^{\lambda X_N}-1\bigr].
	\]
	The right-hand side no longer involves $k$, and by \eqref{eq:prop-expanded} it converges to $2\fk{c}(e^{K\lambda}-1)$ as $N\to+\infty$. Consequently there is an index $N_\lambda$, depending on $\lambda$ but not on $k$, such that
	\[
		\frac{2}{r_N}\, \bb{E}\bigl[e^{\lambda X_N}-1\bigr] \le C_\lambda := 2\fk{c}\bigl(e^{K\lambda}-1\bigr) + 1 \qquad\text{for every $N \ge N_\lambda$}.
	\]
	Since $k\ge 1$ was arbitrary and neither $C_\lambda$ nor $N_\lambda$ depends on it, this proves \eqref{eq:prop-domination}.

	\medskip\noindent\emph{Step 2: Identification of the limits.} We claim that
	\begin{equation}\label{eq:prop-mass-limit}
		\lim_{N\to+\infty} \uppi_k^N = \fk{c}\,\updelta_{k,K}\qquad\text{for every $k\ge 1$}.
	\end{equation}
	By \eqref{eq:prop-domination} with $\lambda=\ln 2$, the sequence $\{\uppi_k^N\}_N$ is bounded for every fixed $k\ge1$. By a diagonal argument, it therefore suffices to prove that whenever $\uppi_k^{N_j}\to\uppi_k$ for every $k\ge1$ along a subsequence, the limit is $\uppi_k = \fk{c}\,\updelta_{k,K}$.

	So let $(\uppi_k)_{k\ge1}$ be such a limit, and fix $t\in\bb{R}$. Expanding the expectation over the values of $X_N$ and dividing by $r_N$, we first record the exact identity
	\begin{equation}\label{eq:prop-series-form}
		\frac{1}{r_N}\,\bb{E}\bigl[e^{t X_N}-1\bigr] = \sum_{k\ge 0} \bigl(e^{tk}-1\bigr)\frac{\bb{P}(X_N = k)}{r_N} = \sum_{k\ge1}\bigl(e^{tk}-1\bigr)\uppi_k^N,
	\end{equation}
	valid for every $N\gg 1$, where the term $k=0$ was dropped because $e^{t\cdot 0}-1 = 0$.

	We now pass to the limit along the subsequence $(N_j)_j$ on the right-hand side of \eqref{eq:prop-series-form}. Choosing $\lambda\ge\ln 2 \vee (t+1)$, the bound \eqref{eq:prop-domination} together with $|e^{tk}-1|\le e^{tk}+1$ gives
	\[
		\bigl|\bigl(e^{tk}-1\bigr)\uppi_k^N\bigr| \le C_\lambda\bigl(e^{(t-\lambda)k} + e^{-\lambda k}\bigr) \le C_\lambda\bigl(e^{-k} + 2^{-k}\bigr)\qquad\text{for every $k\ge1,\, N\ge N_\lambda$},
	\]
	where the right-hand side is summable in $k$ and independent of $N$. Since moreover $\uppi_k^{N_j}\to\uppi_k$ as $j\to\infty$ for each fixed $k\ge 1$, dominated convergence with respect to the counting measure on $\{1,2,\dots\}$ applies to the series in \eqref{eq:prop-series-form} and yields
	\[
		\lim_{j\to\infty}\, \sum_{k\ge1}\bigl(e^{tk}-1\bigr)\uppi_k^{N_j} = \sum_{k\ge1} \bigl(e^{tk}-1\bigr)\uppi_k .
	\]
	Combining this with \eqref{eq:prop-series-form} and \eqref{eq:prop-expanded}, we conclude that
	\[
		\sum_{k\ge1} \bigl(e^{tk}-1\bigr)\uppi_k = \lim_{j\to\infty} \frac{1}{r_{N_j}}\,\bb{E}\bigl[e^{tX_{N_j}}-1\bigr] = \fk{c}\bigl(e^{Kt}-1\bigr).
	\]
	We now turn the last display into an identity between power series. Passing to the limit along the subsequence in \eqref{eq:prop-domination} gives
	\begin{equation}\label{eq:prop-limit-domination}
		\uppi_k \le C_\lambda e^{-\lambda k}\qquad\text{for every $k\ge 1$ and every $\lambda\ge\ln 2$}.
	\end{equation}
	Thus, for each fixed $z>0$, choosing $\lambda\ge\ln 2 \vee (\ln z + 1)$ yields $\uppi_k z^k \le C_\lambda (z e^{-\lambda})^k$ with $z e^{-\lambda}<1$. Therefore,
	\[
		F(z) := \sum_{k\ge 1} \uppi_k z^k\qquad\text{	converges for every $z > 0$.}
	\]
	Moreover, since \eqref{eq:prop-limit-domination} holds for arbitrarily large $\lambda$, the radius of convergence of this power series is in fact infinite, and $F$ extends to an entire function. In particular, $F(1) = \sum_{k\ge1}\uppi_k<+\infty$, so writing $z=e^t$ and letting $t$ range over $\bb{R}$ yields
	\[
		F(z) - F(1) = \fk{c}\bigl(z^K-1\bigr)\qquad\text{for every $z>0$}.
	\]
	The entire functions $z\mapsto F(z)$ and $z\mapsto \fk{c}\,z^K + \bigl(F(1)-\fk{c}\bigr)$ therefore agree on the half-line $(0,+\infty)$, hence on all of $\bb{C}$ by the identity theorem, and so do their Taylor coefficients at the origin. Since $K\ge 1$, comparing the coefficient of $z^0$ gives $F(1) = \fk{c}$, while comparing the coefficient of $z^K$ gives $\uppi_K = \fk{c}$, and comparing the coefficient of $z^k$ for $k\ge1$ with $k\ne K$ gives $\uppi_k = 0$.

	This identifies the limit of every subsequence extracted as above. Since $\{\uppi_k^N\}_N$ is bounded for each fixed $k\ge1$, and all such subsequential limits coincide, the full sequence converges to the same value, which proves \eqref{eq:prop-mass-limit}.

	\medskip\noindent\emph{Step 3: Transfer.} We claim that for every sequence $(u_k)_{k\ge1}\subset\bb{C}$ satisfying $|u_k| \le A e^{t k}$ for some $A>0$ and some $t\in\bb{R}$, one has
	\begin{equation}\label{eq:prop-transfer}
		\lim_{N\to+\infty} \frac{1}{r_N} \sum_{k\ge1} u_k\, \bb{P}(X_N = k) = \fk{c}\, u_K.
	\end{equation}
	Indeed, choosing $\lambda\ge\ln 2\vee(t+1)$ in \eqref{eq:prop-domination} gives $|u_k|\,\uppi_k^N \le A\, C_\lambda\, e^{-k}$ for every $N\ge N_\lambda$, which is again summable in $k$ and independent of $N$. The claim then follows from \eqref{eq:prop-mass-limit} by dominated convergence.

	It remains to choose appropriate $u_k$ in \eqref{eq:prop-transfer}. Taking $u_k \equiv 1$ yields
	\[
		\frac{1}{r_N}\bigl(1-\bb{P}(X_N=0)\bigr) = \frac{1}{r_N}\sum_{k\ge1}\bb{P}(X_N=k) \longrightarrow \fk{c},
	\]
	i.e., $\bb{P}(X_N=0) = 1 - \fk{c}\,r_N + o(r_N)$, which together with \eqref{eq:prop-mass-limit} is precisely \eqref{eq:prop-two-point}. Taking instead $u_k = k^{\theta} e^{\lambda k}$ with $\theta>0$ and $\lambda\in\bb{C}$, which satisfies $|u_k|\le A\,e^{(\Re\lambda+1)k}$ for a suitable $A>0$, and noting that $X_N^\theta e^{\lambda X_N}$ vanishes on the event $\{X_N=0\}$ because $\theta>0$, we obtain
	\[
		\frac{1}{r_N}\,\bb{E}\bigl[X_N^{\theta} e^{\lambda X_N}\bigr] = \frac{1}{r_N}\sum_{k\ge1} k^{\theta} e^{\lambda k}\,\bb{P}(X_N=k) \longrightarrow \fk{c}\, K^{\theta} e^{K\lambda},
	\]
	which is \eqref{eq:prop-moments}, thereby concluding the proof.
\end{proof}

\section{Estimates involving geometric random variables}
In this section, we shall collect some well-known and less well-known estimates for geometric random variables. We start with a simple tail estimate.

    \begin{lemma}[Tail estimate] \label{lemma: tail_estimate}
        Let $G$ be a geometric random variable with mean $\mu$. Then, 
        \begin{equation} \label{eq: tail_estimate}
            \bb{P}(G \geq t) \leq \exp \left ( -\frac{t}{\mu+1} \right )\qquad\text{for all $t \geq 0$.}
        \end{equation}
    \end{lemma}
    \begin{proof}
        A simple computation shows that, for $\lceil t \rceil$ the smallest integer greater than $t$, 
        \begin{equation*}
            \bb{P}(G \geq t) = \left ( \frac{\mu}{\mu+1} \right )^{\lceil t \rceil}.
        \end{equation*}
        Consequently, using the estimate $\log(1-x) \leq -x$, we obtain
        \begin{equation*}
            \left ( \frac{\mu}{\mu+1} \right )^{\lceil t \rceil} = \exp \left ( \lceil t \rceil \log \left ( 1 - \frac{1}{\mu+1} \right ) \right ) \leq \exp \left (-\frac{\lceil t \rceil}{\mu+1} \right ) \leq \exp \left ( -\frac{t}{\mu+1} \right ),
        \end{equation*}
        thereby concluding the proof.
    \end{proof}

Using this estimate, we can derive power-deviation estimates of $G$ from its mean $\mu$.

    \begin{proposition}[Power-deviation estimate] \label{prop: Lp_deviation_geom}
        Let $G$ be a geometric random variable with mean $\mu$. Then for any $\alpha \geq 0$, there exists a constant $\sf{c}_\alpha > 0$ such that
        \[
            \bb{E}[|G-\mu|^{1+\alpha}] \leq \sf{c}_\alpha \max(\mu,\mu^{1+\alpha}).
        \]
    \end{proposition}
    \begin{proof}
    	We separate the estimate into the regimes $\mu \geq 1$ and $\mu \leq 1$, where in the first case, the term $\mu^{1+\alpha}$ dominates, while in the second, the $\mu$-term dominates.\footnote{The cut at $\mu=1$ is a matter of convenience: since the estimate includes a constant, only the regimes of large and small $\mu$ matter.}
    	
        \medskip\noindent(1)\emph{ The $\mu \geq 1$ case.} Using the layer-cake formula, we obtain the decomposition
            \begin{align*}
                \bb{E}[|G-\mu|^{1+\alpha}] &= (1+\alpha) \int_0^{+\infty} t^\alpha\, \bb{P}(|G-\mu| \geq t) \dd{t} \\
                &= (1+\alpha) \int_0^\mu  t^\alpha\, \bb{P}(|G-\mu| \geq t) \dd{t} + (1+\alpha)\int_\mu^{+\infty} t^\alpha\, \bb{P}(|G-\mu| \geq t) \dd{t}. 
            \end{align*}
            The first integral is estimated from above by $\mu^{1+\alpha}$ using $\bb{P}(|G-\mu| \geq t) \leq 1$. 
            
            As for the second term, we notice that since $G \geq 0$, one has $\bb{P}(|G-\mu| \geq t) = \bb{P}(G \geq t + \mu)$ whenever $t > \mu$. Therefore, using the tail estimate \eqref{eq: tail_estimate}, we can bound the last integral by
            \begin{align*}
                \int_\mu^{+\infty} t^\alpha \exp \left ( -\frac{t+\mu}{\mu+1} \right ) \dd{t} &= \mu^{1+\alpha} \int_1^{+\infty} t^\alpha \exp\left  ( -\frac{\mu}{\mu+1} (t+1) \right ) \dd{t} \\
                &\leq \mu^{1+\alpha} \int_1^{+\infty} t^\alpha e^{-(t+1)/2} \dd{t},
            \end{align*}
            where we used the bound $\frac{\mu}{\mu+1} \geq \frac{1}{2}$ for $\mu \geq 1$. Putting the inequalities together, we deduce
            \begin{equation*}
                \bb{E}[|G-\mu|^{1+\alpha}] \leq \mu^{1+\alpha} \left ( 1 + (1+\alpha) \int_1^{+\infty} t^\alpha e^{-(t+1)/2}  \dd{t} \right ) \leq \sf{c}_\alpha^1\, \mu^{\alpha+1}.
            \end{equation*}
        
        \medskip\noindent(2)\emph{ The $\mu \le 1$ case.} For $\mu \leq 1$, we have $|k-\mu|^\alpha = (k-\mu)^\alpha$ for every $k \geq 1$, and hence,
            \begin{align*}
                \bb{E}[|G-\mu|^{1+\alpha}] &= \mu^{1+\alpha}\, \bb{P}(G=0) + \bb{E}[|G-\mu|^{1+\alpha} \mathbbm{1}_{\{G \geq 1\}}] \\
                &= \frac{1}{\mu+1} \mu^{1+\alpha} + \frac{1}{\mu+1} \sum_{k \geq 1} (k-\mu)^{1+\alpha} \left ( \frac{\mu}{\mu+1} \right )^k.
            \end{align*}
            To bound the last term, notice that $\frac{\mu}{\mu+1} \leq \frac{1}{2}$ for $\mu \leq 1$, and $\frac{1}{\mu+1} \leq 1$. Therefore, we obtain
            \begin{align*}
                \frac{1}{\mu+1} \sum_{k \geq 1} (k-\mu)^{1+\alpha} \left ( \frac{\mu}{\mu+1} \right )^k &= \frac{\mu}{(\mu+1)^2} \sum_{k \geq 0} (k+1-\mu)^{1+\alpha} \left ( \frac{\mu}{\mu+1} \right )^k \\
                &\leq \mu\, \sum_{k \geq 0} 2^{-k}(k+1)^{1+\alpha}.
            \end{align*}
            Hence, since $\mu^{1+\alpha} \leq \mu$, we deduce
            \begin{equation*}
                \bb{E}[|G-\mu|^{1+\alpha}] \leq \mu^{1+\alpha} + \mu\, \sum_{k \geq 0} 2^{-k}(k+1)^{1+\alpha} \leq \sf{c}_\alpha^2\, \mu.
            \end{equation*}
            Putting the estimates together yields the assertion with $\sf{c}_\alpha := \max\{\sf{c}_\alpha^1,\sf{c}_\alpha^2\}$.
    \end{proof}

\section{Central limit theorem for infinite triangular arrays}\label{app:clt}
We give a slight extension of the classical Lindeberg central limit theorem to infinite triangular arrays. We first recall the classical finite version.

    \begin{theorem}[Lindeberg CLT] \label{thm: Lindeberg_CLT}
        A triangular array is a family of random variables $\bb{Y} = (Y_{n,q})$ indexed by $n\in\bb{N}_0$ and $q = 1,\ldots,m(n)$ such that
        \begin{enumerate}
        	\item $\{Y_{n,q}: q = 1,\ldots,m(n)\}$ is a family of independent random variables for every $n\ge 0$, 
            \item $\mathbb{E}[Y_{n,q}] = 0$, $q=1,\ldots,m(n)$ and $\sum_{q=1}^{m(n)} \mathbb{E}[|Y_{n,q}|^2] = 1$ for every $n \geq 0$.
        \end{enumerate}
        A triangular array is a Lindeberg triangular array if it satisfies
        \begin{equation} \label{eq: classical_Lindeberg_condition}
            \lim_{n \to \infty} \sum_{q=1}^{m(n)} \mathbb{E}[|Y_{n,q}|^2 \cdot \mathbf{1}_{\{|Y_{n,q}| \geq \epsilon\}}] = 0,\qquad\text{for every $\epsilon > 0$.}
        \end{equation}
        
        If $\bb{Y}$ is a triangular array, then the CLT holds for $S_n:= \sum_{q=1}^{m(n)} Y_{n,q}$, i.e.,
        \[
            \rm{Law}(S_n)\rightharpoonup \cl{N}(0,1)\quad\text{weakly as $n\to\infty$}.
        \]
    \end{theorem}

As explained above, our goal will be to replace the finite sum on $q$ with an infinite sum. To this end, we consider a locally compact separable metric space $\fk{X}$ and let
\[
	\varGamma(\fk{X}) := \bigl\{\Lambda\subset \fk{X}: \#(\Lambda\cap K)<+\infty\;\;\text{for every compact set $K\subset\fk{X}$}\bigr\},
\]
denote the space of locally finite configurations over $\fk{X}$. Accordingly, we introduce the following notions of \emph{generalized triangular and Lindeberg arrays}. 

\begin{definition}[Generalized triangular array]\label{def:generalized-lindeberg-array}
	Let $(\Omega,\cl{F},\bb{P})$ be a probability space and $\{\Lambda_n\}_{n\in\bb{N}_0}\subset \varGamma(\fk{X})$ be a family of locally finite configurations over $\fk{X}$. A \emph{generalized triangular array} is a family 
	\[
		\bb{X} = \bigl\{X_{n,p} : p \in \Lambda_n\in\varGamma(\fk{X}),\, n \in\bb{N}_0\bigr\}
	\] 
	of random variables on $\Omega$ such that
	\begin{enumerate}
		\item $\{X_{n,p}: p \in \Lambda_n\}$ is a family of independent random variables for every $n\ge 0$, 
		\item $X_{n,p}$ has finite second moment for every $p\in\Lambda_n$ and $n\in\bb{N}_0$, and we have that
		\begin{equation*}
			0 < \sigma_n^2:=\sum_{p \in \Lambda_n} \rm{Var}(X_{n,p}) < +\infty.
		\end{equation*}
	\end{enumerate}
	The generalized triangular array $\bb{X}$ is called a \emph{generalized Lindeberg triangular array} if, in addition,
	\begin{equation} \label{eq: Lindeberg_condition}
		\lim_{n \to \infty} \sum_{p \in \Lambda_n} \bb{E}[|S_{n,p}|^2 \mathbbm{1}_{\{|S_{n,p}| > \epsilon\}}] = 0\qquad\text{for every $\epsilon>0$},
	\end{equation}
	where $S_{n,p} := \sigma_n^{-1} (X_{n,p} - \bb{E}[X_{n,p}])$ for every $p\in\Lambda_n$ and $n\in\bb{N}_0$.
\end{definition}

\begin{remark}[Lyapunov condition]
	A sufficient condition for the \emph{Lindeberg condition} \eqref{eq: Lindeberg_condition} is the \emph{Lyapunov condition}: There exists some $\delta>0$ such that
	\begin{equation} \label{eq: Lyapunov_condition}
		\lim_{n\to\infty}\sum_{p \in \Lambda_n} \bb{E}[|S_{n,p}|^{2+\delta}] = 0.
	\end{equation}
\end{remark}

We now state the generalized triangular array version of Lindeberg CLT, from which Theorem~\ref{thm: CLT-geometric} directly follows.

\begin{theorem}[Generalized Lindeberg CLT] \label{thm: generalized_Lindeberg_CLT}
    Let $\bb{X}$ be a generalized Lindeberg array (satisfying condition \eqref{eq: Lindeberg_condition}). Then,
    \[
    \rm{Law}(S_n)\rightharpoonup \cl{N}(0,1)\quad\text{weakly as $n\to\infty$,\quad where\;\;}S_n=\sum_{p\in\Lambda_n}S_{n,p}.
\]
    If the Lyapunov condition \eqref{eq: Lyapunov_condition} holds with parameter $\delta>0$, then the $\theta$-moment of $S_n$ also converges for any $\theta\in[1, 2+\delta)$. 
\end{theorem}
\begin{proof}
    Up to a renormalization of the triangular array, we can assume that $\sigma_n = 1$ and that all the random variables $X_{n,p}$ are centered, i.e., $S_{n,p}=X_{n,p}$ in this case.
    
    Let $(K_j)_{j \geq 0}$ be a compact exhaustion of $\cl{X}$. By construction, we find for each $n \geq 0$ an index $j(n)\in\bb{N}$ such that $\sum_{p \in \Lambda_n \backslash K_{j(n)}} \bb{E}[|X_{n,p}|^2] < 1/n$. Let us consider an enumeration $p_1,\ldots,p_{m(n)}$ of the points of $\Lambda_n {\cap} K_{j(n)}$. Set
    \begin{equation*}
        Y_{n,0} := \sum_{p \in \Lambda_n \backslash K_{j(n)}} X_{n,p}, \qquad Y_{n,k} = X_{n,p_k}, \quad k=1,\ldots,m(n).
    \end{equation*}
    We then introduce the new finite triangular array $\bb{Y} = \{Y_{n,k} : k=0,\ldots,m(n),\, n\in\bb{N}_0\}$.
    
    Since $(X_{n,p})_{p\in\Lambda_n}$ is independent, so is $(Y_{n,k})_{k=0,\ldots,m(n)}$, with $\bb{E}[Y_{n,k}] = 0$, and 
    \begin{equation*}
        \sum_{k=0}^{m(n)} \rm{Var}(Y_{n,k}) = \sum_{p \in \Lambda_n \backslash K_{j(n)}} \rm{Var}(X_{n,p}) + \sum_{p \in K_{j(n)}} \rm{Var}(X_{n,p}) = \sigma_n = 1.
    \end{equation*}
    Furthermore, for any $\epsilon > 0$,
    \begin{align*}
        \sum_{k=0}^{m(n)} \bb{E}[|Y_{n,k}|^2 \mathbbm{1}_{\{|Y_{n,k}| > \epsilon\}}] 
        &= \bb{E}[|Y_{n,0}|^2 \mathbbm{1}_{|Y_{n,0}| > \epsilon}]] + \sum_{p \in \Lambda_n {\cap} K_{j(n)}} \bb{E}[|X_{n,p}|^2 \mathbbm{1}_{|X_{n,p}| \geq \epsilon}] \\
        &\leq \bb{E}[|Y_{n,0}|^2] + \sum_{p \in \Lambda_n {\cap} K_{j(n)}} \bb{E}[|X_{n,p}|^2 \mathbbm{1}_{|X_{n,p}| > \epsilon}] \\
        &= \sum_{i \in \Lambda_n \setminus K_{j(n)}} \bb{E}[|X_{n,i}|^2] + \sum_{p \in \Lambda_n {\cap} K_{j(n)}} \bb{E}[|X_{n,p}|^2 \mathbbm{1}_{|X_{n,p}| > \epsilon}] \\
        &\leq \frac{1}{n} + \sum_{p \in \Lambda_n} \bb{E}[|X_{n,p}|^2 \mathbbm{1}_{\{|X_{n,p}| > \epsilon\}}] \to 0.
    \end{align*}
    Hence, the finite triangular array $\bb{Y}$ satisfies the classical Lindeberg condition \eqref{eq: classical_Lindeberg_condition}, which implies the convergence of the distribution of $S_n=\sum_{q=0}^{m(n)} Y_{n,k} = \sum_{p \in \Lambda_n} X_{n,p}$ to the standard normal distribution $\cl{N}(0,1)$, thereby concluding the first part of the statement. 

    Now let the Lyapunov's condition (\ref{eq: Lyapunov_condition}) be satisfied with parameter $\delta>0$. Then, by Rosenthal's inequality, for any $J \subset \Lambda_n$ finite, we have that
    \begin{align*}
        \bb{E} \left[\biggl | \sum_{p \in J} X_{n,p} \biggr |^{2+\delta}\right] &\leq C_\delta \max \left \{ \sum_{p \in J} \bb{E}[|X_{n,p}|^{2+\delta}], \left (\sum_{p \in J} \bb{E}[|X_{n,p}|^2] \right )^{1+\delta/2} \right \} \\
        &\leq C_\delta \max \left \{ \sum_{p \in \Lambda_n} \bb{E}[|X_{n,p}|^{2+\delta}], 1 \right \}.
    \end{align*}
    Taking the limit along a sequence of finite sets $(J_N)_N$ whose union is $\Lambda_n$ (for instance, intersecting with the compact exhaustion), and using the Lyapunov condition, we obtain
    \begin{equation*}
        \limsup_{n\to\infty} \bb{E}[|S_n|^{2+\delta}] <+\infty.
    \end{equation*}
    This provides a uniform $(2+\delta)$-moment bound, and standard compactness in the Wasserstein space of order $\theta$ for $\theta\in[1, 2+\delta)$ provides the asserted moment convergence. (Here we recall that the Wasserstein space of order $\theta$ is the space of all probability measures with finite $\theta$-th moment, endowed with the Wasserstein distance of order $\theta$, denoted by $\bb{W}_\theta$). 
\end{proof}

\begin{proof}[Proof of Theorem~\ref{thm: CLT-geometric}]
    This follows directly from Theorem~\ref{thm: generalized_Lindeberg_CLT}.
\end{proof}

We finish this Section by stating a technical Lemma, useful for the Bose–Einstein condensation application. It states that if a sum of two independent copies of a random variable satisfies a CLT, then the random variable itself also satisfies a CLT. 

\begin{lemma} \label{lemma: CLT_sum_CLT_non_sum}
    Let $(X_N)_{N \geq 0}$ be a sequence of random variables, and let $Y_N = X_N^1 + X_N^2$ where $X_N^1,X_N^2$ are two independent copies of $X_N$. Then if $\mathrm{Law}(Y_N)$ converges weakly to $\cl{N}(0,1)$ together with all its $\theta$-moments for $\theta \in [1,p)$, with $p > 2$, then $\mathrm{Law}(X_N)$ converges weakly to $\cl{N}(0,1/2)$ together with all its $\theta$-moments for $\theta \in [1,p)$. 
\end{lemma}

\begin{proof}
    We divide into two steps:
    
    \medskip\noindent \emph{Step 1: Tightness.} Let $\theta < p$ and choose $\eta$ with $\theta < \eta < p$. Let $X_N^3$ be a third independent copy of $X_N$, then, as $2 X_N^1 = Y_N^{12} + Y_N^{13} - Y_N^{23}$ with $Y_N^{ij} = X_N^i + X_N^j$, and using that $|a+b+c|^\eta \leq 3^\eta(|a|^\eta + |b|^\eta + |c|^\eta)$ (which follows by convexity). We have
    \[
        2^\eta \bb{E}[|X_N|^\eta] = \bb{E}[|(X_N^1 + X_N^2) + (X_N^1 + X_N^3) - (X_N^2 + X_N^3)|^\eta] \leq 3^{\eta+1} \bb{E}[|Y_N|^\eta].
    \]
    Hence $\sup_N \bb{E}[|X_N|^\eta] < +\infty$, which shows that $(\rm{Law}(X_N))_N$ is relatively compact in $\bb{W}_\theta$. 
        
    \medskip\noindent \emph{Step 2: Identification of limit.} Let $X$ be a $\bb{W}_\theta$ limit point of $(X_N)_N$ along a subsequence, the Lemma follows once we show that $X \sim \cl{N}(0,1/2)$. 
    
    Let $X^1,X^2$ be independent copies of $X$, then $(X_N^1,X_N^2)$ converges in law (up to subsequence) to $(X^1,X^2)$. By the continuous mapping theorem, this implies that $Y_N = X_N^1 + X_N^2$ converges in law (up to subsequence) to $X^1 + X^2$, thereby, $X^1 + X^2 \sim \cl{N}(0,1)$. As a consequence, for all real $\lambda$, we have $\bb{E}[e^{\lambda X}]^2 = \bb{E}[e^{\lambda(X^1 + X^2)}] = e^{\lambda^2/2}$, which shows that $\bb{E}[e^{\lambda X}] = e^{\lambda^2/4}$, i.e. $X \sim \cl{N}(0,1/2)$, concluding the proof.
\end{proof}

\section{A Riemann-sum approximation lemma}\label{app:riemann-sum}

In this section, we prove Lemma~\ref{lem:dense-mean-variance}, together with the Riemann-sum approximation on which it rests. The latter converts the infinite lattice sums of Section~\ref{sec:dense-regime} into integrals over $\bb{R}^d_0$, and is stated for functions with a possible singularity at the origin, as required by the weights $\mu^\mathsf{QD}$ of Section~\ref{sec:BEC}. 

\begin{lemma}\label{lemma: Riemann_approx}
	Let $f : \bb{R}^d_0 \to [0,+\infty)$ be a continuous function satisfying 
	\begin{enumerate}
		\item (Behavior at $+\infty$) There exists $\gamma > d$ such that $\limsup_{|x|\to+\infty} |x|^\gamma f(x)<+\infty$.
		\item (Behavior at $0$) There exists $\alpha < d$ such that $\limsup_{|x|\to 0} |x|^\alpha f(x) <+\infty$.
	\end{enumerate}
	Then we have
	\[
	\lim_{\kappa\to 0} \sum_{p \in \kappa \bb{Z}^d_0} \kappa^d f(p) = \int_{\bb{R}^d_0} f(x) \dd{x}
	\]
	where all the quantities involved are finite.
\end{lemma}

\begin{proof}
    Consider the family of rescaled counting 
    measures
    \[
    	\mu^\kappa := \kappa^d \sum_{p \in \kappa\bb{Z}^d_0} \delta_p\qquad \text{on \;$\kappa\bb{Z}^d$}.
    \]
    It is known that $\mu_\kappa$ converges vaguely (i.e., in duality with $C_c(\bb{R}^d_0)$) to the Lebesgue measure $\uplambda^d$ as $\kappa \to 0$.

    Define the measures $\eta^\kappa := f \mu^\kappa$ and $\eta = f \uplambda^d$. By continuity of $f$, we have that $\eta^\kappa \rightharpoonup \eta$ vaguely. Since $\kappa^d \sum_{p \in \kappa \bb{Z}^d_0} f(p) = \eta^\kappa(\bb{R}^d_0)$ and $\int_{\bb{R}^d_0} f(x) \dd{x} = \eta(\bb{R}^d_0)$, it suffices to show that the family $(\eta^\kappa)_{\kappa \leq \kappa_0}$ is tight for some $\kappa_0 > 0$. We recall that a sufficient condition is the existence of a measurable function $\psi : \bb{R}^d_0 \to [0,+\infty]$ with relatively compact sublevel sets such that $\sup_{\kappa \leq \kappa_0} \int_{\bb{R}^d_0} \psi \dd{\eta^\kappa} < +\infty$. We prove the existence of such a function.

    By the assumptions on $f$, there exists $0 < r < R$ and $C > 0$ such that
    \[
        \begin{cases}
            f(x) \leq \frac{C}{|x|_\infty^\alpha} & \text{for $|x|_\infty \leq r$}, \\
            f(x) \leq \frac{C}{|x|_\infty^\gamma} & \text{for $|x|_\infty \geq R$}.
        \end{cases}
    \]
    Choosing $\delta > 0$ such that $\alpha+\delta < d < \gamma-\delta$ and defining
    \[
        \psi(x) := 
        \begin{cases}
            \frac{1}{|x|_\infty^\delta} & \text{if $|x|_\infty \leq r$}, \\
            |x|_\infty^\delta & \text{if $|x|_\infty \geq R$}, \\
            0 & \text{otherwise}, 
        \end{cases}
    \]
    which, by construction, has relatively compact level sets on $\bb{R}^d_0$, we deduce
    \begin{align*}
        \int_{\bb{R}^d_0} \psi \dd{\eta^\kappa} &= \int_{0 < |x|_\infty \leq r} \frac{1}{|x|_\infty^\delta} f(x)\, \mu^\kappa(\dd x) + \int_{|x|_\infty \geq R} |x|_\infty^\delta\, f(x)\, \mu^\kappa(\dd x) \\
        &\leq C \left ( \kappa^d \sum_{\substack{p \in \kappa \bb{Z}^d_0 \\ |p|_\infty \leq r}} \frac{1}{|p|_\infty^{\alpha+\delta}} + \kappa^d \sum_{\substack{p \in \kappa \bb{Z}^d_0 \\ |p|_\infty \geq R}} \frac{1}{|p|_\infty^{\gamma-\delta}} \right ).
    \end{align*}
    Let $\kappa_0 = \min(r,R/2)$. Then for all $\kappa \leq \kappa_0$, we can apply Lemma~\ref{lemma: sum_tail_estimates} (using $\alpha+\delta < d < \gamma-\delta$) to obtain, after possibly modifying the constant $C$ (independently of $\kappa \leq \kappa_0$),
    \[
        \int_{\bb{R}^d_0} \psi \dd{\eta^\kappa} \leq C (r^{d-\alpha-\delta} + R^{d+\delta-\gamma}).
    \]
    This bound is uniform in $\kappa\le \kappa_0$, which shows that the family $(\eta^\kappa)_{\kappa \leq \kappa_0}$ is tight. This completes the proof, since tightness, together with vague convergence, upgrades $\eta^\kappa \rightharpoonup \eta$ to narrow convergence and, in particular, implies convergence of the total mass.
\end{proof}

\begin{proof}[Proof of Lemma~\ref{lem:dense-mean-variance}]
	Let $f_\mu:\bb{R}^d\backslash\{0\}\to \bb{R}$ be a continuous function satisfying assumption \eqref{ass:dense-conditions-proof} and let $\theta\in[1,d/\beta)$. Then, $\theta\gamma \ge \gamma>d$ and 
	\[
		\limsup_{|x|\to+\infty}|x|^{\theta\gamma} f^\theta(x) = \Bigl(\limsup_{|x|\to+\infty}|x|^\gamma f_\mu(x)\Bigr)^{\theta} < +\infty.
	\]
	Moreover, since $\theta\beta < d$, we also have that
	\[
		\limsup_{|x|\to 0} |x|^{\theta\beta} f^\theta(x) = \Bigl(\limsup_{|x|\to 0}|x|^\beta f_\mu(x)\Bigr)^{\theta} < +\infty.
	\]
	Therefore, $f_\mu^\theta$ satisfies the assumptions in Lemma~\ref{lemma: Riemann_approx} for every $\theta\in[1,d/\beta)$. Applying Lemma~\ref{lemma: Riemann_approx} then concludes the proof.
\end{proof}

When $f$ satisfies some additional growth assumptions, the error term can be measured much more precisely in the above approximation. This is particularly useful in the case of Bose--Einstein condensation, where an explicit formula for the mean $\mu_x$ is available. 

\begin{lemma} \label{lem: improved_Riemann_approx}
    Let $f : \bb{R}^3_0 \to \bb{R}$ be a $C^2$-function satisfying the assumptions of Lemma~\ref{lemma: Riemann_approx} (with $d=3$) and, in addition: 
    \begin{enumerate}
        \item (Behavior at $+\infty$) There exists $\beta > 3$ such that $\limsup_{|x| \to +\infty} |x|^\beta |D^2 f(x)| < +\infty$. 
        \item (Behavior at $0$) There exists $\alpha < 3$ and some constant $a \in \bb{R}$ such that, defining $g(x) := f(x) - a |x|^{-1}$, one has $\limsup_{|x| \to 0} (|g(x)| + |x|^\alpha |D^2 g(x)|) < +\infty$. 
    \end{enumerate}
    Then we have
    \[
        \kappa^3 \sum_{p \in \kappa \bb{Z}^3_0} f(p) = \int_{\bb{R}^3_0} f(x) \dd{x} + O(\kappa^2).
    \]
\end{lemma}

\begin{proof}
    We will write $C$ for a constant depending only on the parameters but not on $\kappa$, which may change from line to line. We assume throughout the proof that $\kappa < 1$. By the assumptions, we can find $R > 1$ such that
    \[
        \begin{cases}
            \displaystyle |g(x)| \leq C \quad \text{and} \quad |D^2 g(x)| \leq \frac{C}{|x|_\infty^\alpha} & \text{on $|x|_\infty \leq R+2$}, \\
            \displaystyle |D^2 f(x)| \leq \frac{C}{|x|_\infty^\beta}, &\text{on $|x|_\infty \geq R$}.
        \end{cases}
    \]
    We set $N_\kappa := \lceil R/\kappa \rceil$. Let $Q_\kappa := [-\kappa/2,\kappa/2]^3$, and for all $p \in \kappa \bb{Z}^3$, set $Q_\kappa(p) := p + Q_\kappa$. Using that $|Q_\kappa(p)| = \kappa^3$, we bound
    \begin{align*}
        &\left | \kappa^3 \sum_{p \in \kappa \bb{Z}^3_0} f(p) - \int_{\bb{R}^3_0} f(x) \dd{x} \right | \leq \int_{Q_\kappa} |f(x)| \dd{x} \;\;+ \sum_{\substack{p \in \kappa \bb{Z}^3_0 \\ |p|_\infty \leq \kappa N_\kappa}} \left | \int_{Q_\kappa(p)} (g(x) - g(p)) \dd{x} \right | \\
        &\hspace{4em}+ |a| \sum_{\substack{p \in \kappa \bb{Z}^3_0 \\ |p|_\infty \leq \kappa N_\kappa}} \left | \int_{Q_\kappa(p)} \left ( \frac{1}{|x|} - \frac{1}{|p|} \right ) \dd{x} \right|\;\; + \sum_{\substack{p \in \kappa \bb{Z}^3_0 \\ |p|_\infty \geq \kappa(N_\kappa+1)}} \left | \int_{Q_\kappa(p)} (f(x) - f(p)) \dd{x} \right | \\
        &= \textit{(I)} + \textit{(II)} + \textit{(III)} + \textit{(IV)}.
    \end{align*}
    We show that each term is $O(\kappa^2)$. 
    
    \medskip \noindent \emph{(I): Near $0$ integral.} As $\kappa < 1 \leq R$, we have $|f(x)| \leq |a| \, |x|^{-1} + C$, thereby
    \[
        \int_{Q_\kappa} |f(x)| \dd{x} \leq |a| \int_{Q_\kappa} \frac{1}{|x|} \dd{x} + C \kappa^3.
    \]
    Let $c_d > 0$ such that $Q_\kappa \subset B_{c_d \kappa}$ for all $\kappa > 0$. Since $\int_{B_{c \kappa}} |x|^{-1} \dd{x} = 2 \pi c^2_d \kappa^2$. We obtain 
    \[
        \int_{Q_\kappa} |f(x)| \dd{x} \leq C (\kappa^2 + \kappa^3) = O(\kappa^2).
    \]
        
    \medskip \noindent \emph{(II): Near $0$ $g$-sum.} We first bound $D^2 g$ on each $Q_\kappa(p)$ for $|p|_\infty \leq \kappa N_\kappa$, $p \in \kappa \bb{Z}^3_0$. If $x \in Q_\kappa(p)$, then $|x|_\infty \leq |p|_\infty + |x-p|_\infty \leq \kappa (N_\kappa +1/2) \leq R + 2$. Therefore $|D^2 g(x)| \leq C |x|_\infty^{-\alpha}$. On the other hand, we have $|x|_\infty \geq |p|_\infty - |x-p|_\infty \geq |p|_\infty - \kappa/2$. Since $|p|_\infty \geq \kappa$ (as $p \neq 0$), we can bound this by $|p|_\infty / 2$. Similarly, we have $|x|_\infty \leq 2 |p|_\infty$. Therefore $|x|_\infty^{-\alpha} \leq c |p|_\infty^{-\alpha}$ and we deduce that $|D^2 g(x)| \leq C |p|_\infty^{-\alpha}$. 

    As a consequence, using second-order Taylor estimates, we have on $Q_\kappa(p)$
    \[
        |g(x) - g(p) - \nabla g(p) \cdot (x-p)| \leq C \frac{1}{|p|_\infty^{\alpha}} |x-p|^2 \leq \frac{C \kappa^2}{|p|_\infty^\alpha}. 
    \]
    Using that $\int_{Q_\kappa(p)} \nabla g(p) \cdot (x-p) \dd{x} = 0$ (since the cube is centered), we obtain 
    \[
        \sum_{\substack{p \in \kappa \bb{Z}^3_0 \\ |p|_\infty \leq \kappa N_\kappa}} \left | \int_{Q_\kappa(p)} (g(x) - g(p)) \dd{x} \right | \leq C \kappa^5 \sum_{\substack{p \in \kappa \bb{Z}^3_0 \\ |p|_\infty \leq \kappa N_\kappa}} \frac{1}{|p|_\infty^\alpha}.
    \]
    Using Lemma~\ref{lemma: sum_tail_estimates} we have that 
    \[
        \kappa^3 \sum_{\substack{p \in \kappa \bb{Z}^3_0 \\ |p|_\infty \leq \kappa N_\kappa}} \frac{1}{|p|_\infty^{\alpha}} \leq C (\kappa N_\kappa)^{3-\alpha} = O(1)
    \]
    as $N_\kappa \sim R \kappa^{-1}$. We conclude that $\textit{(II)} = O(\kappa^2)$. 
        
    \medskip \noindent \emph{(III): Near zero inverse sum.} Let $\chi(x) := |x|^{-1}$. By scaling we have
    \[
        \textit{(III)} = |a| \kappa^2 \sum_{\substack{p \in \bb{Z}^3_0 \\ |p|_\infty \leq N_\kappa}} \left | \int_{p + Q_1} \left ( \frac{1}{|x|} - \frac{1}{|p|} \right ) \dd{x} \right |.
    \]
        This time, a first-order Taylor expansion is not enough. Instead we perform a second-order Taylor expansion and write
        \[
            \chi(x) = \chi(p) + \nabla \chi(p) \cdot (x-p) + \frac{1}{2} (x-p)^T D^2 \chi(p) \cdot (x-p) + R_3(\chi,p,x-p).
        \]
        A rough estimate shows that $|D^3 \chi(x)| \leq |x|^{-4}$, hence, comparing the $|\cdot|$-norm to the $|\cdot|_\infty$-norm, and arguing as in the $\textit{(II)}$ term, we have $R_3(\chi,p,x-p) \leq C |p|^{-4}_\infty$ on $p + Q_1$. Integrating we obtain that
        \[
            \left | \int_{p + Q_1} \left ( \frac{1}{|x|} - \frac{1}{|p|} \right ) \dd{x} \right | \leq \frac{1}{2} \left | \int_{Q_1} x^T D^2 \chi(p) x \dd{x} \right | + \frac{C}{|p|_\infty^4}.
        \]
        But $\int_{Q_1} x \otimes x \dd{x} = \frac{1}{12} \rm{I}_3$. Hence
        \[
            \int_{Q_1} x^T D^2 \chi(p) x \dd{x} = D^2 \chi(p) : \int_{Q_1} x \otimes x \dd{x} = \frac{1}{12} \Delta \chi(p) = 0,
        \]
        where we used that $\chi$ is harmonic in $\bb{R}^3_0$ in dimension $3$. We obtain
        \[
            \sum_{\substack{p \in \bb{Z}^3_0 \\ |p|_\infty \leq N_\kappa}} \left | \int_{p + Q_1} \left ( \frac{1}{|x|} - \frac{1}{|p|} \right ) \dd{x} \right | \leq C \sum_{p \in \bb{Z}^3_0} \frac{1}{|p|^4} < +\infty,
        \]
        which shows $\textit{(III)} = O(\kappa^2)$.
        
    \medskip \noindent \emph{(IV): Far from $0$ sum.} If $|p|_\infty \geq \kappa(N_\kappa+1)$ and $x \in Q_\kappa(p)$, then $|x|_\infty \geq |p|_\infty - |x-p|_\infty \geq \kappa (N_\kappa+1/2) \geq R$. Hence we have $|D^2 f(x)| \leq C |x|^{-\beta}_\infty$ on $Q_\kappa(p)$. Since we also have $|x|_\infty \geq |p|_\infty / 2$, we get that $|D^2 f(x)| \leq C |p|^{-\beta}_\infty$. 
    
    Using the same Taylor expansion argument as for $\textit{(II)}$, we have
    \[
        \left | \int_{Q_\kappa(p)} (f(x) - f(p)) \dd{x} \right | \leq \kappa^5 \frac{C}{|p|_\infty^{\beta}}.
    \]
    Summing up, and using that $\beta > 3$ and Lemma~\ref{lemma: sum_tail_estimates},
    \[
        \textit{(IV)} \leq C \kappa^{2} \kappa^3 \sum_{\substack{p \in \kappa \bb{Z}^3_0 \\ |p|_\infty \geq \kappa (N_\kappa+1)}} \frac{1}{|p|_\infty^\beta} \leq C \kappa^{5-\beta} N_\kappa^{3-\beta} = O(\kappa^2)
    \]
    for $\kappa$ sufficiently small. Here, we used that $N_\kappa \to +\infty$, so that eventually $N_\kappa^{-1} \leq 1/2$. This concludes the proof.
\end{proof}

\begin{remark}
    The condition on the behavior at $0$ is a Coulomb-type behavior. For $d > 3$, we expect the same conclusion to hold assuming $\beta > d > \alpha$, using $g(x) = f(x) - a \chi_d(x)$ instead of $f(x) - a |x|^{-1}$, where $\chi_d$ is the fundamental solution to the Poisson equation, that is
    \[
        \chi_d(x) = \frac{2}{d-2} \frac{1}{|x|^{d-2}}.
    \]
    In dimension $2$, one would like to use $\chi_2(x) = 2 \log |x|$, but issues arise in controlling the central cube. So a similar result might hold if we impose additional control on $f$ near $0$.
\end{remark}

In the proof we used the following technical Lemma to obtain estimates for our Riemann-sum approximation. We will use $|\cdot|_\infty$ for the $l^\infty$-norm on $\bb{R}^d$.

\begin{lemma} \label{lemma: sum_tail_estimates}
    Let $\alpha < d < \gamma$. There exists constant $C_{d,\alpha}$ and $C_{d,\gamma}$ such that
    \begin{equation} \label{eq: sum_tail_lower}
        \kappa^d \sum_{\substack{p \in \kappa\bb{Z}^d_0 \\ |p|_\infty \leq \epsilon}} \frac{1}{|p|_\infty^\alpha} \leq C_{d,\alpha} \epsilon^{d-\alpha} \quad \text{for all $\epsilon \geq \kappa$.}
    \end{equation}
    and
    \begin{equation} \label{eq: sum_tail_upper}
        \kappa^d \sum_{\substack{p \in \kappa\bb{Z}^d_0 \\ |p|_\infty \geq \epsilon^{-1}}} \frac{1}{|p|_\infty^\gamma} \leq C_{d,\gamma} \epsilon^{\gamma-d} \quad \text{$\epsilon > 0$\; with\; $\kappa \epsilon < 1/2$.}
    \end{equation}
\end{lemma}

\begin{proof}
    Let $A_d \geq 0$ be such that 
    \[
        \# \{ p \in \bb{Z}^d_0, |p|_\infty = n\} \leq A_d n^{d-1}.
    \]
    A simple combinatorial argument shows that $A_d$ is well-defined. Writing $p = \kappa q$, and using integral comparison, for $\alpha + 1 \geq d$, since $t \to t^{d-1-\alpha}$ is non-increasing the near-origin sum satisfies
    \[
        \kappa^d \sum_{\substack{p \in \kappa\bb{Z}^d\backslash\{0\} \\ |p| \leq \epsilon}} |p|^{-\alpha}
        \leq A_d\, \kappa^{d-\alpha} \sum_{n=1}^{\lfloor \epsilon/\kappa \rfloor} n^{d-1-\alpha}
        \leq A_d\, \kappa^{d-\alpha} \int_0^{\epsilon/\kappa} t^{d-1-\alpha}\,\dd{t}
        = \frac{A_d}{d-\alpha}\,\epsilon^{d-\alpha}.
    \]
    While for $\alpha+1 < d$ $t \to t^{d-1-\alpha}$ is non-decreasing, therefore
    \begin{align*}
        \kappa^d \sum_{\substack{p \in \kappa\bb{Z}^d\backslash\{0\} \\ |p| \leq \epsilon}} |p|^{-\alpha}
        &\leq A_d\, \kappa^{d-\alpha} \sum_{n=1}^{\lfloor \epsilon/\kappa \rfloor} n^{d-1-\alpha} \\
        & \leq A_d\, \kappa^{d-\alpha} \int_1^{\epsilon/\kappa+1} t^{d-1-\alpha}\,\dd{t}
        = \frac{A_d}{d-\alpha}\,\epsilon^{d-\alpha} \left ( \left ( 1 + \frac{\kappa}{\epsilon} \right )^{d-\alpha} - 1 \right ),
    \end{align*}
    and since $( 1 + \kappa / \epsilon)^{d-\alpha}-1 \leq 2^{d-\alpha}-1$ provided that $\kappa \leq \epsilon$, we conclude. Similarly, the far-field sum satisfies
    \[
       \kappa^d \sum_{\substack{p \in \kappa\bb{Z}^d_0 \\ |p|_\infty \geq \epsilon^{-1}}} \frac{1}{|p|_\infty^{\gamma}} \leq A_d \, \sum_{n=\lceil \epsilon^{-1} \kappa^{-1} \rceil}^{+\infty} n^{d-1-\gamma} \leq A_d \, \kappa^{d-\gamma} \int_{\epsilon^{-1} \kappa^{-1}-1}^{+\infty} t^{d-1-\gamma} = \frac{A_d}{\gamma-d} \frac{\epsilon^{\gamma-d}}{(1-\kappa \epsilon)^{\gamma-d}}.
    \]
    This concludes the proof since $(1-\kappa \epsilon)^{\gamma-d} \geq 2^{d-\gamma}$ if $\kappa \epsilon \leq 1/2$. 
\end{proof}

\section{Explicit integrals for the quantum depletion}\label{app:refined-QD}

In this section we prove Lemma~\ref{lem:QD-integrals}, which supplies the two explicit constants entering the large scattering asymptotics of Theorem~\ref{thm: BEC-large a}.

Throughout, $\nu(x) = \tfrac{1}{4}\ln\bigl(|x|^2/(|x|^2+1)\bigr)$, so that
\begin{equation}\label{eq:mu-QD-explicit}
	\mu_x^\mathsf{QD} = \sinh^2(\nu(x)) = \frac{\Bigl(\sqrt{\frac{|x|^2}{1+|x|^2}}-1\Bigr)^2}{4 \sqrt{\frac{|x|^2}{1+|x|^2}}} = \frac{1}{4}\left(\sqrt{\frac{|x|^2}{1+|x|^2}} + \sqrt{\frac{|x|^2+1}{|x|^2}} - 2\right),
\end{equation}
which depends on $x$ only through $|x|$.

\begin{proof}[Proof of Lemma~\ref{lem:QD-integrals}]
	Writing $\mu_x^\mathsf{QD} = g(|x|)$ and passing to polar coordinates, we have
	\[
		\int_{\bb{R}^3_0} \mu_x^\mathsf{QD}\dd{x} = 4\pi\int_0^{+\infty} r^2 g(r)\dd{r},\qquad
		\int_{\bb{R}^3_0} \mu_x^\mathsf{QD}\bigl(1+\mu_x^\mathsf{QD}\bigr)\dd{x} = 4\pi\int_0^{+\infty} r^2 g(r)\bigl(1+g(r)\bigr)\dd{r}.
	\]
	To handle the factors $1+r^2$ in \eqref{eq:mu-QD-explicit}, we substitute $r=\sinh(t)$ with $t\ge 0$, so that $\dd{r} = \cosh(t)\dd{t}$ and $\sqrt{r^2/(1+r^2)} = \tanh(t)$. Using that $(\tanh(t)-1)^2 = (\sinh(t)-\cosh(t))^2/\cosh^2(t) = e^{-2t}/\cosh^2(t)$, we obtain
	\begin{equation}\label{eq:f-explicit}
		f(t) := g(\sinh(t)) = \frac{(\tanh(t)-1)^2}{4\tanh(t)} = \frac{e^{-2t}}{4\sinh(t)\cosh(t)}.
	\end{equation}

	\noindent\emph{The first integral.} By \eqref{eq:f-explicit} we have $\sinh^2(t)\cosh(t) f(t) = \tfrac{1}{4}\sinh(t)e^{-2t}$, and hence
	\[
	\begin{aligned}
		\int_{\bb{R}^3_0} \mu_x^\mathsf{QD}\dd{x} &= 4\pi \int_0^{+\infty} \sinh^2(t)\cosh(t) f(t) \dd{t} \\
		&= \pi \int_0^{+\infty}\sinh(t)\, e^{-2t}\dd{t}
		= \frac{\pi}{2}\int_0^{+\infty}\bigl(e^{-t}-e^{-3t}\bigr)\dd{t},
	\end{aligned}
	\]
	which equals $\tfrac{\pi}{2}\bigl(1-\tfrac{1}{3}\bigr) = \tfrac{\pi}{3}$, as asserted.

	\medskip\noindent\emph{The second integral.} Splitting off the linear part and using the first integral, we have
	\[
		\int_{\bb{R}^3_0} \mu_x^\mathsf{QD}\bigl(1+\mu_x^\mathsf{QD}\bigr)\dd{x}
		= \frac{\pi}{3} + 4\pi\int_0^{+\infty}\sinh^2(t)\cosh(t) f(t)^2\dd{t}.
	\]
	By \eqref{eq:f-explicit} once more, $\sinh^2(t)\cosh(t)f(t)^2 = \tfrac{1}{4}\sinh(t)e^{-2t}f(t) = \tfrac{1}{16}e^{-4t}/\cosh(t)$, so that the remaining term equals
	\[
		4\pi\int_0^{+\infty}\sinh^2(t)\cosh(t) f(t)^2\dd{t} = \frac{\pi}{4}\int_0^{+\infty}\frac{e^{-4t}}{\cosh(t)}\dd{t}.
	\]
	Substituting $z = e^{-t}$, for which $\dd{t} = -\dd{z}/z$ and $\cosh(t) = (z+z^{-1})/2$, this last integral becomes
	\[
		\int_0^{+\infty}\frac{e^{-4t}}{\cosh(t)}\dd{t} = 2\int_0^1 \frac{z^4}{1+z^2}\dd{z}
		= 2\int_0^1\left(z^2 - 1 + \frac{1}{1+z^2}\right)\dd{z}
		= \frac{\pi}{2}-\frac{4}{3}.
	\]
	Combining the last three displays, we conclude that
	\[
		\int_{\bb{R}^3_0} \mu_x^\mathsf{QD}\bigl(1+\mu_x^\mathsf{QD}\bigr)\dd{x} = \frac{\pi}{3} + \frac{\pi}{4}\left(\frac{\pi}{2}-\frac{4}{3}\right) = \frac{\pi^2}{8},
	\]
	which completes the proof.
\end{proof}

\bibliographystyle{plain}
\bibliography{references}

@article{Araki,
  author    = {H. Araki},
  title     = {On quasifree states of the canonical commutation relations},
  journal   = {Publ.\ Res.\ Inst.\ Math.\ Sci.},
  volume    = {7},
  year      = {1971},
  pages     = {1--66},
}

@article{BBR,
  author    = {N. Behrmann and C. Brennecke and S. Rademacher},
  title     = {Exponential control of excitations for trapped {B}ose--{E}instein condensates in the {G}ross--{P}itaevskii regime},
  journal   = {Lett.\ Math.\ Phys.},
  volume    = {115},
  year      = {2025},
  pages     = {91},
}

@article{BBCS2018,
  author    = {C. Boccato and C. Brennecke and S. Cenatiempo and B. Schlein},
  title     = {Complete {B}ose--{E}instein condensation in the {G}ross--{P}itaevskii regime},
  journal   = {Commun.\ Math.\ Phys.},
  volume    = {359},
  year      = {2018},
  pages     = {975--1026},
}

@article{BBCS2019,
  author    = {C. Boccato and C. Brennecke and S. Cenatiempo and B. Schlein},
  title     = {{B}ogoliubov theory in the {G}ross--{P}itaevskii limit},
  journal   = {Acta Math.},
  volume    = {222},
  year      = {2019},
  pages     = {219--335},
}

@article{Bogoliubov,
  author    = {N. N. Bogoliubov},
  title     = {On the theory of superfluidity},
  journal   = {J.\ Phys.\ (USSR)},
  volume    = {11},
  year      = {1947},
  pages     = {23--32},
}

@book{BratelliRobinson,
  author    = {O. Bratteli and D. W. Robinson},
  title     = {Operator Algebras and Quantum Statistical Mechanics {II}},
  publisher = {Springer},
  address   = {Berlin},
  edition   = {2nd},
  year      = {1997},
}

@article{Bravyi,
  author    = {S. Bravyi and D. Gosset and R. K{\"o}nig},
  title     = {Quantum advantage with shallow circuits},
  journal   = {Science},
  volume    = {362},
  year      = {2018},
  pages     = {308--311},
}

@article{BSS1,
  author    = {C. Brennecke and B. Schlein and S. Schraven},
  title     = {{B}ose--{E}instein condensation with optimal rate for trapped bosons in the {G}ross--{P}itaevskii regime},
  journal   = {Math.\ Phys.\ Anal.\ Geom.},
  volume    = {25},
  year      = {2022},
  pages     = {12},
}

@article{BSS2,
  author    = {C. Brennecke and B. Schlein and S. Schraven},
  title     = {{B}ogoliubov theory for trapped bosons in the {G}ross--{P}itaevskii regime},
  journal   = {Ann.\ Henri Poincar{\'e}},
  volume    = {23},
  year      = {2022},
  pages     = {1583--1658},
}

@article{Caves,
  author    = {C. M. Caves and B. L. Schumaker},
  title     = {New formalism for two-photon quantum optics. {I}. {Q}uadrature phases and squeezed states},
  journal   = {Phys.\ Rev.\ A},
  volume    = {31},
  year      = {1985},
  pages     = {3068--3092},
}

@article{Coreggi1,
  author    = {M. Correggi and F. Pinsker and N. Rougerie and J. Yngvason},
  title     = {Critical rotational speeds for superfluids in a homogeneous trap},
  journal   = {Journal of Math. Phys.},
  volume    = {53 (9)},
  year      = {2012},
  pages     = {095203},
}

@article{Coreggi2,
  author    = {M. Correggi and N. Rougerie and J. Yngvason},
  title     = {The Transition to a Giant Vortex Phase in a fast rotation {B}ose--{E}instein condensate},
  journal   = {Comm. Math. Phys.},
  volume    = {303},
  year      = {2011},
  pages     = {2409--2439},
}

@article{Dimonte,
  author    = {D. Dimonte and E.L. Giacomelli},
  title     = {On {B}ose--{E}instein condensates in the {T}homas--{F}ermi regime},
  journal   = {Math. Phys. Anal. Geom. },
  volume    = {25},
  year      = {2022},
  pages     = {25},
}

@article{Gasperini,
  author    = {M. Gasperini},
  title     = {{B}ogoliubov transformations and negative binomial distributions for particle creation by gravitational fields},
  journal   = {Prog.\ Theor.\ Phys.},
  volume    = {84},
  year      = {1990},
  pages     = {899--907},
}

@book{Haag,
  author    = {R. Haag},
  title     = {Local Quantum Physics: Fields, Particles, Algebras},
  publisher = {Springer},
  address   = {Berlin},
  edition   = {2nd},
  year      = {1996},
}

@article{H2021,
  author    = {C. Hainzl},
  title     = {Another proof of {B}ose--{E}instein condensation in the {G}ross--{P}itaevskii limit},
  journal   = {J.\ Math.\ Phys.},
  volume    = {62},
  year      = {2021},
  pages     = {051101},
}

@article{HST2022,
  author    = {C. Hainzl and B. Schlein and A. Triay},
  title     = {{B}ogoliubov theory in the {G}ross--{P}itaevskii limit: {A} simplified approach},
  journal   = {Forum Math.\ Sigma},
  volume    = {10},
  year      = {2022},
  pages     = {e10},
}

@article{Idziaszek,
  author    = {Z. Idziaszek and M. Gajda and P. Navez and M. Wilkens and K. Rz{\k{a}}{\.z}ewski},
  title     = {Fluctuations of the weakly interacting {B}ose--{E}instein condensate},
  journal   = {Phys.\ Rev.\ Lett.},
  volume    = {82},
  year      = {1999},
  pages     = {4376--4379},
}

@article{Janson,
  author    = {S. Janson},
  title     = {Tail bounds for sums of geometric and exponential variables},
  journal   = {Statist.\ Probab.\ Lett.},
  volume    = {135},
  year      = {2018},
  pages     = {1--6},
}

@article{KocharovskyCNBIII,
  author    = {V. V. Kocharovsky and Vl. V. Kocharovsky and M. O. Scully},
  title     = {Condensation of $N$ bosons. {III}. {A}nalytical results for all higher moments of condensate fluctuations},
  journal   = {Phys.\ Rev.\ A},
  volume    = {61},
  year      = {2000},
  pages     = {053606},
}

@article{KocharovskyBosonSampling,
  author    = {V. V. Kocharovsky and Vl. V. Kocharovsky and W. D. Shannon and S. V. Tarasov},
  title     = {Towards the simplest model of quantum supremacy: {A}tomic boson sampling in a box trap},
  journal   = {Entropy},
  volume    = {25},
  year      = {2023},
  pages     = {1584},
}

@article{Kristensen,
  author    = {M. A. Kristensen and others},
  title     = {Observation of atom number fluctuations in a {B}ose--{E}instein condensate},
  journal   = {Phys.\ Rev.\ Lett.},
  volume    = {122},
  year      = {2019},
  pages     = {163601},
}

@article{LiebSeiringer2002,
  author    = {E. H. Lieb and R. Seiringer},
  title     = {Proof of {B}ose--{E}instein condensation for dilute trapped gases},
  journal   = {Phys.\ Rev.\ Lett.},
  volume    = {88},
  year      = {2002},
  pages     = {170409},
}

@book{LiebSeiringer,
  author    = {E. H. Lieb and R. Seiringer},
  title     = {The Stability of Matter in Quantum Mechanics},
  publisher = {Cambridge University Press},
  address   = {Cambridge},
  year      = {2010},
}

@article{LS2006,
  author    = {E. H. Lieb and R. Seiringer},
  title     = {Derivation of the {G}ross--{P}itaevskii equation for rotating {B}ose gases},
  journal   = {Commun.\ Math.\ Phys.},
  volume    = {264},
  year      = {2006},
  pages     = {505--537},
}

@article{NNRT,
  author    = {P. T. Nam and M. Napi{\'o}rkowski and J. Ricaud and A. Triay},
  title     = {Optimal rate of condensation for trapped bosons in the {G}ross--{P}itaevskii regime},
  journal   = {Anal.\ PDE},
  volume    = {15},
  year      = {2022},
  pages     = {1585--1616},
}

@unpublished{NR2024,
  author    = {P. T. Nam and S. Rademacher},
  title     = {Exponential bounds of the condensation for dilute {B}ose gases},
  note      = {Trans.\ Amer.\ Math.\ Soc., in press (arXiv:2307.10622)},
  year      = {2024},
}

@article{NRS,
  author    = {P. T. Nam and N. Rougerie and R. Seiringer},
  title     = {Ground states of large bosonic systems: {T}he {G}ross--{P}itaevskii limit revisited},
  journal   = {Anal.\ PDE},
  volume    = {9},
  year      = {2016},
  pages     = {459--485},
}

@unpublished{NT,
  author    = {P. T. Nam and A. Triay},
  title     = {{B}ogoliubov excitation spectrum of trapped {B}ose gases in the {G}ross--{P}itaevskii regime},
  note      = {J.\ Math.\ Pures Appl., to appear (arXiv:2106.11949)},
  year      = {2021},
}

@article{Quesada,
  author    = {N. Quesada and F. Arzani and N. Killoran},
  title     = {Gaussian boson sampling using squeezed states},
  journal   = {Phys.\ Rev.\ A},
  volume    = {100},
  year      = {2019},
  pages     = {022341},
}

@article{Rademacher,
  author    = {S. Rademacher},
  title     = {Generating function for quantum depletion of {B}ose--{E}instein condensates},
  journal   = {J.\ Stat.\ Phys.},
  volume    = {192},
  year      = {2025},
  pages     = {108},
}

@incollection{Tanas,
  author    = {R. Tana{\'s} and A. Miranowicz and T. Gantsog},
  title     = {Quantum phase properties of nonlinear optical phenomena},
  booktitle = {Progress in Optics},
  volume    = {35},
  editor    = {E. Wolf},
  publisher = {Elsevier},
  year      = {1996},
  pages     = {355--446},
}

@article{Vibel,
  author    = {T. Vibel and others},
  title     = {Atom number fluctuations in {B}ose gases: {S}tatistical analysis of parameter estimation},
  journal   = {J.\ Phys.\ B},
  volume    = {57},
  year      = {2024},
  pages     = {195301},
}

@book{Wald,
  author    = {R. M. Wald},
  title     = {Quantum Field Theory in Curved Spacetime and Black Hole Thermodynamics},
  publisher = {University of Chicago Press},
  address   = {Chicago},
  year      = {1994},
}

@article{Weedbrook,
  author    = {C. Weedbrook and S. Pirandola and R. Garc{\'i}a-Patr{\'o}n and N. J. Cerf and T. C. Ralph and J. H. Shapiro and S. Lloyd},
  title     = {Gaussian quantum information},
  journal   = {Rev.\ Mod.\ Phys.},
  volume    = {84},
  year      = {2012},
  pages     = {621--669},
}

@article{beane_n-boson_2007,
    title = {n-{Boson} {Energies} at {Finite} {Volume} and {Three}-{Boson} {Interactions}},
    volume = {76},
    doi = {10.1103/PhysRevD.76.074507},
    journal = {Phys. Rev. D},
    author = {Beane, Silas R. and Detmold, William and Savage, Martin J.},
    year = {2007},
    pages = {074507},
}

@book{Perelomov,
  author    = {A. Perelomov},
  title     = {Generalized Coherent States and Their Applications},
  series    = {Texts and Monographs in Physics},
  publisher = {Springer},
  address   = {Berlin, Heidelberg},
  year      = {1986},
  doi       = {10.1007/978-3-642-61629-7},
}

@unpublished{Deuchert,
  author    = {A. Deuchert and Xuanyu Li},
  title     = {Large Deviations for the Bose--Einstein Condensate of the Ideal Gas in the Canonical Ensemble},
  note      = {Preprint:arXiv:2608.26378},
  year      = {2026},
}

@unpublished{Deuchert2,
  author    = {A. Deuchert and P. T. Nam and M. Napiórkowski },
  title     = {The Gibbs state of the mean-field Bose gas},
  note      = {Preprint:arXiv:2501.19396},
  year      = {2025},
}

\end{document}